\documentclass[11pt]{article}
\usepackage{rotating,amsmath,amssymb,amsfonts,amsthm}
\usepackage{epsfig,mathrsfs,natbib,color,graphics,comment}
\usepackage{graphicx}
\usepackage{array,wrapfig}
\usepackage{hyperref}
\usepackage{booktabs}
\usepackage{subcaption}

\def\boxit#1{\vbox{\hrule\hbox{\vrule\kern6pt \vbox{\kern6pt#1\kern5pt}
\kern6pt\vrule}\hrule}}

\newcommand{\bX}{{\boldsymbol X}}
\newcommand{\bY}{{\boldsymbol Y}}
\newcommand{\bZ}{{\boldsymbol Z}}

\newcommand{\bx}{{\boldsymbol x}}
\newcommand{\bz}{{\boldsymbol z}}

\newcommand{\bi}{{\boldsymbol i}}
\newcommand{\bc}{{\boldsymbol c}}

\newcommand{\bv}{{\boldsymbol v}}

\newcommand{\mR}{\mathbb{R}}

\newcommand{\BX}{\mathbb{X}}

\newcommand{\mG}{{\cal G}}

\newcommand{\mK}{{\cal K}}

\newcommand{\mX}{{\cal X}}

\newcommand{\bxi}{{\boldsymbol \xi}}
\newcommand{\bbeta}{{\boldsymbol \beta}}
\newcommand{\btheta}{{\boldsymbol \theta}}

\newcommand{\bvartheta}{{\boldsymbol \vartheta}}

\newcommand{\bvarepsilon}{{\boldsymbol \varepsilon}}
\newcommand{\bomega}{{\boldsymbol \omega}}

\newcommand{\var}{{\mbox{Var}}}

\newcommand{\diag}{{\mbox{diag}}}
\newtheorem{theorem}{Theorem}[]
\newtheorem{lemma}{Lemma}[]

\newtheorem{prop}{Proposition}[]
\newtheorem{corollary}{Corollary}[]

\newtheorem{algorithm}{Algorithm}[]
\newtheorem{remark}{Remark}[]
\begin{document}

\title{A Langevinized Ensemble Kalman Filter for Large-Scale Static and Dynamic Learning}

\author{Peiyi Zhang, Qifan Song, and Faming Liang
\thanks{To whom correspondence should be addressed: Faming Liang.
  P. Zhang is Graduate Student (email: zhan2763@purdue.edu),  Q. Song is Assistant Professor (email: qfsong@purdue.edu),
  and F. Liang is Professor (email: fmliang@purdue.edu), Department of Statistics, Purdue University, 
  West Lafayette, IN 47907.
 }
 }


\maketitle

\begin{abstract}

The Ensemble Kalman Filter (EnKF) has achieved great successes in data assimilation in  atmospheric and oceanic sciences, but its failure in convergence to the right filtering distribution precludes its use for uncertainty quantification. We reformulate the EnKF under the framework of Langevin dynamics, which leads to a new particle filtering algorithm, the so-called Langevinized EnKF. The Langevinized EnKF inherits the forecast-analysis procedure from the EnKF and the use of mini-batch data from the stochastic gradient Langevin-type algorithms, which make it scalable with respect to both the dimension and sample size. We prove that the Langevinized EnKF converges to the right filtering distribution in Wasserstein distance
under the big data scenario that the dynamic system consists of a large number of stages and has a large number of samples observed at each stage. We reformulate the Bayesian inverse problem as a dynamic state estimation problem based on the techniques of subsampling and Langevin diffusion process. We illustrate the  performance of the Langevinized EnKF using a variety of examples, including the Lorenz-96 model, high-dimensional  variable selection, Bayesian deep learning, and Long Short Term Memory (LSTM) network learning with dynamic data. 
 
\vspace{2mm}


\end{abstract}

\section{Introduction}

Coming with the new century, the integration of computer technology into science and daily life has enabled scientists to collect massive volumes of data, such as satellite data, high-throughput biological assay data and website transaction logs. To address computational difficulty encountered in Bayesian analysis of big data, a variety of scalable MCMC 
algorithms have been developed, which include 
stochastic gradient MCMC algorithms 
\citep{Welling2011BayesianLV, Ding2014BayesianSU, Ahn2012BayesianPS, Chen2014StochasticGH, Betancourt2015, Li2016PSGLD, Ma2015ACR,NemethF2019},
split-and-merge algorithms \citep{ConsensusMC2016, LiSD2017, Srivastava2018, XueLiang2019},
mini-batch Metropolis-Hastings algorithms \citep{chen2016min, korattikara2014austerity, bardenet2014towards, Maclaurin2014FireflyMC, Bardenet2017OnMC}, 
 nonreversible Markov process-based algorithms \citep{BierkensFR2016, BouchardVD2016}, and 
 Bayesian bootstrapping \citep{LiangKS2016,FongLH2019}, among others.

 Although the scalable MCMC algorithms have achieved great successes in Bayesian learning with static data, none of them could be directly applied to dynamic data.
 In the literature, learning with static data is often termed as static or off-line learning, and learning with dynamic data is often termed as dynamic or on-line learning. 
 Dynamic learning is important and challenging, as dynamic data collection is general, heterogeneous and messy.
 We note that classical sequential Monte Carlo or particle filter algorithms, see e.g. \cite{DoucetFG2001}, 
 lack the scalability necessary for dealing with large-scale dynamic data, which strive to make use of all available data at each processing step.
 The ensemble Kalman filter (EnKF) \citep{Evensen1994} 
 has achieved great successes in high-dimensional data assimilation problems, see e.g., \cite{EvensenVan1996}, \cite{Aanonsen2009} and \cite{HoutekamerM2001}, 
 but it fails to converge to the right filtering 
 distribution for general nonlinear dynamic systems \citep{LawTt2016}. 
 How to make Bayesian on-line learning with large-scale dynamic data has posed a great challenge on current statistical methods.
 
 In this paper, we are interested in conducting Bayesian on-line learning for the dynamic system:
 \begin{equation} \label{EnKFasimeq}
 \begin{split}
   x_{t}& =g(x_{t-1})+u_{t}, \quad u_{t}\sim N(0, U_{t}), \\
   y_t & =H_t x_t+\eta_t, \quad \eta_{t} \sim N(0, \Gamma_t),   \\
 \end{split}
  \end{equation}
 for $t=1,2,\ldots, T$, where $x_t \in \mathbb{R}^p$ and $y_{t}\in \mathbb{R}^{N_t}$ denote, respectively, the state and observations at stage $t$; and the dimension $p$, the number of stages $T$, and the sample sizes $N_t$'t are all assumed to be very large. 
 For the dynamic system,
 the top equation is called the state evolution equation, where  $g(\cdot)$ is called state propagator and can be nonlinear;
 and the bottom equation is called the measurement equation, where the propagator $H_t$ relates the state variable to the measurement variable and yields the expected value of the prediction given the model states and parameters. 
 The dynamic system (\ref{EnKFasimeq}) is of central importance in this paper. It itself models the data 
 assimilation problems with linear measurement 
 equations (studied in Section 3.3), the data assimilation problems with nonlinear measurement equations (studied in Section 3.4) and the 
 inverse problems (studied in Sections 3.1 and 3.2) 
 can be converted to it via appropriate transformations. 
 For simplicity, we assume that  
 both the model error $u_t$ and observation error $\eta_t$ are zero-mean Gaussian random variables, and 
  that the covariance matrices $U_t$ and $\Gamma_t$ and the propagator $g(\cdot)$ and $H_t$ are all fully specified, i.e. containing no unknown parameters. How to extend our study to the problems with non-Gaussian random error and/or unknown parameters will be discussed at the end of the paper. Throughout this paper, we will let $t$ index the stage of the dynamic system, 
  let $f(y_t|x_t)$ denote the likelihood function  of $y_t$, let $\pi(x_t|y_{1:t})$ denote the filtering distribution at stage $t$ given the data $\{y_1,y_2,\ldots,y_t\}$, and let
  $\pi(x_{t}|y_{1:t-1})=\int \pi(x_{t}|x_{t-1}) \pi(x_{t-1}|y_{1:t-1}) d x_{t-1}$ 
  denote the predictive distribution of $x_{t}$ given $\{y_1,y_2,\ldots,y_{t-1}\}$.

  Our goal is to develop a scalable particle filtering algorithm under the big data scenario that the dimension $p$, the numbers of stages $T$ and the sample sizes $N_t$'s are all very large. Toward this goal, we reformulate the EnKF under the framework of Langevin dynamics. The resulting  algorithm inherits the forecast-analysis procedure from the EnKF and the use of mini-batch data from the stochastic gradient Langevin-type algorithms. The former makes the new algorithm scalable with respect to the dimension, and the latter makes it scalable with respect to the sample size. To credit to its two precursors, the new algorithm is coined as  Langevinized EnKF. We prove that the Langevinized EnkF 
 converges to the right filtering distributions in Wasserstein distance under the scenario that 
  the number of stages $T$ and the sample sizes $N_t$'s
  are all large. 
  We illustrate the performance of the Langevinized EnKF using a variety of examples, including the Lorenz-96 model \citep{Lorenz1996}, Bayesian deep learning, and Long Short Term Memory (LSTM) network learning. 
 
 The remaining part of this paper is organized as follows. Section 2 provides a brief review for the EnKF and explains its scalability with respect to the dimension.
 Section 3 describes the Langevinized EnKF for both statistic learning and dynamic learning and 
 studies its convergence, where how the inverse problems and the data assimilation problems with a nonlinear measurement equation can be converted to the system (\ref{EnKFasimeq}) is detailed. 
 Sections 4 and 5 illustrate the performance of the Langevinized EnKF for the inverse and data assimilation problems, respectively. Section 6 concludes the paper with a brief discussion.

\section{Why is the EnKF Efficient for High-Dimensional Problems?} \label{EnKFsect}

 Consider the dynamic system (\ref{EnKFasimeq}). To estimate the state variables $x_1,x_2,\ldots,x_T$, where $T$ denotes the total number of stages, \cite{Evensen1994} proposed the EnKF algorithm: 

 \begin{algorithm} \label{EnKFAlg} (EnKF Algorithm)
 \begin{itemize}
 \item[0.] (Initialization) Start with an initial ensemble $x_0^{a,1}, x_0^{a,2}, \ldots, x_0^{a,m}$,
  where $m$ denotes the ensemble size. For each stage $t=1,2,\ldots,T$, do the
  following:

  \item[1.] (Forecast) For $i=1,2,\ldots, m$, draw $u_t^i \sim N(0, U_t)$ and calculate
   \[
    x_t^{f,i}=g(x_{t-1}^{a,i})+u_t^i. 
   \]
   Calculate the sample covariance matrix of $x_t^{f,1},\ldots,x_t^{f,m}$ and denote it by $C_t$.

  \item[2.] (Analysis) For $i=1,2,\ldots,m$, draw $\eta_t^i \sim N(0, \Gamma_t)$ and calculate
   \[
    x_t^{a,i} = x_t^{f,i} +\hat{K}_t (y_t-H_t x_t^{f,i}-\eta_t^i) \stackrel{\Delta}{=}x_t^{f,i} +\hat{K}_t (y_t-y_{t}^{f,i}),
   \]
   where $\hat{K}_t=C_t H_t^T (H_t C_t H_t^T+\Gamma_t)^{-1}$ forms an estimator for
   the Kalman gain matrix $K_t=S_t H_t^T$ $ (H_t S_t H_t^T+\Gamma_t)^{-1}$ and
   $S_t$ denotes the covariance matrix of $x_t^f$.
 \end{itemize}
\end{algorithm}

 If $U_t$, $\Gamma_t$, $g(\cdot)$ and $H_t$ contain some unknown parameters,
 the state augmentation method \citep{Anderson2001, Baeketal2006, GillijnsDeMoor2007} can be used, where the state vector is augmented with the model parameters and the EnKF estimate the state and parameters in a simultaneous manner.

 
 The rationale underlying the EnKF can be explained as follows. Let $x_t^f$ and $x_t^a$ denote a generic sample  obtained at 
 the forecast and analysis step, respectively. 
 The forecast step is to use the forecasted samples $\{x_t^{f,1},\ldots,x_t^{f,m}\}$ to approximate the predictive distribution $\pi(x_t|y_{1:t-1})$. 
 Let  $\mu_t^{\prime}$ and $S_t$ denote the mean and variance of $\pi(x_t|y_{1:t-1})$, respectively. Hence, one can rewrite $x_t^f$ as 
 \[
 x_t^f=\mu_t^{\prime}+w_t,
 \]
 where $w_t$ is a random error with mean 0 and variance $S_t$. If $\pi(x_t|y_{1:t-1})$ is Gaussian, 
   by the identity $K_t=S_tH_t^T(H_t S_t H_t^T+\Gamma_t)^{-1}=
  (I-K_tH_t)S_t H_t^T \Gamma_t^{-1}=
  (H_t^T \Gamma_t^{-1} H_t+S_t^{-1})^{-1} H_t^T \Gamma_t^{-1}$, one can show
 \[
 \begin{split}
   x_t^a &=\big[\mu_t^{\prime}+{K}_t(y_t-H_t\mu_t^{\prime})\big]+\big[(I-{K}_t H_t)w_t-{K}_t \eta_t\big] =\mu_t+e_t,\\
\end{split}
\]
where $\mu_t=\mu_t^{\prime}+{K}_t(y_t-H_t\mu_t^{\prime})$ is the mean of 
$\pi(x_t|y_{1:t}) \propto \pi(y_t|x_t) \pi(x_t|y_{1:t-1})$, and 
$e_t=(I-{K}_t H_t)w_t-{K}_t \eta_t$ is a Gaussian random error with mean 0 and variance  
$\var(e_t)=(I-{K}_t H_t)S_t$; that is, 
 $x_t^a$ is a sample following
 the filtering distribution $\pi(x_t|y_{1:t})$.  
 
The EnKF has two attractive features which make it 
extremely successful in dealing with 
high-dimensional data assimilation problems such as those encountered in 
 reservoir modeling \citep{Aanonsen2009},
 oceanography \citep{EvensenVan1996},  and
 weather forecasting \citep{HoutekamerM2001}.
  First, it approximates each filtering distribution $\pi(x_t|y_{1:t})$ using an ensemble of particles.  Since the ensemble size $m$ is typically much smaller than $p$,  it leads to dimension reduction and computational feasibility
 compared to the Kalman filter, see e.g. \cite{ShumwayStoffer2006}. In
 particular, it approximates $S_t$ by $C_t$, 
  and the storage for the matrix $C_t$ is replaced by 
  particles and thus much reduced. 
  Second, in generating particles from 
 each filtering distribution, it avoids covariance matrix decomposition  
 compared to conventional particle filters. It is known that an LU-decomposition of the covariance matrix has a computational complexity of $O(p^3)$. Instead, the EnKF employs a forecast-analysis procedure to generate particles, which has a computational complexity of 
 $O(\max\{p^2 N_t, N_t^3\}+mpN_t)$ for $m$ particles at stage $t$. That is, the forecast-analysis procedure reduces the computational complexity of the algorithm 
 when $m$ and $N_t$ are smaller than $p$. This explains why the EnKF is so efficient for high-dimensional problems.   

 Despite its great successes,
 the performance of the EnKF is sub-optimal. As shown by \cite{LawTt2016},
 it converges only to a mean-field filter, which provides the optimal linear estimator of the
 conditional mean but not the filtering distribution except in the large sample limit for linear systems. Similar results can be found in \cite{LeGlandMT2009}, \cite{BergouGM2019} and \cite{Kwiatkowaki2015}. 
 
 As an extension, \cite{IglesiasLS2013} applied the EnKF to solve the inverse problem, which is to find
  the parameter $z$ given observations of the form
 \begin{equation} \label{inverseeq}
  y=\mG(z)+\eta,
 \end{equation}
 where $\mG(\cdot)$ is the forward response operator mapping the unknown parameter
 $z$ to the  space of observations, $\eta \sim N(0,\Gamma)$ is Gaussian random noise,
  and $y$ is observed data.
 With the state augmentation approach, they defined the new state vector as
 $x^T=(z^T, \mG(z)^T)$ and an artificial dynamic system as
 \begin{equation} \label{inverseeq2}
 \begin{split}
  x_{t} &= x_{t-1}, \\
  y_{t} &= H x_{t}+\eta_{t}, \\
 \end{split}
 \end{equation}
 where $H=(0,I)$ such that $H X=\mG(z)$ holds, $y_{t} \equiv y$ for all $t=1,2,\ldots$, and $\eta_t \sim N(0,\Gamma)$.
 However, as mentioned previously, the EnKF does not converge to the filtering distribution, so 
 the posterior distribution $\pi(z|y)$ cannot be well approximated by the ensemble,
 and thus uncertainty of the estimate of $z$ cannot be correctly quantified. 
 Numerically, \cite{ErnstSS2015} demonstrated that for nonlinear inverse
 problems the large sample limit does not lead to a good approximation to the posterior distribution.

\section{Langevinized EnKF}

 To motivate the development of the Langevinized EnKF, we first consider a linear inverse problem and
  then extend it to nonlinear inverse and data assimilation problems. 
  
 \subsection{Linear Inverse Problem}  \label{Algsection1}
 
 Consider a Bayesian inverse problem for the linear regression
 \begin{equation} \label{inverseeq3} 
   y=Hx+\eta, 
 \end{equation}
 where $\eta \sim N(0,\Gamma)$ for some covariance matrix $\Gamma$, $y \in \mR^{N}$, and $x\in \mR^{p}$ is an unknown continuous parameter vector.
 To accommodate the case that $N$ is extremely large, we assume
 that $y$ can be partitioned into $B=N/n$ independent and identically distributed blocks $\{y_1,\ldots,y_B\}$,
 where each block is of size $n$ and has the covariance matrix $V$ such that
 $\Gamma=\mbox{diag}[V,\cdots,V]$.

Let $\pi(x)$ denote the prior density function of $x$, which is assumed to be differentiable with respect to $x$. Let $\pi(x|y)$ denote the posterior distribution.
To develop an efficient algorithm for simulating from $\pi(x|y)$, which is scalable with respect to both the sample size $N$ and the dimension $p$,  
 we reformulate the model 
 (\ref{inverseeq3}) as a state-space model through subsampling and Langevin diffusion: 
 \begin{equation} \label{inveq1}
 \begin{split}
  x_{t}&=x_{t-1}+\epsilon_t \frac{n}{2N} \nabla \log\pi(x_{t-1})+w_t,  \\
  y_t & =H_t x_t+v_t, \\
 \end{split} 
 \end{equation}
 where $w_t \sim N(0,\frac{n}{N}\epsilon_t I_p)=N(0, \frac{n}{N}Q_t)$, i.e., $Q_t=\epsilon_t I_p$,
 $y_t$ denotes a block randomly drawn from $\{y_1,\ldots,y_B\}$, 
 $v_t \sim N(0,V_t)$ with $V_t=V$, and
 $H_t$ is a submatrix of $H$ extracted with the corresponding $y_t$. 
 In the state-space model, at each stage $t$, the state (i.e., the parameters of model (\ref{inverseeq3})) 
 evolves according to an Euler-discretized Langevin equation of the prior distribution, and the measurement varies with subsampling. As shown in Theorem S1 of the Supplementary Material, the filtering distribution of the state-space model
 converges to the target posterior 
 $\pi(x|y)$ as $t \to \infty$, provided that 
 $\epsilon_t$ decays to zero in an appropriate rate and 
 the matrix $V$ satisfies some regularity conditions. 
 To simulate from dynamic system (\ref{inveq1}), we propose Algorithm \ref{EnKFnew}, which makes use of both techniques, subsampling and the forecast-analysis procedure and is thus scalable with respect to both the sample size $N$ and the dimension $p$.

 \begin{algorithm} \label{EnKFnew} (Langevinized EnKF for Linear Inverse Problems)
 \begin{itemize}
 \item[0.] (Initialization) Start with an initial ensemble $x_0^{a,1}, x_0^{a,2}, \ldots, x_0^{a,m}$,
  where $m$ denotes the ensemble size. For each stage $t=1,2,\ldots, T$ (first loop), do the following: 

  \item[1.] (Subsampling) Draw without replacement a mini-batch data, 
  denoted by $(y_t,H_t)$, of size $n$ from 
  the full dataset of size $N$. 

  \item Set $Q_t=\epsilon_t I_p$, $R_t=2V_t$, and the Kalman gain matrix $K_t=Q_t H_t^T (H_t Q_t H_t^T+R_t)^{-1}$. 
   For each chain $i=1,2,\ldots,m$ (second loop), do steps 2-3:

  \item[2.] (Forecast) Draw $w_t^i \sim N_{p}(0, \frac{n}{N} Q_t)$ and calculate
   \begin{equation} \label{Algeq1}
    x_t^{f,i}=x_{t-1}^{a,i}+\epsilon_t \frac{n}{2N} \nabla \log\pi(x_{t-1}^{a,i})+w_t^i. 
   \end{equation}

  \item[3.] (Analysis) Draw $v_t^i \sim N_{n}(0, \frac{n}{N} R_t)$ and calculate
   \begin{equation} \label{Algeq2}
    x_t^{a,i} = x_t^{f,i}+K_t (y_t-H_t x_t^{f,i}-v_t^i) \stackrel{\Delta}{=}x_t^{f,i} +K_t (y_t-y_{t}^{f,i}).
   \end{equation}
 \end{itemize}
\end{algorithm}

\begin{remark} \label{remark1}
 The Langevinized EnKF is different from the existing formulation 
 of EnKF for inverse problems \citep{IglesiasLS2013} in three aspects: (i) 
  It reformulates the Bayesian inverse problem as a state-space 
  model, where the state (i.e., parameters) evolves according to a 
  Langevin diffusion process converging 
  to the prior $\pi(x)$ and the measurement varies with subsampling; the subsampling technique 
  enables the algorithm scalable with respect to the sample size $N$; (ii) the measurement noise is drawn
  from a variance inflated distribution $N(0, 2V_t)$ in the analysis step;  and (iii) $Q_t=\epsilon_t I_p$ is a designed diagonal matrix with the learning rate $\epsilon_t=O(t^{-\varpi})$ for some $0<\varpi<1$. 
\end{remark}
 
The convergence of Algorithm \ref{EnKFnew} is established in Theorem S1.  An informal restatement of the theorem is given 
in Proposition \ref{prop1}, which facilitates 
discussions for the property of the Langevinized EnKF algorithm.

 \begin{prop} \label{prop1} (Convergence of Langevinized EnKF) 
  Let $x_t^a$ denote a generic member of the ensemble produced by Algorithm \ref{EnKFnew} in the analysis step of stage $t$.
  If the eigenvalues of $\Sigma_t=\frac{n}{N}(I-K_t H_t)$ are uniformly bounded with respect to $t$,
  $\log\pi(x)$ is differentiable with respect to $x$, and the learning rate $\epsilon_t=O(t^{-\varpi})$ for some
  $0<\varpi <1$,
 then $\lim_{t\to \infty} W_2(\tilde{\pi}_t,\pi_*)=0$, where $\tilde{\pi}_t$ denotes the empirical distribution of $x_t^a$, $\pi_*=\pi(x|y)$ denotes the target posterior distribution, and $W_2(\cdot,\cdot)$ denotes the second-order Wasserstein distance.
 \end{prop}

 In the proof of Theorem S1, it is shown that 
 \begin{equation} \label{SGRLDeq}
  x_t^a  =x_{t-1}^a+ \frac{\epsilon_t}{2} \Sigma_t
  \widehat{\nabla} \log \pi(x_{t-1}^a|y_t)
  +e_t, 
 \end{equation}
 where $e_t$ is a zero mean Gaussian random error with covariance $\var(e_t)=\epsilon_t \Sigma_t$, and   
 $\widehat{\nabla} \log \pi(x_{t-1}^a|y_t) = 
 \frac{N}{n}H_t^T V_t^{-1}(y_t-H_t {x}_{t-1}^a) + \nabla \log \pi(x_{t-1}^a)$ denotes an unbiased
 estimate of $\nabla \log\pi(x_{t-1}^a|y_t)$. That is, the Langevinized EnKF forms a new type of 
 stochastic gradient Riemannian Langevin dynamics (SGRLD) algorithm \citep{PattersonTeh2013, Ahn2012BayesianPS, GirolamiG2011},
 where the Fisher information matrix is adapted with the mini-batch of data by noting that 
 $\epsilon_t \Sigma_t =\frac{n}{N} (I-K_t H_t)Q_t$ is exactly the inverse of the Fisher information matrix of the distribution $\pi(x_t^a|x_{t-1}^a, y_t)$.
 It is known that use of the Fisher information, 
 which rescales parameter updates according to the geometry of the target distribution, can generally improve the convergence of SGMCMC especially when the target distribution exhibits pathological curvature and contains some saddle points \citep{Li2016PSGLD,Dauphin2014}. 

Since we set the learning rate $\epsilon_t=O(t^{-\varpi})$ for some $0<\varpi<1$ in Algorithm \ref{EnKFnew}, by Theorem 2 of \cite{SongLiang2020}, $\{x_t^{a}: t=t_0+1,\ldots,T\}$ can be treated as equally weighted samples, where $t_0$ represents the burn-in period. 
  That is, for any Lipschitz function $\rho(x)$, the posterior mean 
  $E_{\pi} \rho(x)$ $=\int \rho(x) \pi(x|y)$
  can be estimated by 
  \begin{equation} \label{esteq} 
  \widehat{E_{\pi} \rho(x)}= \frac{1}{(T-t_0)m} \sum_{t=t_0+1}^T \sum_{i=1}^m \rho(x_t^{a,i}),
  \end{equation} 
  which converges to $E_{\pi} \rho(x)$ in probability as $T \to \infty$. 
 Alternatively, we can apply a weighted 
  averaging scheme as suggested by \cite{chen2015convergence} and \cite{Teh2016SGLD} for estimating  $E_{\pi} \rho(x)$. As for a conventional SGLD algorithm, we can also set $\epsilon_t$ to a small constant. In this case, the convergence of 
  $x_t^a$ to the posterior distribution is up to an approximation error even when $t\to \infty$.

   It is interesting to point out that when the dimension of $x$ is high, Langevinized EnKF  can be much more efficient than directly implementing (\ref{SGRLDeq}). The latter requires an LU-decomposition of $\Sigma_t$, which has a computational complexity of $O(p^3)$, in generating the random error $e_t$. 
  While the Langevinized EnKF gets around this issue with the forecast-analysis procedure. 
  As shown in the proof of Theorem S1, we have $e_t=(I-K_t H_t)w_t-K_t v_t$ and $\var(e_t)=\epsilon_t \Sigma_t$. 
  The computational complexity of the forecast-analysis procedure is $O(\max\{n^2 p, n^3\}+mnp)$ for generating $m$ particles per iteration, where the 
  first term represents the cost  
  for calculating $K_t$, the second term represents the total cost of $m$ chains for forecasting and analysis, and the cost for calculating $K_t$ is counted 
  as the overhead at each iteration.  
  This procedure is even faster than in the original EnKF algorithm as $Q_t$ is diagonal. For high-dimensional problems, we typically set 
    $n \ll p$, so the total computational complexity of the Langevinized EnKF is $O((n^2p+mnp)T)$, which implies that the algorithm is scalable with respect to both the sample size $N$ and the dimension $p$. In contrast, if (\ref{SGRLDeq}) is directly simulated as a SGRLD algorithm, the total computational complexity will be $O((p^3+np^2)mT)$ for $mT$ iterations, where 
    $O(np^2)$ represents the cost for computing $\Sigma_t \widehat{\nabla} \log\pi(x_{t-1}^a|y_t)$.
    Here we note that the particles in the same ensemble of the Langevinized EnKF are not independent, as they are generated based on the same randomly selected mini-batch of data at each stage. However, they are independent if the full data is used at each iteration.

  Finally, we note that conventional SGRLD  algorithms lack the scalability necessary for high-dimensional problems as computation of the Fisher information matrix can be very costly. For this reason, preconditioned SGLD \citep{Li2016PSGLD} approximates the Fisher information matrix using a diagonal matrix estimated based on the current gradient information only.
  The Langevinized EnKF forms a new type of SGRLD algorithm,
  where the forecast-analysis procedure enables 
  the Fisher information efficiently used in the simulation.

 \subsection{Nonlinear Inverse Problem} \label{Algsection2}

 Consider the nonlinear inverse problem
 \[
  y=  \mG(z)+\eta, \quad \eta \in N(0,\Gamma),
 \]
 where $y=(y_1^T,y_2^T,\ldots,y_B^T)^T$, $\Gamma=\diag[V,V,\ldots,V]$ is a diagonal block matrix, each block $V$ is of size $n\times n$, and $N=B n$ for some positive constant $B$. To reformulate the problem in the central dynamic system (\ref{EnKFasimeq}),
 we define an augmented state vector by an $n$-vector $\gamma_t$:
 \begin{equation} \label{augeq1}
 x_t  =\begin{pmatrix} z \\ \gamma_{t} \\ \end{pmatrix}, \quad 
 \gamma_{t}  = \mG_{t}(z)+u_{t}, \quad u_{t} \sim N(0, \alpha V), 
 \end{equation}
 where $\mG_{t}(\cdot)$ is the mean response function for a mini-batch of data drawn at stage $t$, 
 and $0<\alpha<1$ is a pre-specified constant. In this paper, $\alpha$ is called the variance splitting proportion.  
 
 Let $\pi(z)$ denote the prior density function of $z$, which is differentiable with respect to $z$.
 The conditional distribution of $\gamma_{t}$ is $\gamma_{t}|z \sim N( \mG_{t}(z), \alpha V)$, and then the joint density function of $x_t$ is  $\pi(x_t)=\pi(z) \pi(\gamma_{t}|z)$.
 Based on Langevin dynamics, a system identical to 
 (\ref{inveq1}) in symbol can be constructed for the nonlinear inverse problem: 
 \begin{equation} \label{augeq2}
 \begin{split} 
  x_{t} &= x_{t-1}+ \epsilon_t \frac{n}{2N}\nabla_x \log \pi(x_{t-1}) +  w_t \\ 
  y_{t} &= H_t x_{t}+v_{t}, \\ 
 \end{split}
 \end{equation}
where $w_t \sim N(0,\frac{n}{N} Q_t)$, $Q_t=\epsilon_t I_p$,
 $p$ is the dimension of $x_t$;  $H_t=(0,I)$ such that $H_t x_t=\gamma_{t}$;
 $v_t \sim N(0, (1-\alpha) V)$, which
 is independent of $w_t$ for all $t$; and $y_t$ is a mini-batch 
 sample randomly drawn from $\{y_1, y_2, \ldots, y_{B} \}$.
 
 By the variance splitting state augmentation approach, we have successfully converted the nonlinear inverse problem to 
 the dynamic system (\ref{EnKFasimeq}).
 In this approach, $z$, $\gamma_t$ and $y_t$ form a hierarchical model, and it is easy to derive that 
 \begin{equation} \label{condeq}
 \gamma_t|z,y_t \sim N(\alpha y_t+(1-\alpha) \mG_t(z),  \alpha (1-\alpha)V),
 \end{equation} 
 which will be used later in justifying efficiency of the Langevinized EnKF for nonlinear inverse problems. In particular, if $\alpha$ is close to 1, the conditional variance of $\gamma_t$ given $y_t$ and $z$ can be much smaller than $V$.   
 
 With the above formulation, the following variant of the Langevinized EnKF can be applied to
 simulate from the posterior distribution $\pi(x|y)$. The posterior samples of $\pi(z|y)$
 can be obtained from those of $\pi(x|y)$ via marginalization. 
  Let $x_{t,k}^{a,i}$ denote the $i$-th sample obtained at iteration $k$ of stage $t$. 
 In each stage of the algorithm, a mini-batch sample
  $y_t$ is drawn at random and the augmented state $x_{t}$ is updated for $\mK$ iterations.
  
\begin{algorithm} \label{EnKFnonlinear} (Langevinized EnKF for Nonlinear Inverse Problems)
 \begin{itemize}
 \item[0.] (Initialization) Start with an initial ensemble 
  $x_{1,0}^{a,1}, x_{1,0}^{a,2}, \ldots, x_{1,0}^{a,m}$,
  where $m$ denotes the ensemble size. 
  For each stage $t=1,2,\ldots, T$ (first loop), do the following: 

  \item[1.] (Subsampling) Draw without replacement a mini-batch sample, denoted by $(y_t,H_t)$, of size $n$ from
             the full dataset of size $N$. 
 
  \item For each iteration $k=1,2,\ldots,\mK$ (second loop), 
   set $Q_{t,k}=\epsilon_{t,k} I_p$,
   $R_t=2(1-\alpha) V$ and the Kalman gain matrix $K_{t,k}=Q_{t,k} H_{t}^T (H_{t} Q_{t,k} H_{t}^T+R_t)^{-1}$, 
   and do steps 2-3:  

  \item[2.] (Forecast) For each chain $i=1,2,\ldots,m$ (third loop), draw $w_{t,k}^i \sim N_{p}(0, \frac{n}{N} Q_{t,k})$ 
   and calculate
   \begin{equation} \label{Algeq1b}
    x_{t,k}^{f,i}=x_{t,k-1}^{a,i} + \epsilon_{t,k} \frac{n}{2N} \nabla \log \pi(x_{t,k-1}^{a,i}) +w_{t,k}^i, 
   \end{equation}
   where, if $k=1$, $x_{t,0}^{a,i}=x_{t-1,\mK}^{a,i}$ for its $z$-component and $x_{t,0}^{a,i}=y_t$ for its $\gamma$-component.

  \item[3.] (Analysis) For each chain $i=1,2,\ldots,m$ (third loop), draw $v_{t,k}^i \sim N_{n}(0, \frac{n}{N} R_t)$, 
    and calculate
   \begin{equation} \label{Algeq2b}
    x_{t,k}^{a,i} = x_{t,k}^{f,i}+K_{t,k} (y_{t}-H_{t} x_{t,k}^{f,i}-v_{t,k}^i) 
     \stackrel{\Delta}{=}x_{t,k}^{f,i} +K_{t,k} (y_{t}-y_{t,k}^{f,i}).
   \end{equation}
 \end{itemize}
\end{algorithm}

Compared to Algorithm \ref{EnKFnew}, this algorithm includes a few more iterations at each stage. The
added iterations help to drive $\gamma_t$ towards its 
conditional equilibrium (\ref{condeq}). 
Regarding the convergence of the algorithm, we have the following  remark:

 \begin{remark} \label{rem2a}
 In equation (\ref{Algeq1b}), $\nabla_x \log \pi(x_{t,k-1})$ is calculated based on a mini-batch of data: 
 \begin{equation} \label{remeq1}
 \nabla_x \log \pi(x_{t,k-1}) =\begin{pmatrix} \nabla_z \log \pi(z_{t,k-1})+\frac{1}{\alpha} \frac{N}{n}
  \nabla_z \mG_{t}(z_{t,k-1}) V^{-1} \left(\gamma_{t,k-1}- \mG_{t}(z_{t,k-1})\right) \\ 
  -\frac{1}{\alpha} V^{-1} \left(\gamma_{t,k-1}-\mG_{t}(z_{t,k-1})\right) \\ 
 \end{pmatrix},
 \end{equation}
 where the component $\nabla_z \log \pi(z_{t,k-1})+\frac{1}{\alpha} \frac{N}{n}
 \nabla_z \mG_{t}(z_{t,k-1}) V^{-1} \left(\gamma_{t,k-1}- \mG_{t}(z_{t,k-1})\right)$
 provides an unbiased estimate of $\nabla_z \log\pi(z|y)$ as implied by (\ref{condeq}).
 It follows from the 
 standard convergence theory of SGLD that the $z$-component of $x_{t,k}$ will converge to 
 $\pi(z|y)$, provided that $\epsilon_{t,k}$ satisfies the condition:  
  $\{\epsilon_{t,k}\}$ is a positive sequence, decreasing in $t$ and non-increasing in $k$,  such that for any $k \in \{1,2,\ldots, \mK\}$,  $\epsilon_{t,k}=O(1/t^{\varsigma})$ for some $0 < \varsigma <1$. 
  \end{remark}

  From the Kalman gain matrix $K_{t,k}$, it is easy to see that only the $\gamma$-component of $x_{t,k}$ is updated 
  at the analysis step. Intuitively, $\gamma_{t,k}$ can converge very fast, as it is updated with the second-order gradient information.  Therefore,
  $\mK$ is not necessarily very large. In this paper, $\mK=5$ is set as the default. Further, 
  by (\ref{condeq}),  $\nabla_z \log \pi(z_{t,k-1})+\frac{1}{\alpha} \frac{N}{n}
  \nabla_z \mG_{t}(z_{t,k-1}) V^{-1} (\gamma_{t,k-1}-\mG_{t}(z_{t,k-1}))$ represents an improved gradient estimator compared to the estimator $\nabla_z \log \pi(z_{t})+\frac{N}{n}
  \nabla_z \mG_{t}(z_{t})$ $V^{-1} \left(y_{t}- \mG_{t}(z_{t})\right)$ used by  
  SGLD in simulating from $\pi(z|y)$. It is easy to show that the two stochastic gradients have the same mean value, but the former has a smaller variance than the latter. More precisely, 
  \begin{equation} \label{effeq}
  \var\left (\frac{1}{\alpha} \frac{N}{n} \nabla_z \mG_{t}(z_{t}) V^{-1} \left(\gamma_{t}- \mG_{t}(z_{t})\right) \big| z_t \right) = \frac{1-\alpha}{\alpha} \var \left (  \frac{N}{n} \nabla_z \mG_{t}(z_{t}) V^{-1} \left(y_{t}- \mG_{t}(z_{t})\right) \big | z_t\right ),
  \end{equation}
  which implies that for nonlinear inverse problems, the Langevinized EnKF represents a variance reduction version of SGLD, and it is potentially more efficient than SGLD if $0.5 <\alpha<1$ is chosen. 
  In this paper, we set $\alpha=0.9$ as the default and initialized $\gamma_{t,0}$ by $y_t$ at each stage, which enhances the convergence of the simulation.
   

 \subsection{Data Assimilation with Linear Measurement Equation} \label{Algsection3}

 Consider the dynamic system (\ref{EnKFasimeq}), 
 for which we assume that at each stage $t$, $y_t$ can be partitioned into $B_t=N_t/n_t$ blocks such that
 $y_{t,k}=H_{t,k}x_t+v_{t,k}$, $k=1,2,\ldots B_t$,
 where $N_t$ is the total number of observations at stage $t$,
 $y_{t,k}$ is a block of $n_t$ observations randomly drawn from  $y_t=\{y_{t,1},\ldots, y_{t,B_t}\}$,
 $v_{t,k} \sim N(0,V_t)$ for all $k$, and $v_{t,k}$'s are mutually independent.
 
 To motivate the development of the algorithm, we first consider the Bayesian formula
 \begin{equation} \label{Bayeseq}
 \pi(x_t|y_{1:t})=\frac{f(y_t|x_t) \pi(x_t|y_{1:t-1})}{\int f(y_{t}|x_t)\pi(x_t|y_{1:t-1}) dx_t },
 \end{equation}
 which suggests that to get the filtering distribution $\pi(x_t|y_{1:t})$, the predictive distribution 
 $\pi(x_t|y_{1:t-1})$ should be used as the prior at stage $t$. To estimate the gradient $\nabla \log\pi(x_t|y_{1:t-1})$, we employ the following identity established in \cite{SongLiang2020}:
\begin{equation} \label{Fishereq1} 
\nabla_{\beta}\log\pi(\beta \mid D)=
 \int \nabla_{\beta}\log\pi(\beta \mid \gamma,D) \pi(\gamma \mid \beta, D) d\gamma,
\end{equation}
where $D$ denotes data, and $\beta$ and $\gamma$ denote two parameters of a posterior distribution $\pi(\beta,\gamma|D)$.  
By the identity, we have 
\begin{equation} \label{identeq} 
\begin{split}
\nabla_{x_t} \log \pi(x_t|y_{1:t-1}) & =\int \nabla_{x_t} \log\pi(x_{t}|x_{t-1},y_{1:t-1}) \pi(x_{t-1}|x_t,y_{1:t-1}) d x_{t-1} \\
&= \int \nabla_{x_t} \log\pi(x_{t}|x_{t-1}) \frac{\pi(x_{t-1}|x_{t},y_{1:t-1})}{\pi(x_{t-1}|y_{1:t-1})} \pi(x_{t-1}|y_{1:t-1}) d x_{t-1} \\
&= \int \nabla_{x_t} \log\pi(x_{t}|x_{t-1}) \omega(x_{t-1}|x_t) \pi(x_{t-1}|y_{1:t-1}) d x_{t-1},
\end{split}
\end{equation}
where $\omega(x_{t-1}|x_t)=
{\pi(x_{t-1}|x_{t},y_{1:t-1})}/{\pi(x_{t-1}|y_{1:t-1})} 
=\pi(x_{t}|x_{t-1})/\pi(x_t|y_{1:t-1}) \propto \pi(x_{t}|x_{t-1})$, as $\pi(x_t|y_{1:t-1})$ is a constant for a given particle $x_t$ and the data $\{y_1, y_2,\ldots,y_{t-1}\}$.
Therefore, given a set of samples $\mX_{t-1}=\{x_{t-1,1}, x_{t-1,2}, \ldots, x_{t-1,m'} \}$ 
 drawn from $\pi(x_{t-1}|y_{1:t-1})$, 
 an importance resampling procedure can be employed to draw samples from $\pi(x_{t-1}|x_t, y_{1:t-1})$. 
The importance resampling procedure can be executed very fast, as calculation of the importance weight $\omega(x_{t-1}|x_t)$ does not involve any data.

 With the above formulas and Langevin dynamics,  
 we can construct a dynamic system at stage $t$ as
 \begin{equation} \label{asimeq2}
 \begin{split} 
  x_{t,k} &= x_{t,k-1}- \epsilon_t \frac{n_t}{2N_t} U_t^{-1} (x_{t,k-1}-g(\tilde{x}_{t-1,k-1})) +  w_{t,k}, \\ 
  y_{t,k} &= H_{t,k} x_{t,k}+v_{t,k}, \\ 
 \end{split}
 \end{equation}
 for $k=1,2,\ldots$, where $x_{t,0}=g(x_{t-1})+u_t$;   $\tilde{x}_{t-1,k-1}$ denotes an approximate
 sample drawn from $\pi(x_{t-1}|x_{t,k-1}, y_{1:t-1})$ at iteration $k$ of stage $t$ 
 through the importance resampling procedure;
 $w_{t,k} \sim N(0, \frac{n_t}{N_t} \epsilon_{t,k} I_p)$, $Q_{t,k}=\epsilon_{t,k} I_p$,
 and $p$ is the dimension of $x_{t}$. Note that
 $-U_t^{-1} (x_{t,k-1}-g(\tilde{x}_{t-1,k-1}))$ forms  
 an unbiased estimator of $\nabla \log \pi(x_{t,k-1}|y_{1:t-1})$ only in the ideal case that 
the samples in $\mX_{t-1}$ follows from the distribution $\pi(x_{t-1}|y_{1:t-1})$ and the sample size $|\mX_{t-1}|$ is sufficiently large. Our theory for the convergence of the algorithm 
has taken care of the deviations from the ideal case as discussed in Remark \ref{rem3b}. In practice, 
$\mX_{t-1}$ can be collected at stage $t-1$ from iterations $k_0+1,\ldots, \mK$, where $k_0$ denotes the burn-in period and $\mK$ denotes the total number of iterations performed at each stage.
Applying the Langevinized EnKF at stage $t$ leads to the following algorithm:

\begin{algorithm}\label{EnKFasimb} (Langevinized EnKF for Data Assimilation)
 \begin{itemize}
 \item[0.] (Initialization) Start with an initial ensemble $x_{1,0}^{a,1}, x_{1,0}^{a,2}, \ldots, x_{1,0}^{a,m}$ drawn from  
  the prior distribution $\pi(x_1)$, where $m$ denotes the ensemble size.  
  Set $\mX_t=\emptyset$ for $t=1,2,\ldots,T$. Set $k_0$ as the common burnin period for each stage $t$.  
  For each stage $t=1,2,\ldots,T$ (first loop), 
   and for each iteration $k=1,2,\ldots,\mK$ (second loop), 
   do the following: 

 \item[1.] (Subsampling) Draw without replacement a mini-batch sample, denoted by $(y_{t,k},H_{t,k})$, of size $n_t$ from
          the full dataset of size $N_t$.
 
 \item Set $Q_{t,k}=\epsilon_{t,k} I_p$, $R_t=2V_t$, and the Kalman gain matrix
        $K_{t,k}=Q_{t,k} H_{t,k}^T (H_{t,k} Q_{t,k} H_{t,k}^T+R_t)^{-1}$. 
        For each chain $i=1,2,\ldots, m$ (third loop), do steps 2-5:

 \item[2.] (Importance resampling) If $t>1$,  
  calculate importance weights $\omega_{t,k-1,j}^{i}
  =\pi(x_{t,k-1}^{a,i}|x_{t-1, j})=\phi(x_{t,k-1}^{a,i}: g(x_{t-1,j}),U_t)$ for $j=1,2,\ldots,|\mX_{t-1}|$, 
  where $\phi(\cdot)$ denotes a Gaussian density, and $x_{t-1, j} \in \mX_{t-1}$ denotes the $j$th
  sample in $\mX_{t-1}$; if $k=1$, set $x_{t,0}^{a,i}=g(x_{t-1,\mK}^{a,i})+u_t^{a,i}$ and  $u_t^{a,i} \sim N(0, U_t)$. 
  Resample $s\in \{1,2,\ldots,|\mX_{t-1}|\}$ with a probability $\propto \omega_{t,k-1,s}^i$, 
 i.e., $P(S_{t,k,i}=s)=\omega_{t,k-1,s}^i/\sum_{j=1}^{|\mX_{t-1}|} \omega_{t,k-1,j}^i$, and   
 denote the sample drawn from $\mX_{t-1}$ by $\tilde{x}_{t-1,k-1}^i$. 
  
  \item[3.] (Forecast) Draw $w_{t,k}^i \sim N_{p}(0, \frac{n}{N} Q_{t,k})$. 
    If $t=1$, set 
   \begin{equation} \label{Algeq1c0}
    x_{t,k}^{f,i}=x_{t,k-1}^{a,i}-\epsilon_{t,k} \frac{n_t}{2N_t} \nabla\log\pi(x_{t,k-1}^{a,i})+w_{t,k}^i,
   \end{equation}
   where $\pi(\cdot)$ denotes the prior distribution of $x_1$.
   If $t>1$, set 
   \begin{equation} \label{Algeq1c}
    x_{t,k}^{f,i}=x_{t,k-1}^{a,i}-\epsilon_{t,k} \frac{n_t}{2N_t} U_t^{-1} (x_{t,k-1}^{a,i}-g(\tilde{x}_{t-1,k-1}^i))+w_{t,k}^i.
   \end{equation}

  \item[4.] (Analysis) Draw $v_{t,k}^i \sim N_{n}(0, \frac{n}{N} R_t)$ and set  
   \begin{equation} \label{Algeq2c}
    x_{t,k}^{a,i} = x_{t,k}^{f,i}+K_{t,k} (y_{t,k}-H_{t,k} x_{t,k}^{f,i}-v_{t,k}^i) 
     \stackrel{\Delta}{=}x_{t,k}^{f,i} +K_{t,k} (y_{t,k}-y_{t,k}^{f,i}).
   \end{equation}

  \item[5.] (Sample collection) If $k>k_0$, add the sample $x_{t,k}^{a,i}$ into the set $\mX_{t}$. 
 \end{itemize}
\end{algorithm}

 Since $\mK$ is usually not allowed to go to $\infty$ for such a dynamic system,  $x_{t,\mK}^{a,1},\ldots,x_{t,\mK}^{a,m}\sim \pi(x_t|y_{1:t})$ only 
approximately holds for each stage $t$. 
However, as shown in Theorem S2 of the Supplementary Material, a smaller approximation error at one stage  helps to reduce the approximation error at the next stage; and the approximation error becomes negligible as the number of stages increases, if the number of observations at each stage is reasonably large. 
To facilitate discussion, Theorem S2 is restated as a proposition in a less rigorous language in what follows.

\begin{prop} \label{prop2}
 Assume that for each stage $t$, 
the matrices $H_t$, $U_t$ and $V_t$ and the state propagator $g(x_t)$ satisfy mild regularity conditions (listed in the Appendix).
Then there exists a sufficiently large iteration number $\mK$ such that $\limsup_{t\to \infty} W_2(\tilde{\pi}(x_t|y_{1:t}),\pi(x_t|y_{1:t})) 
= O\left( \frac{1}{{\liminf_t N_t}}\right)$.

\end{prop}



\begin{remark} \label{rem3b} 
In Step 2, the probability $P(S_{t,k,i}=s)$ is calculated in a self-normalized importance sampling estimator, which is known to be consistent but biased. The bias has 
been taken into account in our proof of Theorem S2 as
implied by equation (S17) in the Supplementary Material.
\end{remark}

\begin{remark} \label{biasremark} We argue that
the convergence rate $O({1}/{\liminf_t N_t})$ is a reasonable precision. If the state $x_t$ is observed, then it is not difficult to derive that $\pi(x_{t+1})\sim N(g(x_t), \Gamma_t)$ and $\pi(x_{t+1}|y_{1:t+1})=\pi(x_{t+1}|y_{t+1})\sim N(\mu_{t+1}, K)$ for some $\mu_{t+1}$ and $K^{-1}=H_t^T\Sigma_t^{-1}H_t+\Gamma_t^{-1}$. When the eigenvalues of $\Sigma_t$ are bounded by constant, and $H$ is row-wisely independent, then the eigenvalues of $K$ have a lower bound of $O(1/N_t)$. In other words, even when the state  $x_t$ is known, the filtering distribution $\pi(x_{t+1}|y_{1:t+1})$ has a variation scale of at least $O(1/\sqrt{N_t})$. Compared to this variation scale, the estimation inaccuracy of  $W_2(\tilde\pi_{t+1},\pi_{t+1})$ is really negligible.
\end{remark}



It is known that weight degeneracy is an inherent default of sequential importance sampling (SIS) (also known as sequential Monte Carlo), especially when the dimension of the system is high. When it 
occurs, the importance weights concentrates on a few samples, the effective sample size is low, and the 
resulting importance sampling estimate is heavily biased. Fortunately, the Langevinized EnKF is essentially immune to this issue. In Langevinized EnKF, the importance resampling procedure  is to draw from $\mX_{t-1}$ a particle which matches a given particle $x_t$ in state propagation such that the gradient $\nabla_{x_t} \log \pi(x_t|y_{1:t-1})$ can be reasonably well estimated, which is then combined with 
the gradient of the likelihood function of the new data $y_t$ to have $x_t$ updated. By (\ref{Bayeseq}), $\pi(x_t|y_{1:t-1})$ works as the prior distribution of $x_t$ for the filtering distribution $\pi(x_t|y_{1:t})$. Therefore, the effect of the importance resampling procedure on the performance of the algorithm is limited if the sample size of $y_t$ is reasonably large at each stage $t$.
In contrast, the importance resampling procedure in SIS is to draw a particle from $\mX_{t-1}$ and treat the particle as from 
the filtering distribution $\pi(x_t|y_{1:t})$. 
For high-dimensional problems, the overlap between the high density regions of the neighboring stage filtering distributions can be very small, which naturally 
causes the weight degeneracy issue.

In summary, the importance resampling step of the Langevinized EnKF aims to draw a sample for the prior $\pi(x_t|y_{1:t-1})$ used at each stage $t$, 
while that of SIS aims to draw a sample for the posterior $\pi(x_t|y_{1:t})$. In consequence, the Langevinized EnKF is less bothered by the weight degeneracy issue.

\subsection{Data Assimilation with Nonlinear Measurement Equation} \label{Algsection4}

To apply the Langevinized EnKF to the case that the measurement equation is nonlinear, 
the variance splitting state augmentation approach proposed in Section \ref{Algsection2} can be used. For simplicity, we describe the algorithm under the full data scenario. Extension of the algorithm to the mini-batch scenario  is straightforward.  Consider the dynamic system
 \begin{equation} \label{Nasim}
 \begin{split}
   z_{t}& =g(z_{t-1})+u_{t}, \quad u_{t}\sim N(0, U_{t}), \\
   y_t & =h(z_t)+\eta_t, \quad \eta_{t} \sim N(0, \Gamma_t),   \\
 \end{split}
 \end{equation}
 where both $g(\cdot)$ and $h(\cdot)$ are nonlinear.
 As for nonlinear inverse problems, we can
 augment the state $z_t$ by $\gamma_{t}$ at each stage $t$, where
\[
   \gamma_{t}=h(z_t)+\xi_t, \quad \xi_{t} \sim N(0, \alpha_t \Gamma_t),
 \]
 for some constant $0<\alpha_t<1$.  
 Let $x_t^T=(z_t^T,\gamma_{t}^T)$ and thus
 $\pi(x_t|z_{t-1}, y_{1:t-1})=\pi(z_t|z_{t-1},y_{1:t-1}) \pi(\gamma_{t}|z_t)$.
 Similar to (\ref{asimeq2}), at stage $t$,  we have the following dynamic system
 \begin{equation} \label{augeqnon}
 \begin{split} 
  x_{t,k} &= x_{t,k-1}+ \frac{\epsilon_t}{2} \nabla_z \log \pi(x_{t,k-1}|\tilde{x}_{t-1,k-1},y_{1:t-1}) +  w_{t,k} \\ 
  y_{t,k} &= H_t x_{t,k}+v_{t,k}, \\ 
 \end{split}
 \end{equation}
where $\tilde{x}_{t-1,k-1}$ denotes a sample drawn from $\pi(x_{t-1}|x_{t,k-1},y_{1:t-1})$, 
$x_{t,k}$ denotes the estimate of $x_t$ obtained at iteration $k$ of stage $t$, 
 $y_{t,k}=y_t$ for all $k=1,2,\ldots, \mK$, $H_t=(0,I)$ such that $H_t x_{t}=\gamma_{t}$,
 $w_{t,k} \sim N(0,\epsilon_{t,k} I_p)$, $p$ is the dimension of $x_t$,
 and $v_{t,k} \sim N(0, (1-\alpha_t) \Gamma_t)$.
 Algorithm \ref{EnKFasimb} can then be applied to simulate samples from
 $\pi(x_t|y_{1:t})$ and thus the samples from $\pi(z_t|y_{1:t})$ can be obtained by marginalization. 
 In mapping the ensemble from stage $t-1$ to stage $t$, $z_{t}$ can be set according to the 
 state evolution equation given in (\ref{Nasim}), while $\gamma_{t}$ can be set to $y_t$. 
 The convergence of the algorithm follows from Theorem S2 and Remark \ref{rem2a}. 


\section{Numerical Studies for Static Learning Problems}

 This section illustrates the performance of the Langevinized EnKF as a general 
 Monte Carlo algorithm for complex inverse problems. 
 Two big data examples are considered, one 
 is Bayesian variable selection for linear regression, and the other is Bayesian 
 variable selection for general nonlinear systems modeled by deep neural networks (DNNs). 
  
\subsection{Bayesian Variable Selection for Large-Scale Linear Regression} \label{LinearEx}

 Consider the linear regression
 \begin{equation} \label{Lineq}
  \bY=\bZ \bbeta +\varepsilon,
 \end{equation}
 where $\bY \in \mR^N$, $\bZ=(Z_1,Z_2,\ldots,Z_p)\in \mR^{N\times p}$, $\bbeta \in \mR^p$, and $\varepsilon\sim N(0, I_N)$.
 An intercept term has been implicitly included in the model.
 We generate ten datasets from this model with $N=50,000$, $p=2,000$, and
 $\bbeta=(\beta_1,\beta_2,\ldots,\beta_{p})=(1,1,1,1,1,-1,-1,-1,0,\ldots,0)$. That is, the first 8 variables are true
  and the others are false. Each variable in $Z$ has a marginal distribution
 of $N(0,I_N)$, but they are mutually correlated with a correlation coefficient of 0.5.

 To conduct Bayesian analysis, we consider the following hierarchical mixture Gaussian 
 prior distribution which, with the latent variable $\xi_i \in \{0,1\}$, can be 
  expressed as 
 \begin{equation} \label{prioreq1}
 \begin{split}
 & \beta_i |\xi_i  \sim (1-\xi_i) N(0, \tau_1^2)+\xi_i N(0,\tau_2^2), \\
 &  P(\xi_i=1) =1-P(\xi_i=0)=p_0, \\
 \end{split} 
 \end{equation} 
 for $i=1,2,\ldots,p$. Such a prior distribution has been widely 
 used in the literature of Bayesian variable selection, see e.g. \cite{GeorgeM1993}. 
 To apply the Langevinized EnKF to this problem, we first 
 integrate out $\xi_i$ from the prior distribution (\ref{prioreq1}), which leads to the marginal distribution 
 \begin{equation} \label{prioreq2}
  \beta_i \sim (1-p_0) N(0,\tau_1^2) +p_0 N(0,\tau_2^2),  \quad i=1,2,\ldots,p,
 \end{equation}  
 for which the log-prior density is differentiable. 
 Algorithm \ref{EnKFnew} can then be applied to simulate 
 from the posterior distribution $\pi(\bbeta|\bY,\bZ)$. 
 In the simulation, we set $p_0=1/p=0.0005$, $\tau_1^2=0.01$ and $\tau_2^2=1$ for
 the prior distribution, and set the ensemble size $m=100$, the mini-batch size $n=100$,
 and the learning rate $\epsilon_t=0.2/\max\{t_0, t\}^{0.6}$, where $t_0= 100$.
  The algorithm was run for 10,000 iterations, which 
 cost 375 CPU seconds on a personal computer with 2.9GHz Intel Core i7 CPU and 16GB RAM. All computation reported in this paper were done on the same computer.
 
 To accomplish the goal of variable selection, we consider the factorization of the posterior distribution 
 $\pi(\bbeta,\bxi|\bY) \propto \pi(\bY|\bbeta) \pi(\bbeta|\bxi) \pi(\bxi)$, where $\bxi=(\xi_1,\xi_2,\ldots,\xi_p)$. Under 
 {\it a priori} independent assumptions for $\beta_i$'s and $\xi_i$'s,  
  we are able to draw posterior samples of $\bxi$ from the following distribution:
\begin{equation} \label{connectioneq}
  \pi(\xi_{ti}=1|\beta_{ti},\bY) = \frac{a_{ti}}{a_{ti}+b_{ti}}, \quad i=1,2,\ldots,p, 
\end{equation}
 where $a_{ti}=\frac{p_0}{\tau_2} \exp(-\beta_{ti}^2/2\tau_2^2)$,
 $b_{ti}=\frac{1-p_0}{\tau_1} \exp(-\beta_{ti}^2/2 \tau_1^2)$, 
 and $\beta_{ti}$ denotes the posterior sample of $\beta_i$ drawn by Algorithm \ref{EnKFnew}
 at stage $t$. Here we denote by $\bbeta_t=(\beta_{t,1},\beta_{t,2},\ldots,\beta_{t,p})$  
 a posterior sample of $\bbeta$ drawn by Algorithm \ref{EnKFnew} at the analysis step of stage $t$.  
 
 Figure \ref{fig:linear} summarizes the variable selection results for one data set. 
 The results for other data sets are similar. 
 Figure \ref{fig:linear}(a) shows the sample trajectories of $\beta_1,\beta_2,\ldots, \beta_9$, which are averaged over the ensemble
 along with iterations.   
 All the samples converge to their true values within 100 iterations, taking about 3.7 CPU seconds.
 This is extremely fast!
 Figure \ref{fig:linear}(b) shows the marginal inclusion probabilities of 
 the variables $Z_1, Z_2,\ldots, Z_{p}$. 
  From this graph, we can see that each of the 8 true variables (indexed 1-8) has a  
  marginal inclusion probability close to 1, while each false variable has a marginal inclusion probability close to 0. 
 Figure \ref{fig:linear}(c) shows the scatter plot of the response variable
  and its fitted value for the training data, 
 and Figure \ref{fig:linear}(d) shows the scatter plot of the response variable and its predicted value for  
 200 test samples generated from the model (\ref{Lineq}).
 In summary, Figure \ref{fig:linear} shows that the Langevinized EnKF is able to 
 identify true variables for large-scale linear regression and, moreover, it is extremely efficient. 
 
 \begin{figure}[htbp]
      \centering
      \includegraphics[width=5.0in]{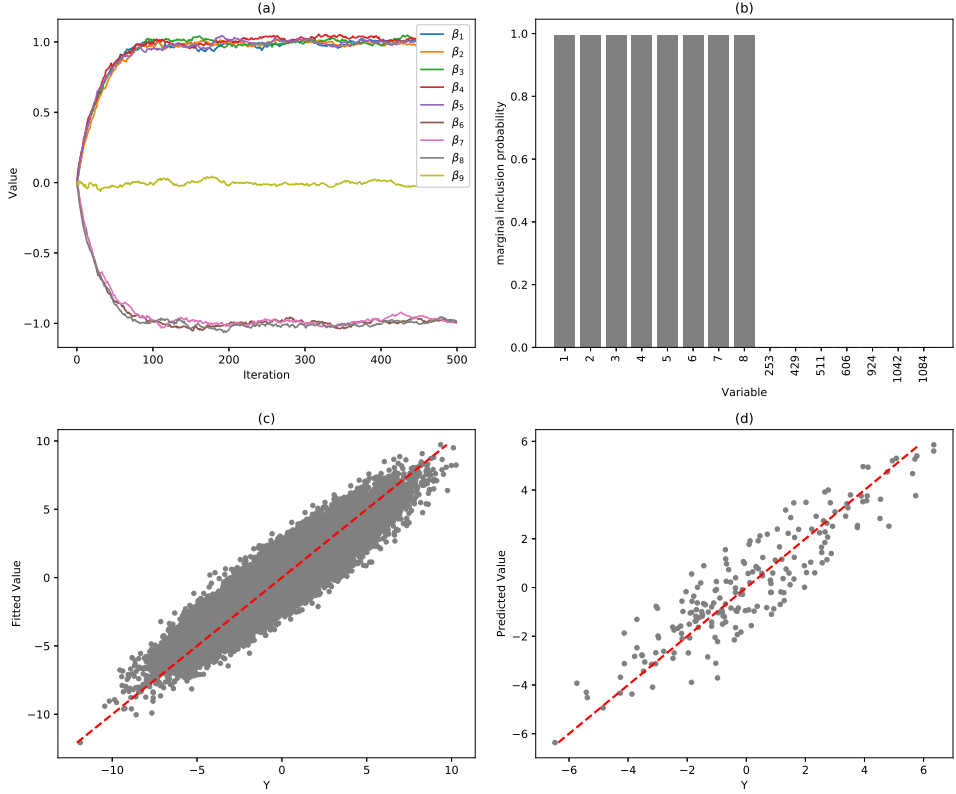}
      \caption{\quad Langevinized EnKF for large-scale linear regression with 500 iterations: (a) Trajectories of 
  $\beta_1,\ldots,\beta_9$,
 where $\beta_1,\ldots,\beta_5$ have a true value of 1,  $\beta_6,\ldots,\beta_8$ have 
 a true value of $-1$,  and $\beta_9$ has a true value of 0.  
   (b) marginal inclusion probabilities of all covariates $X_1,\ldots, X_p$, where the covariates are shown in the rank of marginal inclusion probabilities; 
   (c) scatter plot of the response Y and the fitted value for training samples; 
   and (d) scatter plot of the response Y and the predicted value for test samples. }
  \label{fig:linear}
\end{figure}

 

 For comparison, SGLD \citep{Welling2011BayesianLV}, preconditioned SGLD \citep[pSGLD,][]{Li2016PSGLD}, and 
stochastic gradient Nos\'e-Hoover thermostat \citep[SGNHT,][]{Ding2014BayesianSU}
 were applied to this example. For these algorithms, the learning rates have been tuned to
 their maximum values such that the simulation converges fast while not exploding.
 For SGLD, we set $\epsilon_t=4\times10^{-6}/\max\{t_0, t\}^{0.6}$ with $t_0 = 1000$; 
  for pSGLD, we set
 $\epsilon_t=5\times10^{-6}/\max\{t_0, t\}^{0.6}$ with $t_0 = 1000$; and for SGNHT, we set $\epsilon=0.0001$.
 Other than the learning rate,  
 pSGLD contains two more tuning parameters, 
 which control the extremes of the
 curvatures and the balance of the weights of the historical and current gradients, respectively. 
 They both were set to the default values as suggested by \cite{Li2016PSGLD}.  
 SGNHT also contains an extra parameter, the so-called diffusion parameter, for which different values, including 1, 5, 10, and 20, have been tried. 
 The algorithm performed very similarly with each of 
 the choices. Figures \ref{Linearfig} summarizes the results of the algorithm
 with the diffusion parameter being set to $10$.
 
For a fair comparison, we ran SGLD for 20,000 iterations, pSGLD for 10,000 iterations, and SGNHT for 15,000 iterations, which took about 387 CPU seconds, 410 CPU seconds and 380 CPU seconds, respectively.
All these algorithms cost slightly longer CPU time than the Langevinized EnKF.
 Figure \ref{Linearfig} compares the trajectories of $(\beta_1,\beta_2,\ldots,\beta_9)$ produced by the three 
 algorithms. It indicates the Langevinized EnKF can converge significantly faster than
 SGLD, pSGLD, SGNHT for this example, which can be explained as the advantage of using the second-order gradient information in the Langevinized EnKF. 

\begin{figure}[htbp]
\begin{center}
\begin{tabular}{c}
\epsfig{figure=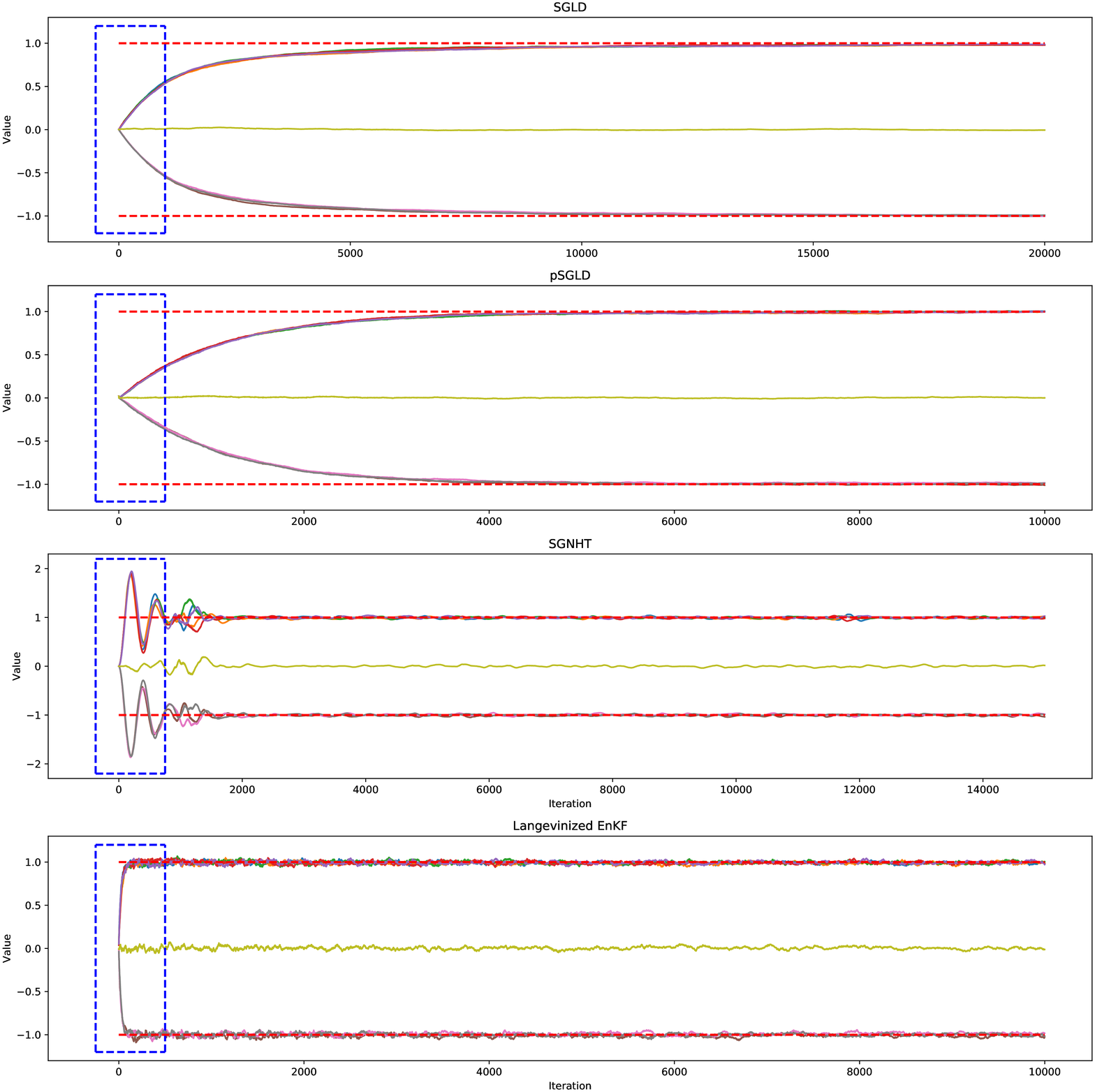,height=4.0in,width=5.0in,angle=0} 
\end{tabular}
\end{center}
 \caption{Convergence trajectories of SGLD, pSGLD, SGNHT and Langevinized EnKF for a large-scale linear regression example:  Trajectories of $(\beta_1,\beta_2,\ldots,\beta_9)$ produced by SGLD (upper), pSGLD (upper middle), SGNHT (lower middle),
  and Langevinized EnKF (lower) in their whole runs, where the blue rectangle highlights the first 
  5\% iterations of the runs.} 
\label{Linearfig}
\end{figure}

 Since the Langevinized EnKF is essentially a parallel SGRLD algorithm, we have also compared it with parallel SGLD, pSGLD and SGNHT. The results were presented in the Supplementary Material. The comparison gave a much larger margin that the Langevinized EnKF significantly 
 outperforms parallel runs of these algorithms.

\subsection{Bayesian Nonlinear Variable Selection with Deep Neural Networks}

The goal of this example is to show that the Langevinized EnKF can be used to train Bayesian DNNs and as one of its important application, the Bayesian DNN can be used in variable selection 
for nonlinear systems.   
We generated 10 data sets from the nonlinear regression model
\begin{equation}
y = \frac{10 x_1^2}{1 + x_2^2} + 5 \sin(x_3x_4) + x_5+\epsilon, 
\end{equation}
where $\epsilon \sim N(0,I_N)$. The variables $x_1$, $\cdots$, $x_5$ together with additional 94 variables were generated as in Section \ref{LinearEx}: all the variables
are mutually correlated with a correlation coefficient of 0.5, and each has a marginal distribution of $N(0,I_N)$. Each data set consists of $N=2,000,000$ samples for training and 200 samples for testing.
 
We modeled the data using a 3-hidden-layer neural network, with $p = 100$ input units including a bias unit 
(for the intercept term), 5 units on each hidden layer and one unit on the output layer. The activation function is LeakyRelu given by $LeakyRelu(x)=1_{(x<0)}(0.1 x)+1_{(x>=0)}(x)$. That is, we modeled the data by the neural network function 
 \[
 y = LeakyRelu(LeakyRelu(LeakyRelu(X\cdot W_1)\cdot W_2+b_1)\cdot W_3+b_2) \cdot W4 + b_3 +\epsilon,
 \]
 where $X=(x_0, x_1, x_2,\ldots, x_{99})$, 
 $W_1 \in \mR^{p\times 5}$, $W_2 \in \mR^{5\times 5}$, $W_3 \in \mR^{5\times 5}$ and $W_4\in \mR^{5 \times 1}$ represent the weight matrices at different layers of the neural network, and
 $b_1\in \mR^{5\times 1}$, $b_2 \in \mR^{5\times 1}$, $b_3 \in \mR^{5\times 1}$ represent bias vectors at different hidden layers. For each element of $W_i$'s and $b_i$'s, we assume 
 a hierarchical mixture Gaussian prior as given in (\ref{prioreq1}). 
 For the prior hyperparameters, we set  $p_0 = 0.01$, $\tau_1^2 = 0.05$, $\tau_2^2 = 1$.
 Conditioned on each posterior sample of $(W_1,W_2,W_3,W_4,b_1,b_2,b_3)$, a sparse neural network can be drawn by simulating an indicator variable for each potential connection of the neural network according to the Bernoulli distribution given in (\ref{connectioneq}).  
 
  To identify important input variables, for each sparse neural network drawn above, we define four $\bxi$-matrices $G_1$, $G_2$, $G_3$ and $G_4$, which correspond to $W_1$, $W_2$, $W_3$ and $W_4$, respectively. Each element of the $\bxi$-matrix indicates the status, existing or not, of the corresponding connection. 
  Define $G=G_1 \cdot G_2 \cdot G_3 \cdot G_4$ as a $p$-vector and further truncates its each element to 1 if greater than 1. That is, each element of $G$ indicates the effectiveness of the corresponding input variable.
  Averaging $G$ over the sparse network samples produces the marginal inclusion probabilities of the input variables.  
  
  Algorithm \ref{EnKFnonlinear} was first applied to this example with the variance split proportion $\alpha=0.9$, the ensemble size $m=20$, the mini-batch size $n=100$, the iteration number $\mK=5$ for each stage, and the total number of stages $T=20,000$.
  The learning rate was set to $\epsilon_{t,k}=4\times10^{-4}/\max\{k_0, k\}^{0.9}$ with $k_0 = 1$ for $k = 1,\cdots, \mK$. Each run took about 4000 CPU seconds. 
  For stochastic gradient MCMC algorithms, to avoid simulations to explode and meanwhile to achieve fast convergence, 
  we often need to start with a very small learning rate and then decrease it very 
  slowly or even keep it as a constant. This note also applies to the 
  remaining examples of this paper. 
  
 \begin{figure}[htbp]
      \centering
      \includegraphics[width = 5.0in]{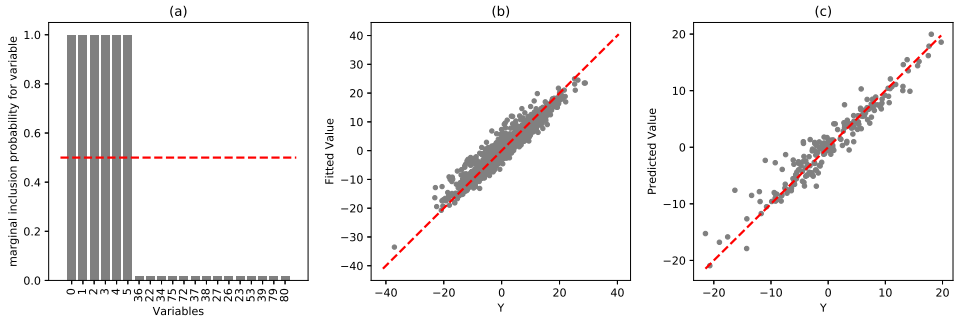}
      \caption{\quad Langevinized EnKF for the nonlinear variable selection example: 
   (a) marginal inclusion probabilities of the variables, where the variables are shown in the rank of marginal inclusion probabilities; 
   (b) scatter plot of the response Y and the fitted value for 2,000 randomly selected training samples; 
   and (c) scatter plot of the response Y and the predicted value for 200 test samples. }
  \label{fig:nonlinear}
\end{figure}

  Figure \ref{fig:nonlinear} summarizes the results for one dataset. The results for the other datasets are similar.
 Figure \ref{fig:nonlinear}(a) shows the marginal inclusion probabilities of all input variables $x_1, x_2, \ldots, x_p$ with a cutoff value of 0.5 (red dash line).
 The results are very encouraging: Each of the true variables (indexed 1-5) has a
  marginal inclusion probability close to 1, while each of the false variables has a 
  marginal inclusion probability close to 0. 
 Figure \ref{fig:nonlinear}(b) shows the scatter plot of the response variable and its fitted value for 2,000 randomly chosen training samples.
 Figure \ref{fig:nonlinear}(c) shows the scatter plot of the response variable and its predicted value for  200 test samples. 
 In summary, Figure \ref{fig:nonlinear} shows that the Langevinized EnKF provides a feasible algorithm for training Bayesian deep neural networks, through which important variables can be identified for nonlinear systems. 
  
For comparison, SGLD and pSGLD were applied to this example. Each algorithm was run in parallel with 20 chains and each chain consisted of 100,000 iterations. For both algorithms, we set the learning rate as 
$\epsilon_t=1\times10^{-4}/\max\{t_0, t\}$ with $t_0 = 10,000$, which has been tuned to its 
 maximum value such that the simulations converge very fast but won't explode. Each run of SGLD cost about 4000 CPU seconds, and each run of pSGLD cost about 5600 CPU seconds. For the Langevinized EnKF, we measured the fitting and prediction errors, in mean squared errors (MSEs), at the last iteration of each stage. 
 For SGLD and pSGLD, we measured the fitting and prediction errors
 at every 5th iteration. Figure \ref{fig:nonlincomp_parallel} (a) shows the paths of the best fitting and prediction MSEs produced by the time by each chain of pSGLD, SGLD, and Langevinized EnKF. Figure \ref{fig:nonlincomp_parallel} (b) shows the path of the best fitting and prediction MSEs by the time produced by respective algorithms. Both plots indicates the superiority of the Langevinized EnKF over SGLD and pSGLD: the Langevinized EnKF tends to produce smaller fitting and prediction errors than SGLD and pSGLD for this example.
 
 \begin{figure}[htbp]
		\centering
		\begin{tabular}{cc}
		\includegraphics[height=2.5in,width = 3.25in]{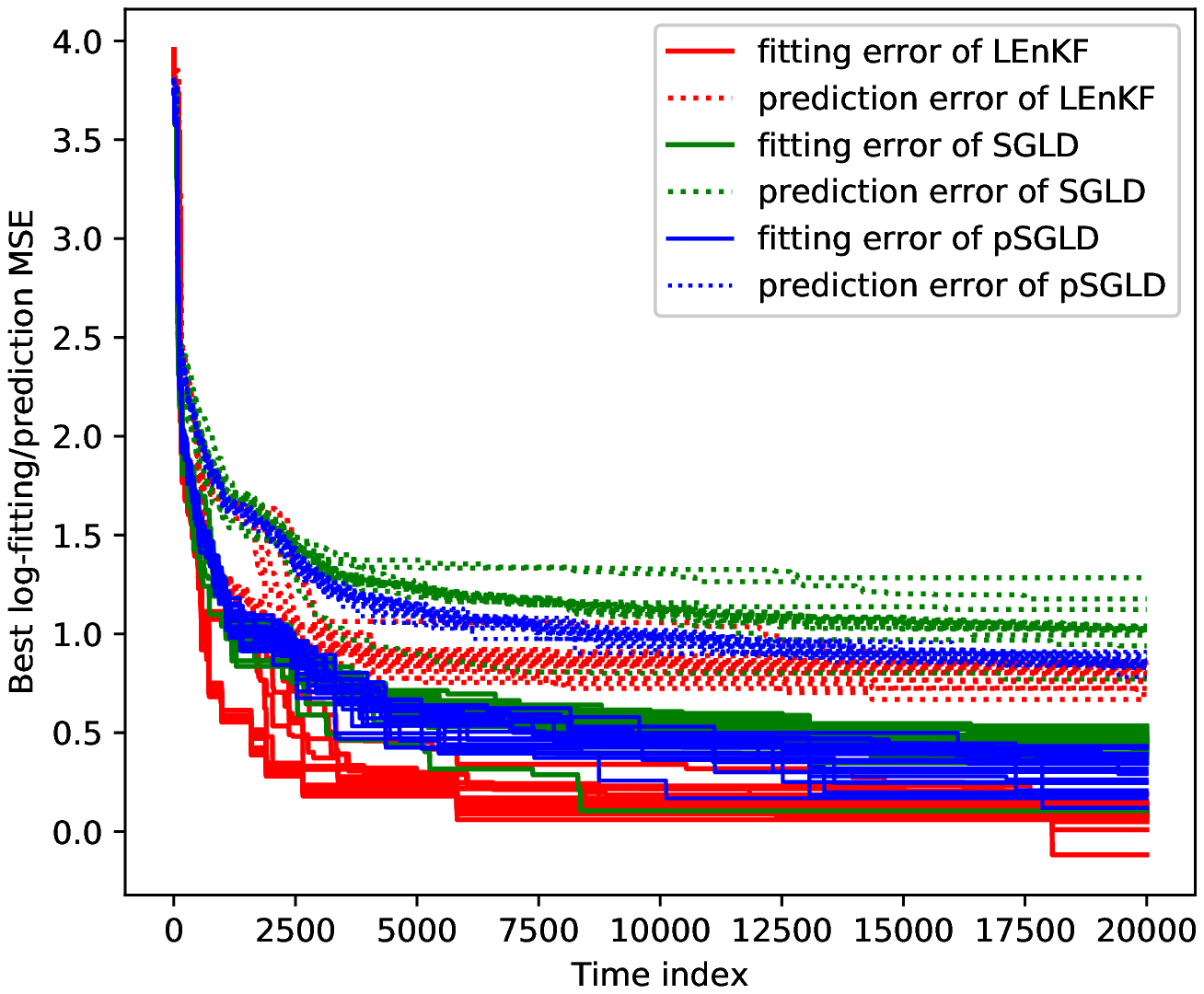} & 
		\includegraphics[height=2.5in,width = 3.25in]{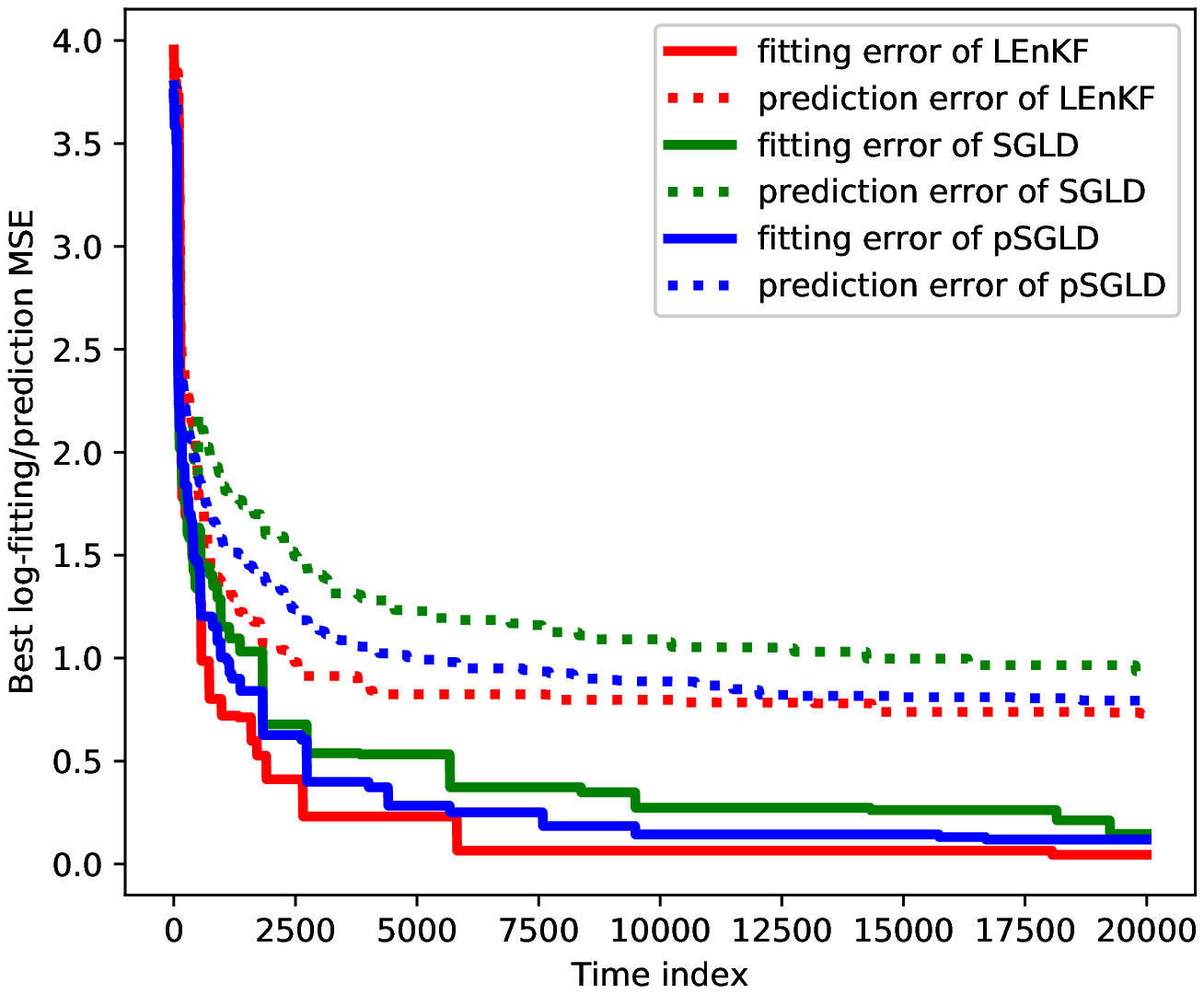} 
		\end{tabular}
		\caption{Comparison of the best fitting and prediction MSEs by the time (plotted in logarithm): (a) by each chain of SGLD, pSGLD and Langevinized EnKF; 
		(b) by ensemble averaging of SGLD, pSGLD and averaging Langevinized EnKF.}
		\label{fig:nonlincomp_parallel}
\end{figure}

Table \ref{tab:nonlincomp} summarizes the results of Langevinized EnKF for 10 datasets under different settings of $\mK$ and $\alpha$, along with comparisons with SGLD and pSGLD. 
For each dataset, we calculated  
 ``MeanMSFE'' by averaging
 the fitting MSEs over the stages $t=15,001,\ldots, 20,000$,  and ``MeanMSPE'' by averaging the prediction MSEs over the stages $t=15,001,\ldots, 20,000$. 
Then their values were averaged over 10 datasets and denoted by 
``Ave-MeanMSFE'' and ``Ave-MeanMSPE'', respectively. The comparisons indicate that for this example, the Langevinized EnKF outperforms SGLD and pSGLD when $\mK$ and $\alpha$ are chosen appropriately, and it prefers to run with slightly large values of $\mK$ and $\alpha$. For this example, the Langevinized EnKF performed similarly with $\mK=5$ and $\mK=10$, while much better than 
 with $\mK=1$ and 2; and 
 the choice of $\alpha$ affects not much on its training error, but its prediction error: A larger value of $\alpha$ tends to lead to a smaller prediction error. 
As implied by (\ref{effeq}), running the Langevinized EnKF with a mini-batch size of $n$
is equivalent to running SGLD with a mini-batch size of 
$n\alpha/(1-\alpha)$. Indeed, as shown in 
Table \ref{tab:nonlincomp}, SGLD with $n=900$ produced  similar training and test errors as Langevinized EnKF with $n=100$, $\alpha=0.9$ and $\mK=5$.  
How the batch size affects on the prediction error is  an interesting problem in machine learning, which can
be partially explained by Theorem 1 of \cite{LiangCheng2013} which accounts for the effect of batch size on the convergence of SGD,
but deserves a further study. 


 \begin{table}[htbp]
    \caption{Comparison of Langevinized EnKF, parallel SGLD and parallel pSGLD for the nonlinear regression example, where the numbers in the parentheses denote the standard deviations of the averaged MeanMSFE and MeanMSPE values over 10 datasets. The CPU time was measured 
    on a personal computer with RAM16GB and 2.9GHz Intel Core i7. 
    }  
    \label{tab:nonlincomp}
    \centering
    \begin{tabular}{ccccccc}
         \toprule
        & $n$  & $\mK$ & $\alpha$ & Ave-MeanMSFE & Ave-MeanMSPE  & CPU(s) \\ \midrule
                 & 100      & 5 & 0.1 & 2.4749(0.11)  & 3.4482(0.16) & 4050.41 \\ 
                 &100 & 5 & 0.3 &  2.4171(0.09) & 3.3056(0.13) &  4171.32\\
                &100     & 5 & 0.5 & 2.4472(0.10)  & 3.1841(0.13) & 4082.99 \\
  & 100 & 5 & 0.8 & 2.5852(0.09)  & 3.1449(0.12) & 4040.84\\
    Langevinized EnKF & 100 & 5 & 0.9 & 2.4319(0.10) & 3.1437(0.14) &	4041.91\\ 
 & 100 & 10 & 0.9 & 2.3422(0.07) & 3.1469(0.17) & 7766.84 \\
 & 100 & 2 & 0.9 & 2.8509(0.17) & 4.4075(0.27)  & 1943.71\\
  & 100 & 1 &  0.9 & 3.6848(0.25) & 4.8561(0.22) & 1093.09\\
 \midrule
     & 100 &  & &	2.4747(0.12) &	3.3877(0.12) &	4111.09\\
\raisebox{1.5ex}{SGLD} & 900 &  & &	2.4255(0.17) &	3.1997(0.13)&	5772.94\\ \midrule
         pSGLD &100 &  & &	2.6317(0.11)&	3.3988(0.10) &	5624.10\\
         \bottomrule
    \end{tabular}
\end{table}

 In summary, this example shows that the Langevinized EnKF provides a more efficient algorithm than SGLD and pSGLD for training a Bayesian DNN, and that
 Bayesian DNN can be used in variable selection for 
 nonlinear systems by imposing a mixture Gaussian prior on each weight of the DNN.
 For simplicity, we consider only the case that the sample size is larger 
 than the DNN size, i.e., total number of connections of the fully connected DNN.  
 As shown in \cite{SunSLiang2019}, such a prior also works for the case that the sample size is much smaller than the DNN size. In this case, 
 posterior consistency and variable selection selection still hold. 

\section{Numerical Studies for Dynamic Learning Problems}

This section illustrates the performance of the Langevinized EnKF as a 
 particle filtering algorithm for dynamic problems. 
 Two examples are considered under the Bayesian framework.  
 One is uncertainty quantification for the Lorenz-96 model
 which has been considered as the benchmark example for weather forecasting. 
 The other is on-line learning with Long short-term memory (LSTM) networks. 

\subsection{Uncertainty Quantification for the Lorenz-96 Model} \label{LorenzSect}

The Lorenz-96 model was developed by Edward Lorenz in 1996
to study difficult questions regarding predictability in weather forecasting \citep{Lorenz1996}. 
 The model is given by
\[
\frac{d x^{i}}{d t}=(x^{i+1}-x^{i-2}) x^{i-1} - x^{i}+F, \quad i=1,2, \cdots, p,
\]
where $F=8$, $p=40$, and it is assumed that $x^{-1} = x^{p-1}$, $x^{0} = x^{p}$, and $x^{p+1}=x^{1}$.
Here $F$ is known as a forcing constant, and $F = 8$ is a common value known to cause chaotic behavior.
In order to generate the true state $\bX_{t}=(X^1_{t},\ldots,X^{p}_{t})$ for
 $t=1,2,\ldots,T$, we initialized $\bX_{0}$ by setting $X^i_{0}$ to $20$
for all $i$ but adding to $X^{20}_{0}$ a small perturbation of 0.1; 
we solved the differential equation using the fourth-order Runge-Kutta numerical method
with a time interval of $\Delta t=0.01$; and for each $i$ and $t$, we added to $X^i_t$ a random noise  
generated from $N(0,1)$. At each stage $t$,
data was observed for half of the state variables and masked with a standard Gaussian random noise, i.e.,
\[
y_t = H_t \bX_{t} + \epsilon_t, \quad t=1,2,\ldots,T,
\]
where $\epsilon_t \sim N(0, I_{p/2})$, and $H_t$ is a random selection matrix.
Figure \ref{fig:chaos} shows the simulated path of the 
partial state variables $(X_t^1, X_t^2, X_t^3)$ for $t=1,2,\ldots,T$, 
whose chaotic behavior indicates the challenge of the problem.     

\begin{figure}[htbp]
  \centering
  \includegraphics[width=4in]{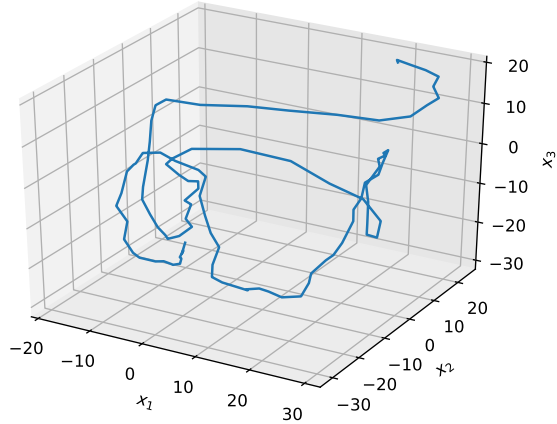}
  \caption{\quad Chaotic path of the partial state variables $(X_t^1, X_t^2, X_t^3)$ for $t=1,2,\ldots, 100$, 
     simulated from the Lorenz-96 Model.}
  \label{fig:chaos}
\end{figure}

Algorithm \ref{EnKFasimb} was applied to this example with the ensemble size $m=50$, 
the iteration number $\mK=20$, $k_0=\mK/2$,
and the learning rate $\epsilon_{t,k}=0.5/k^{0.9}$ for $k=1,2,\ldots,\mK$ and $t=1,2,\ldots,T$.
At each stage $t$, the state was estimated by averaging over the ensembles generated in iterations $k_0+1, k_0+2,\ldots, \mK$. The accuracy of the estimate was measured using the root mean-squared error (RMSE)
defined by
\[
RMSE_t =|| \hat{\bX}_{t}- {\bX}_t||_{2}/\sqrt{p},
\]
where $\hat{\bX}_t$ denote the estimate of $\bX_t$.
For comparison, the EnKF algorithm was also applied to this example
with the same ensemble size $m=50$. To be fair, it was run in a similar way to Langevinized EnKF except that the Kalman gain matrix was estimated based on the ensemble, without the resampling 
step being performed, and the random error was drawn from $N(0, V_t)$ in the analysis step. 

Figure \ref{fig:conf} compares the estimates of $X_t^1$, $X_t^2$ and $X_t^3$ produced by 
Langevinized EnKF and EnKF for one simulated dataset. 
The plots are similar for the other components $X_t^4,\ldots,X_t^{40}$. 
The comparison shows that
the EnKF and Langevinized EnKF produced comparable RMSE$_t$'s (see Figure \ref{fig:conf}~(d)). 
This is interesting as the EnKF is known to provide the optimal linear estimator of the 
conditional mean \citep{LawTt2016}. However,  
the Langevinized EnKF provided better uncertainty quantification for the estimates. For example, in Figure \ref{fig:conf}~(c), many state values  
are covered by the confidence band produced by the Langevinized EnKF, but not
by the confidence band by the EnKF.  
This is consistent with the existing result that the EnKF scheme is known to underestimate 
the confidence intervals, see e.g. \cite{SaetromO2013}.

\begin{figure}[htbp]
\begin{tabular}{cc}
(a) & (b) \\
\epsfig{figure=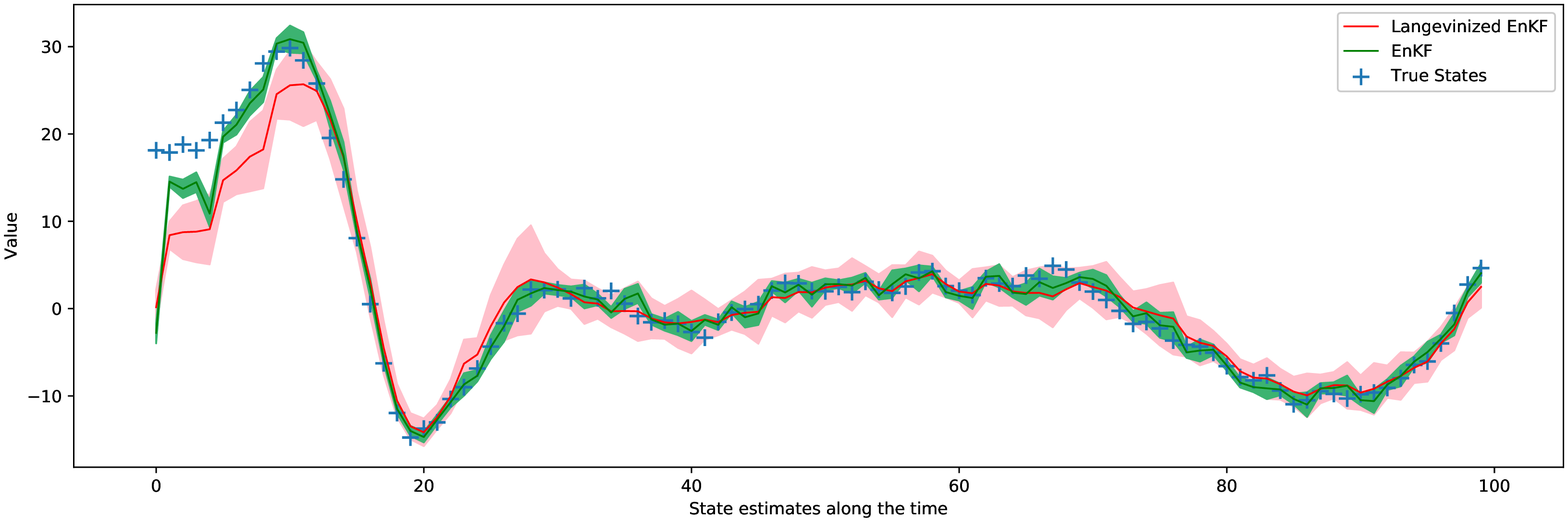,height=1.75in,width=3.0in,angle=0} &
\epsfig{figure=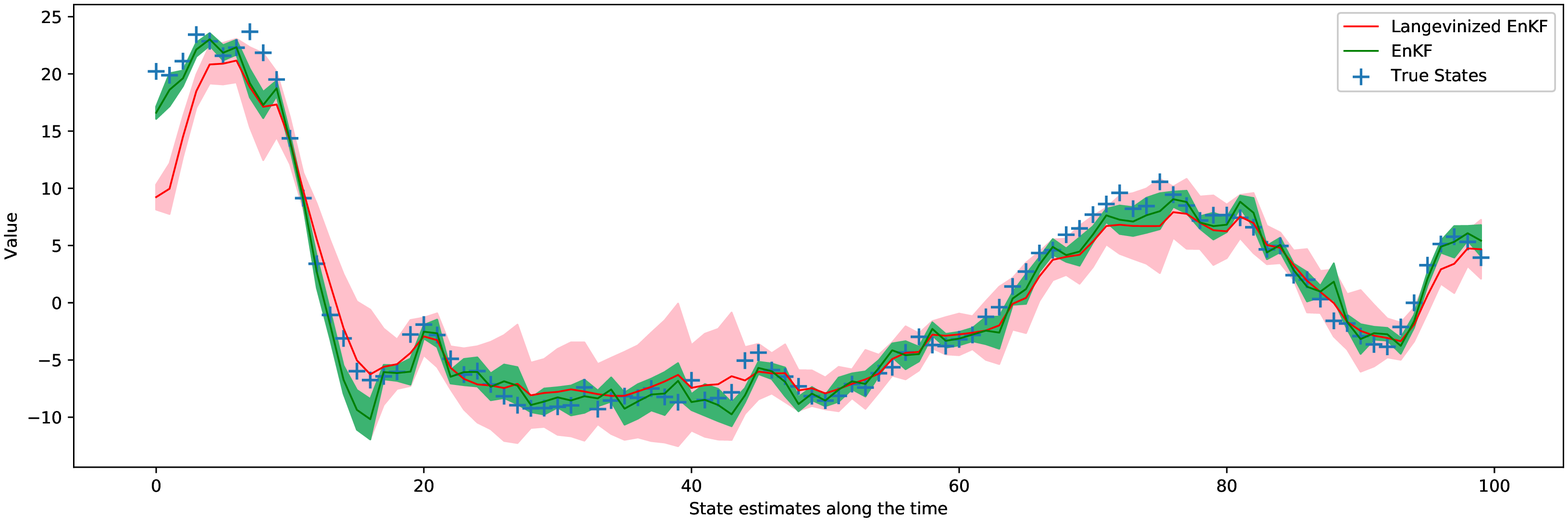,height=1.75in,width=3.0in,angle=0}  \\
(c) & (d) \\
\epsfig{figure=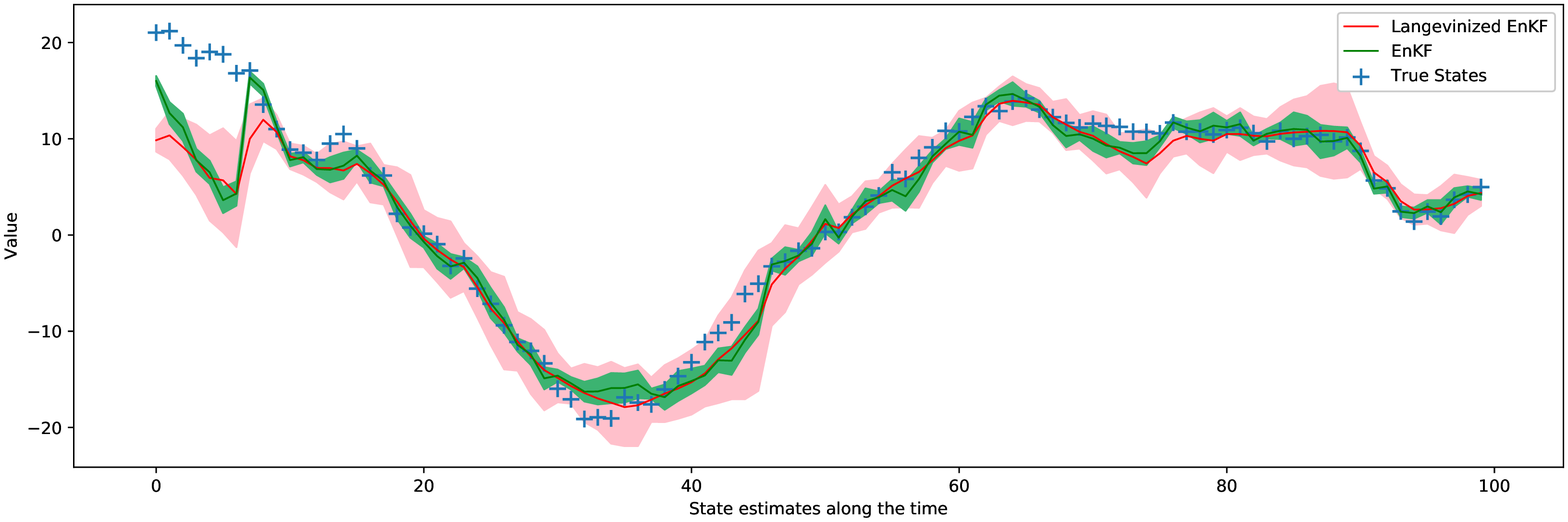,height=1.75in,width=3.0in,angle=0} & 
\epsfig{figure=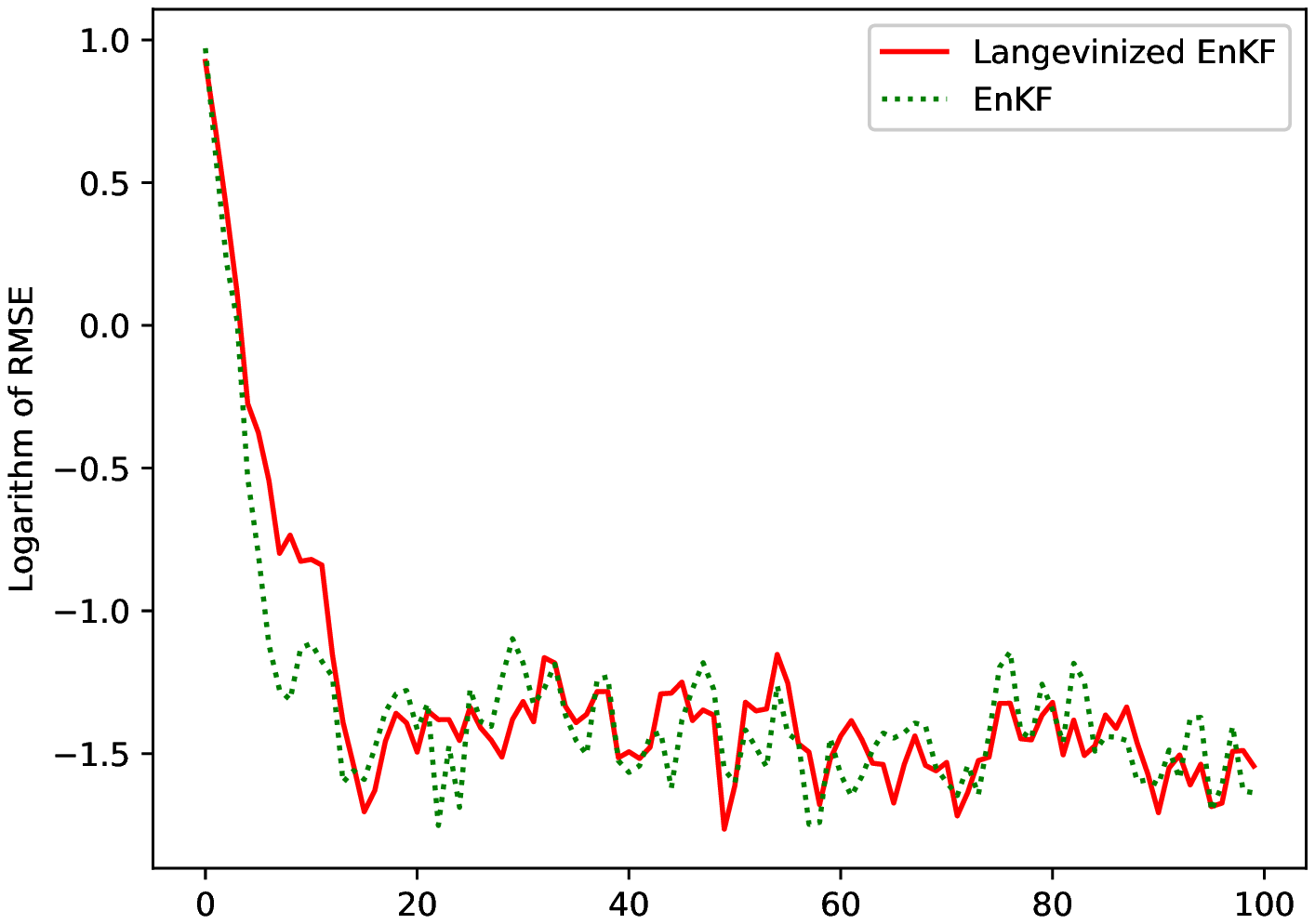,height=1.75in,width=3.0in,angle=0} \\
\end{tabular}
\caption{State estimates produced by the EnKF and Langevinized EnKF for the Lorenz-96 model 
 with $t=1,2,\ldots,100$:  
 plots (a)-(c) show, respectively, the estimates of $X_t^1$,  
 $X_t^2$ and $X_t^3$, where the true state values are 
 represented by `+', the estimates are represented by solid lines,
 and their 95\% confidence intervals are represented by shaded bands;  
 plot (d) shows $\log$(RMSE$_t$) along with stage $t$. } 
 \label{fig:conf}
\end{figure}

Figure \ref{coverageplot}~(a) shows the coverage probabilities of the 95\% confidence intervals 
produced by the EnKF and Langevinized EnKF, where the coverage probability was 
calculated by averaging over 40 state components at each stage $t \in \{1,2,\ldots,100\}$. 
Figure \ref{coverageplot}(b) shows the averaged coverage probabilities over 10 datasets. 
The comparison shows that the Langevinized EnKF produces the coverage probabilities closing to their 
nominal level, while the EnKF does not. This implies that the Langevinized EnKF is able to correctly quantify 
uncertainty of the estimates as $t$ becomes large. This is a remarkable result 
given the nonlinear and dynamic nature of the Lorenz-96 model!

\begin{figure}[htbp]
\begin{tabular}{cc}
(a) & (b) \\
\epsfig{figure=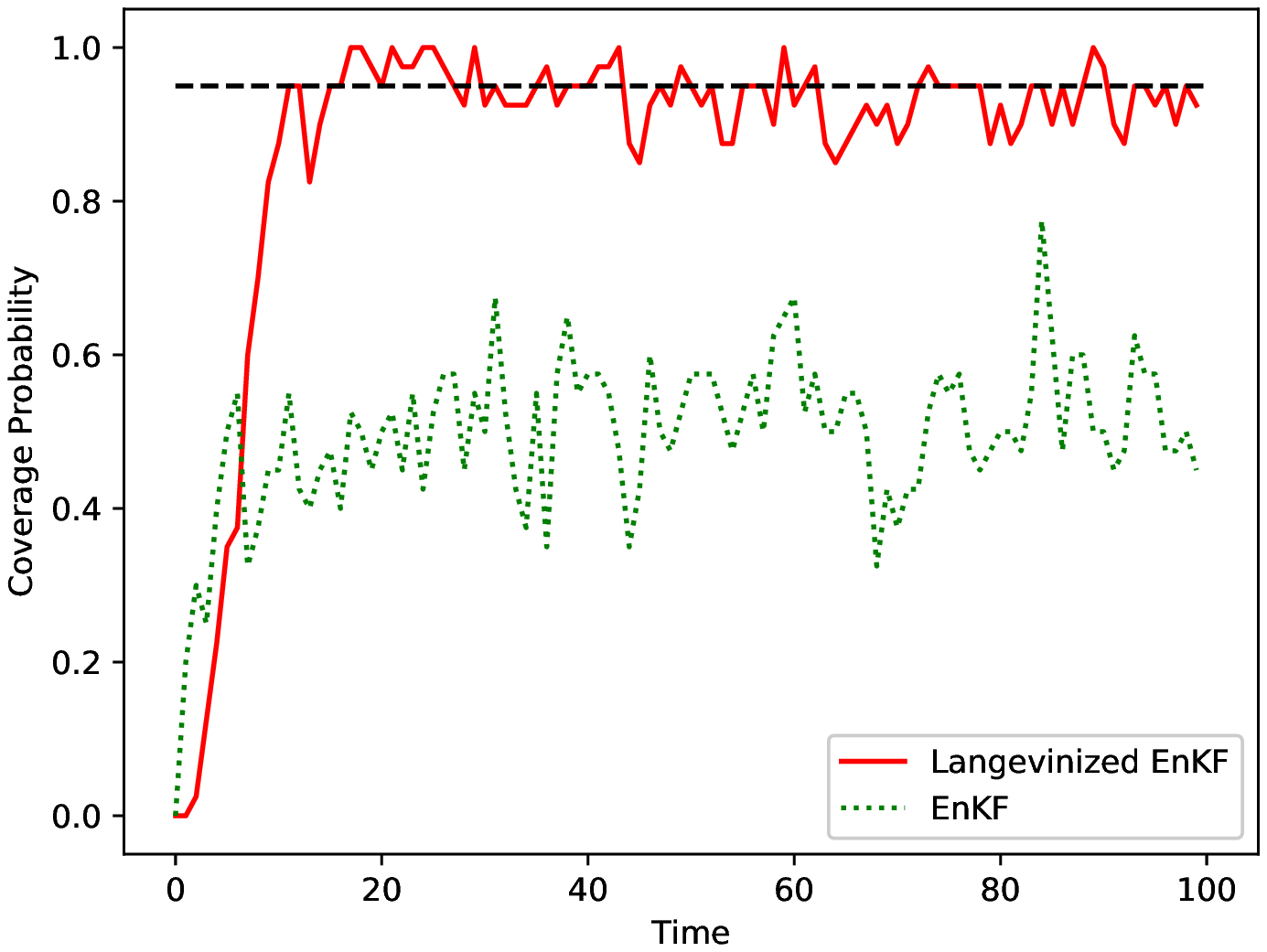,height=2.0in,width=3.25in,angle=0} &
\epsfig{figure=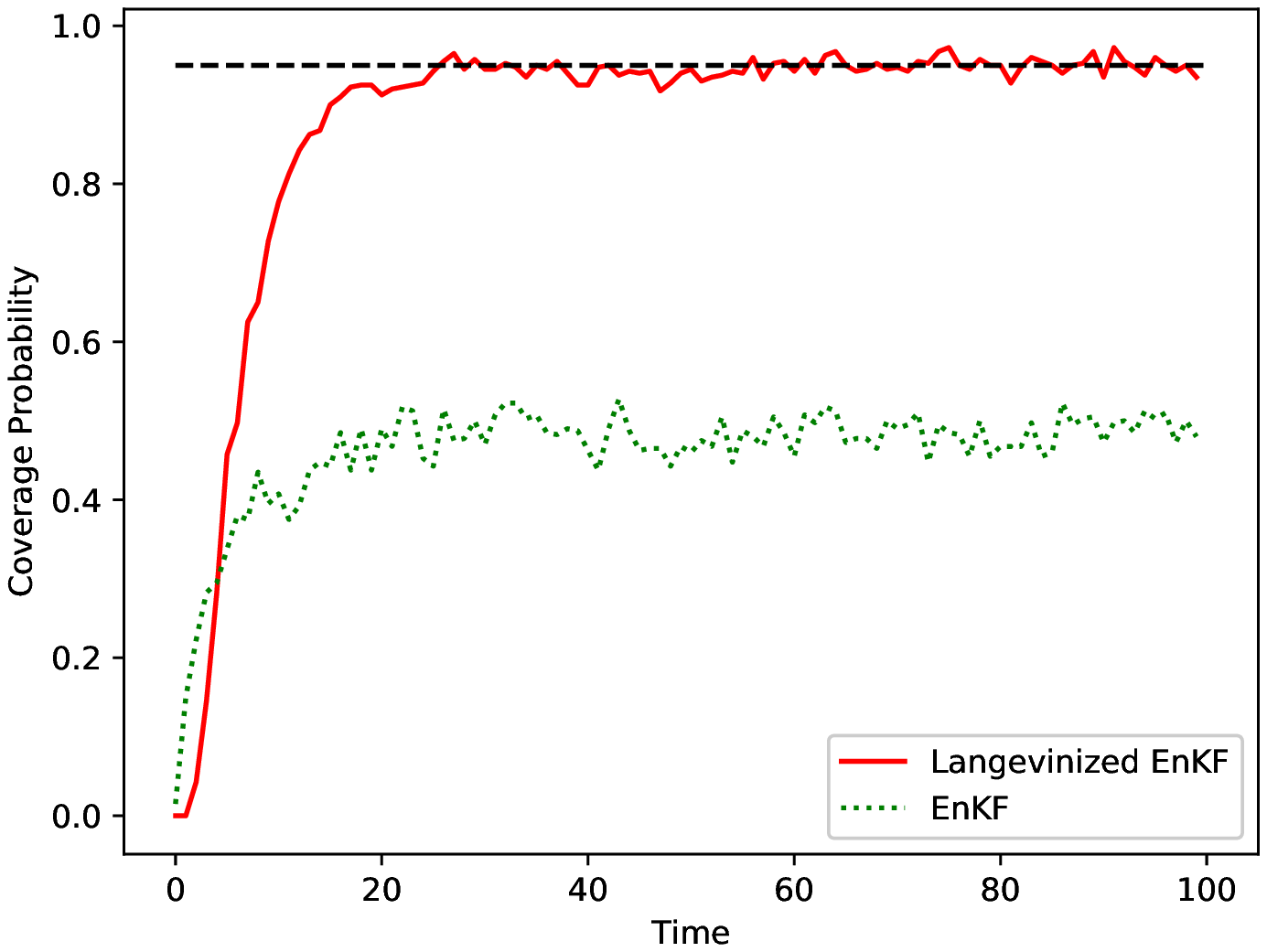,height=2.0in,width=3.25in,angle=0}  \\
\end{tabular}
\caption{Coverage probabilities of the 95\% confidence intervals 
 produced by EnKF and Langevinized EnKF for Lorenz-96 Model for stage $t=1,2,\ldots,100$:  
 (a) coverage probabilities with one dataset; (b) coverage 
  probabilities averaged over 10 datasets. }
 \label{coverageplot}
\end{figure}

Table \ref{Lorenztab} summarizes the results produced by the two methods on 10 datasets. In Langevinized EnKF, two choices of $k_0$ were tried.
For each dataset, we calculated the MeanRMSE  by averaging
RMSE$_t$ over the stages $t=21, 22,\ldots, 100$. 
Similarly, we calculated the MeanCP by averaging CP$_t$'s over the stages $t=21, 22,\ldots, 100$,
where CP$_t$ denotes the coverage probability calculated for one dataset at stage $t$. 
Then their values were averaged over 10 datasets and denoted by 
``Ave-MeanRMSE'' and ``Ave-MeanCP'', respectively.   
Table \ref{Lorenztab} also gives the CPU time cost by each method.
Compared to EnKF, Langevinized EnKF  produced slightly lower RMSEt’s, but much more accurate uncertainty quantification for the estimates. We note that Langevinized EnKF also produced very good results with $k_0=\mK-1$, which cost much less CPU time than with $k_0=\mK/2$.

\begin{table}[htbp]
\centering
\caption{Comparison of the EnKF and Langevinized EnKF, 
  where the averages over 10 independent datasets are reported with the 
  standard deviation given in the parentheses. The CPU time was recorded on a personal computer with RAM16GB and 2.9GHz Intel Core i7 for a single run of the method.} 
\label{Lorenztab}
   \begin{tabular}{ccccc}
      \toprule
           &$k_0$& Ave-MeanRMSE & Ave-MeanCP  & CPU(s) \\ \midrule
      & $\mK/2$ &  1.702(0.0343) & 0.948(0.0028) & 6.377(0.3942)\\
  \raisebox{1.5ex}{Langevinized EnKF}                      & $\mK-1$ &  1.714(0.0360) & 0.947(0.0034) & 3.350(0.0807)\\
      \midrule
      EnKF              &         &  1.722(0.0230) & 0.460(0.0029) & 0.817(0.0426) \\
      \bottomrule
    \end{tabular}
\end{table}

 In summary, the Langevinized EnKF can significantly outperform the EnKF for the Lorenz-96 model. The Langevinized EnKF can produce the same accurate state estimates as EnKF, but outperforms EnKF in uncertainty quantification. 

\subsection{Online Learning with LSTM Neural Networks}

\paragraph{Reformulation of LSTM Network} 
The LSTM network is a recurrent neural network architecture proposed by \cite{HochreiterSchmidhuber1997},
which has been widely used for machine learning tasks in dealing with time series data. 
Compared to traditional recurrent neural networks, 
hidden Markov models and other sequence learning methods, LSTM is less
sensitive to gap length of the data sequence. In addition, it is easy to train, less bothered 
by exploding and vanishing gradient problems.  
The LSTM network has been successfully used in natural language processing and handwriting recognition.
It won the ICDAR handwriting competition 2009 \citep{Graves4531750} and achieved a record 17.7\% phoneme error rate on the 
classic TIMIT natural speech dataset \citep{Graves2013SpeechRW}. 
In this section we show that the Langevinized EnKF is not only able to train LSTM networks as well as the stochastic gradient descent (SGD) method does, but also able to quantify 
uncertainty of the resulting estimates.  

Consider an autoregressive model of order $q$, denoted by AR($q$). Let $\bz_t=(z_{t-q+1},\cdots, z_{t-1}, z_t)$ denote the regression vector at stage $t$. Let $y_t=z_{t+1}\in \mathbb{R}^d$ denote the target output at stage $t$. The LSTM network with $s$ hidden neurons is defined by the following set of equations:
\begin{equation}
\begin{aligned} 
\eta_{t} &=h\left(\boldsymbol{W}^{(\eta)}\boldsymbol{z}_{t}+\boldsymbol{R}^{(\eta)} \boldsymbol{\psi}_{t-1}+\boldsymbol{b}^{(\eta)}\right), \\ 
\boldsymbol{i}_{t} &=\sigma\left(\boldsymbol{W}^{(i)}\boldsymbol{z}_{t}+\boldsymbol{R}^{(i)} \boldsymbol{\psi}_{t-1}+\boldsymbol{b}^{(i)}\right), \\ 
\boldsymbol{f}_{t} &=\sigma\left(\boldsymbol{W}^{(f)} \boldsymbol{z}_{t}+\boldsymbol{R}^{(f)} \boldsymbol{\psi}_{t-1}+\boldsymbol{b}^{(f)}\right), \\ 
\boldsymbol{c}_{t} &=\mathbf{\Lambda}_{t}^{(i)} \eta_{t}+\mathbf{\Lambda}_{t}^{(f)} \boldsymbol{c}_{t-1}, \\ \boldsymbol{o}_{t} &=\sigma\left(\boldsymbol{W}^{(o)} \boldsymbol{z}_{t}+\boldsymbol{R}^{(o)} \boldsymbol{\psi}_{t-1}+\boldsymbol{b}^{(o)}\right), \\ \boldsymbol{\psi}_{t} &=\mathbf{\Lambda}_{t}^{(o)} h\left(\boldsymbol{c}_{t}\right),
\end{aligned}
\end{equation}
where $\boldsymbol{\Lambda}_{t}^{(f)}=\operatorname{diag}\left(\boldsymbol{f}_{t}\right), \boldsymbol{\Lambda}_{t}^{(i)}=\operatorname{diag}\left(\boldsymbol{i}_{t}\right)$, 
 and  $\mathbf{\Lambda}_{t}^{(o)}=\operatorname{diag}\left(\boldsymbol{o}_{t}\right)$.
The activation function $h(\cdot)$ applies to vectors pointwisely and is
commonly set to $tanh(\cdot)$. The sigmoid function $\sigma(\cdot)$
also applies pointwisely to the vector elements. 
In terms of LSTM networks, 
$\bz_t\in \mR^{qd}$ is called input vector, $\bc_t \in \mR^s$ is called the state vector, 
$\boldsymbol{\psi}_t \in \mR^s$ is called the output vector, and  
$\bi_t$, $\boldsymbol{f}_t$ and $\boldsymbol{o}_t$ are called the input, forget and output gates, respectively. 
For the coefficient matrices and weight vectors, we have
 $\boldsymbol{W}^{(\eta)}, \boldsymbol{W}^{(i)}, \boldsymbol{W}^{(f)}, \boldsymbol{W}^{(o)} \in \mathbb{R}^{s \times qd}$, 
 $\boldsymbol{R}^{(\eta)}, \boldsymbol{R}^{(i)}, \boldsymbol{R}^{(f)}, \boldsymbol{R}^{(o)} \in \mathbb{R}^{s \times s}$, 
 and $\boldsymbol{b}^{(\eta)},\boldsymbol{b}^{(i)}, \boldsymbol{b}^{(f)}, \boldsymbol{b}^{(o)} \in \mathbb{R}^{s}$. For initialization, we set  
$\boldsymbol{\psi}_{0} =\boldsymbol{0}$, and $\boldsymbol{c}_{0} =\boldsymbol{0}$.
Given the output vector $\boldsymbol{\psi}_{t}$ , we can model the target output $y_t$ as 
\begin{equation} \label{LSTMeq2}
y_t=\boldsymbol{W} \boldsymbol{\psi}_{t} + \boldsymbol{b}+\boldsymbol{u}_t,
\end{equation}
where {$\boldsymbol{W} \in \mathbb{R}^{d \times s}, \boldsymbol{b} \in \mathbb{R}^{d}$}, 
 and $\boldsymbol{u}_t \sim N(0,\Gamma_t)$.  

For convenience, we group the parameters of the LSTM model as 
$\boldsymbol{\theta} = \{\boldsymbol{W}, \boldsymbol{b}, \boldsymbol{W}^{(\eta)}, \boldsymbol{R}^{(\eta)}, \boldsymbol{b}^{(\eta)}, \boldsymbol{W}^{(i)}$, $\boldsymbol{R}^{(i)}, \boldsymbol{b}^{(i)}, \boldsymbol{W}^{(f)}, \boldsymbol{R}^{(f)}, \boldsymbol{b}^{(f)}, \boldsymbol{W}^{(o)}, \boldsymbol{R}^{(o)}, \boldsymbol{b}^{(o)}\} \in \mathbb{R}^{n_{\theta}}$, where {$n_{\theta}= 4s^2+4sqd+4s+sd+d$}.
 With the state-augmentation approach, we can rewrite the LSTM model as a state-space model with a linear measurement equation as follows: 
\begin{equation} \label{LSTMeq3}
    \begin{aligned}\left[\begin{array}{l}{\boldsymbol{\theta}_{t}} \\ {\boldsymbol{c}_{t}} \\ {\boldsymbol{\psi}_{t}} \\ {\boldsymbol{\gamma}_{t}} \end{array}\right]&=\left[\begin{array}{c}{\boldsymbol{\theta}_{t-1}} \\
    {\boldsymbol{\Omega}\left(\boldsymbol{c}_{t-1}, \boldsymbol{z}_{t}, \boldsymbol{\psi}_{t-1}\right)} \\
    {\tau\left(\boldsymbol{c}_{t}, \boldsymbol{z}_{t}, \boldsymbol{\psi}_{t-1}\right)} \\
    {\boldsymbol{W}_{t} \boldsymbol{\psi}_{t}+\boldsymbol{b}}
    \end{array}\right] + \left[\begin{array}{c}{\boldsymbol{e}_{t}}\\
    {\boldsymbol{\zeta}_{t}} \\
    {\boldsymbol{\xi}_{t}} \\ {\bvarepsilon_{t}} \end{array}\right], \\ 
    y_{t} &=\boldsymbol{\gamma}_{t}+\bv_{t}, \end{aligned}   
\end{equation}
 where  $\bvarepsilon_{t}\sim N(0,\alpha \Gamma_t)$
for some constant $0<\alpha_t<1$,  $\bv_t \sim N(0,(1-\alpha) \Gamma_t)$, and 
$\Gamma_t$ is as defined in (\ref{LSTMeq2}).  
Let $\bx_t^T=(\btheta_t^T,\boldsymbol{c}_t^T, \boldsymbol{\psi}_t^T, \boldsymbol{\gamma}_t^T)$.
Then  
\[
 \pi(\bx_t|\bx_{t-1},\boldsymbol{z}_t)=\pi(\btheta_t|\btheta_{t-1},\boldsymbol{z}_t)
 \pi(\boldsymbol{c}_t|\btheta_t, \boldsymbol{c}_{t-1}, \boldsymbol{\psi}_{t-1} ,\boldsymbol{z}_t) 
 \pi(\boldsymbol{\psi}_t|\btheta_t, \boldsymbol{c}_{t}, \boldsymbol{\psi}_{t-1} ,\boldsymbol{z}_t)
 \pi(\boldsymbol{\gamma}_t|\btheta_t, \boldsymbol{\psi}_{t}). 
 \]  
 As in (\ref{augeqnon}), we can rewrite the state-space model (\ref{LSTMeq3}) as a dynamic system at each stage $t$: 
 \begin{equation} \label{augeqnon2}
 \begin{split} 
  \bx_{t,k} &= \bx_{t,k-1}+ \frac{\epsilon_t}{2} \nabla_{\bx} \log  \pi(\bx_{t,k-1}|\bx_{t-1}^s, \boldsymbol{z}_t) +  \bomega_{t,k} \\ 
  y_{t,k} &= H_t \bx_{t,k}+\bv_{t,k}, \\ 
 \end{split}
 \end{equation}
 where $\bx_{t,k}$ denote an estimate of $\bx_t$ obtained at iteration $k$ for $k=1,2,\ldots, \mK$, $y_{t,k}=y_t$ for $k=1,2,\ldots,\mK$, $H_t=(0,I)$ 
 such that $H_t x_{t}=\boldsymbol{\gamma}_{t}$, 
 $\bomega_{t,k} \sim N(0,\epsilon_{t} I_{p})$, $p$ is the dimension of $\bx_t$, 
 and $\bv_{t,k} \sim N(0, (1-\alpha) \Gamma_t)$.  With this formulation, 
 Algorithm \ref{EnKFasimb} can be applied to train the LSTM model and 
 and quantify uncertainty of the resulting estimate.
 
 \paragraph{Wind Stress Data}

We considered the wind stress dataset, which can be downloaded at https:// iridl.ldeo.columbia.edu. The dataset
consists of gridded (at a $2\times 2$ degrees resolution and corresponding to $d=1470$ spatial locations) 
monthly summaries of meridional wind pseudo-stress collected from Jan 1961 to Feb 2002. 
For this dataset, we set $q=6$ and $T=300$, i.e., modeling the data of the first 300 months using an AR(6) LSTM model. The data was scaled into the range $(-1,1)$ in preprocessing and 
then scaled back to the original range in results reporting.   

 The Langevinized EnKF was first applied to this example. For the model part, we set 
 $\boldsymbol{e}_{t} \sim N(0,0.0001 I)$, $\boldsymbol{\zeta}_{t}\sim N(0,0.0001 I)$, 
 $\boldsymbol{\xi}_{t}\sim N(0,0.0001 I )$, $\boldsymbol{u}_{t}\sim N(0,0.0001 I)$.  
 These model parameters are assumed to be known, although this is not necessary as 
 discussed at the end of the paper. For this example, we have tried different settings 
 for the model parameters. In general, a smaller variance setting will lead to a better fitting to the observations. 
 For the algorithmic part, 
 we set the ensemble size $m = 100$, $\mK = 10$, 
 $k_0=9$, 
 $\alpha = 0.9$, the number of hidden 
 neurons $s= 20$,  and 
 the learning rate $\epsilon_{t, k} = 0.0001/\max\{\kappa_b,k\}^{0.95}$  with $\kappa_b=8$ for $k = 1,\cdots, \mK$ 
 and  $t = 1,\cdots, T$.  At each stage $t$, the wind stress was estimated by averaging over $\hat{y}_{t,k}=H_t \bx_{t,k}$ for 
 last $\mK/2$ iterations. In addition, the credible interval for each component of $\bx_t$ was calculated 
 based on the ensemble obtained at stage $t$.  Each run cost about 5334.5 CPU seconds. The results are summarized in 
 Figure \ref{fig:windpath}, where the wind stress estimates at four selected 
 spatial locations and their 95\% credible intervals are plotted along with stages. 
 
 For comparison, SGD was also applied to this example with the same setting as the Langevinized EnKF, 
 i.e., they share the same learning rate and the same iteration number $\mK=10$ at each stage.
 The results are also summarized in Figure \ref{fig:windpath}, where the wind stress estimates at four 
 selected spatial locations are plotted along with stages. Each run of SGD cost 
 about 15.9 CPU seconds. Since the Langevinized EnKF had an ensemble size $m=100$,  each 
 chain cost only 53.3 CPU seconds. 
 The Langevinized EnKF cost more CPU time and as return, it produced more samples for uncertainty quantification. 
 
\begin{figure}[htbp]
\begin{tabular}{c} 
  \includegraphics[width=\textwidth]{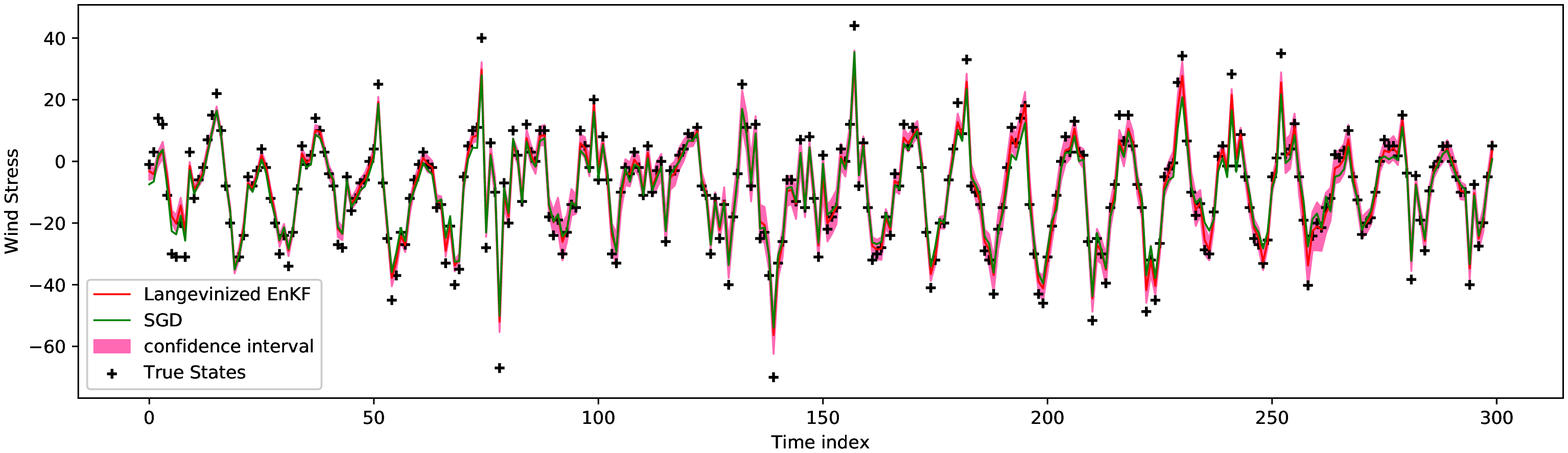} \\
  \includegraphics[width=\textwidth]{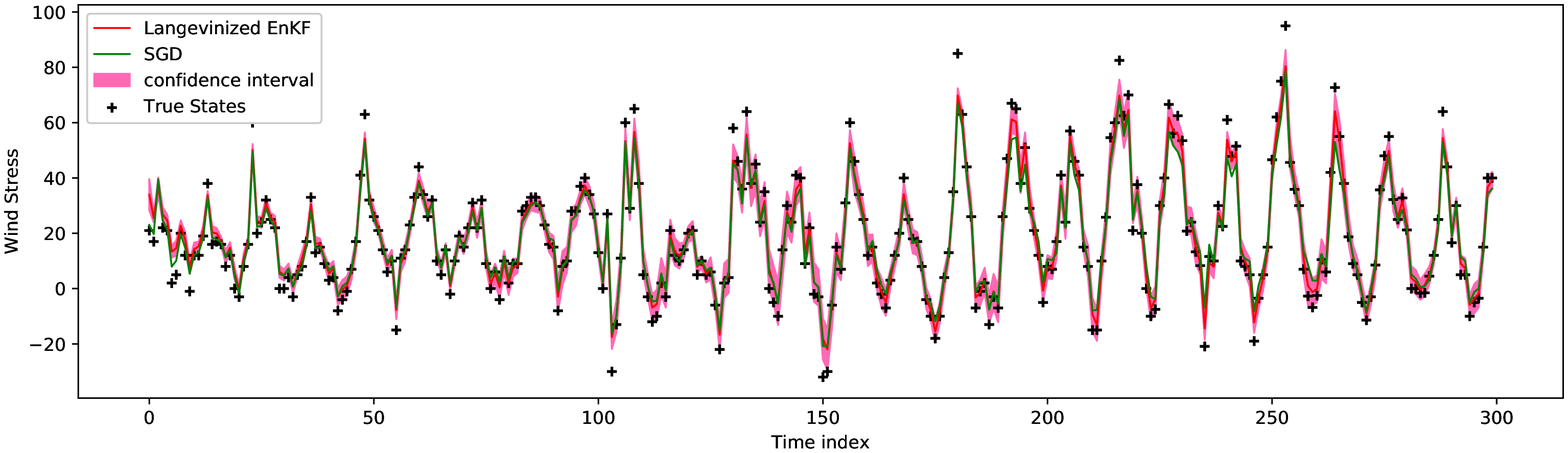} \\
  \includegraphics[width=\textwidth]{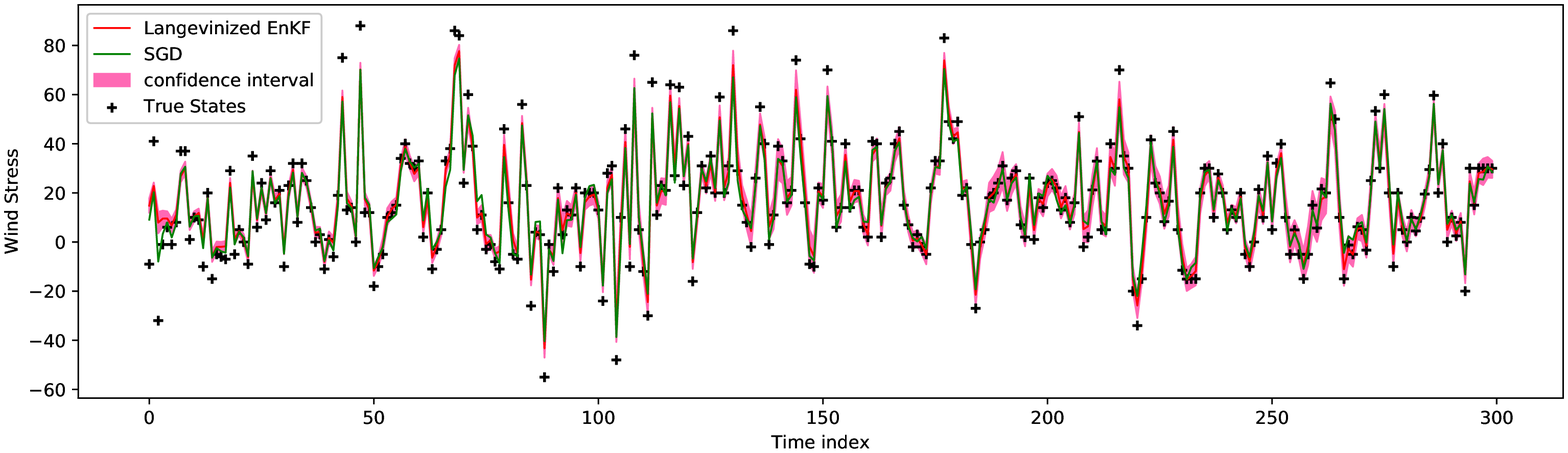}  \\
 \includegraphics[width=\textwidth]{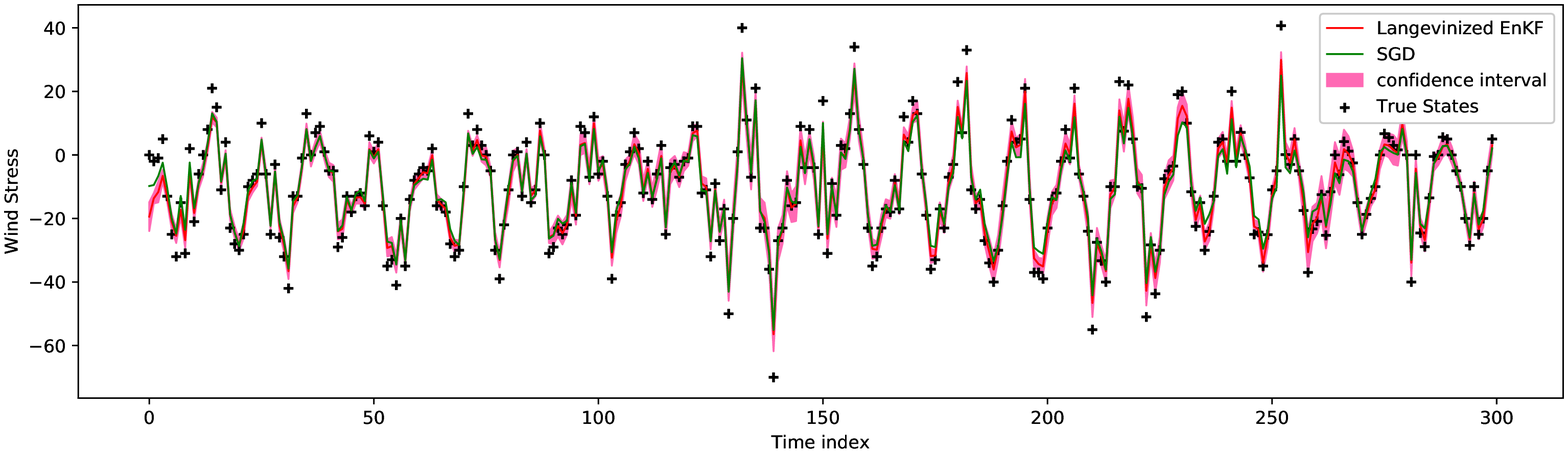}
 \end{tabular}
 \caption{ Wind stress estimates at four spatial locations and their 95\% credible interval along 
  with stages: the red line is for the Langevinized EnKF estimate; the pink shaded band is 
  for credible intervals of the Langevinized EnKF,  
  the green line is for the SGD estimate; and the blue cross '+' is for the true wind stress value.}
 \label{fig:windpath}
 \end{figure}

\begin{figure}[htbp]
        \centering
        \includegraphics[width=0.75\textwidth]{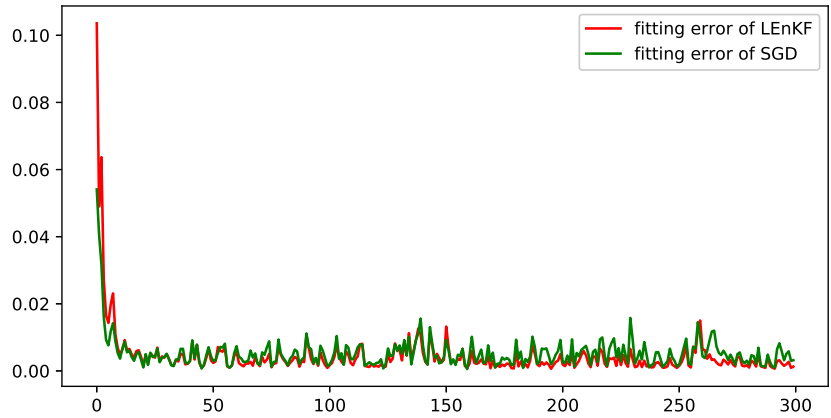}
		\caption{\quad Comparison of the mean squared fitting errors produced by SGD, Langevinized EnKF and   ensemble averaging Langevinized EnKF. }
		\label{fig:windloss}
\end{figure}

\begin{figure}[htbp]
		\centering
		\begin{subfigure}[b]{0.5\textwidth}
        \includegraphics[width=\textwidth]{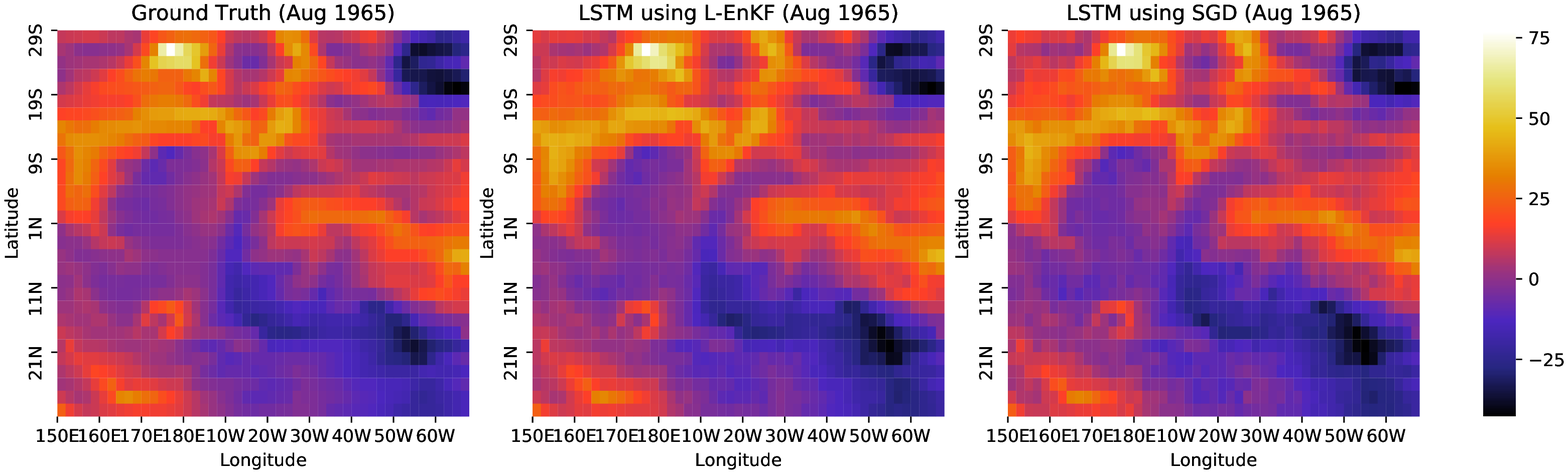}
        \end{subfigure}%
        ~
        \begin{subfigure}[b]{0.5\textwidth}
        \includegraphics[width=\textwidth]{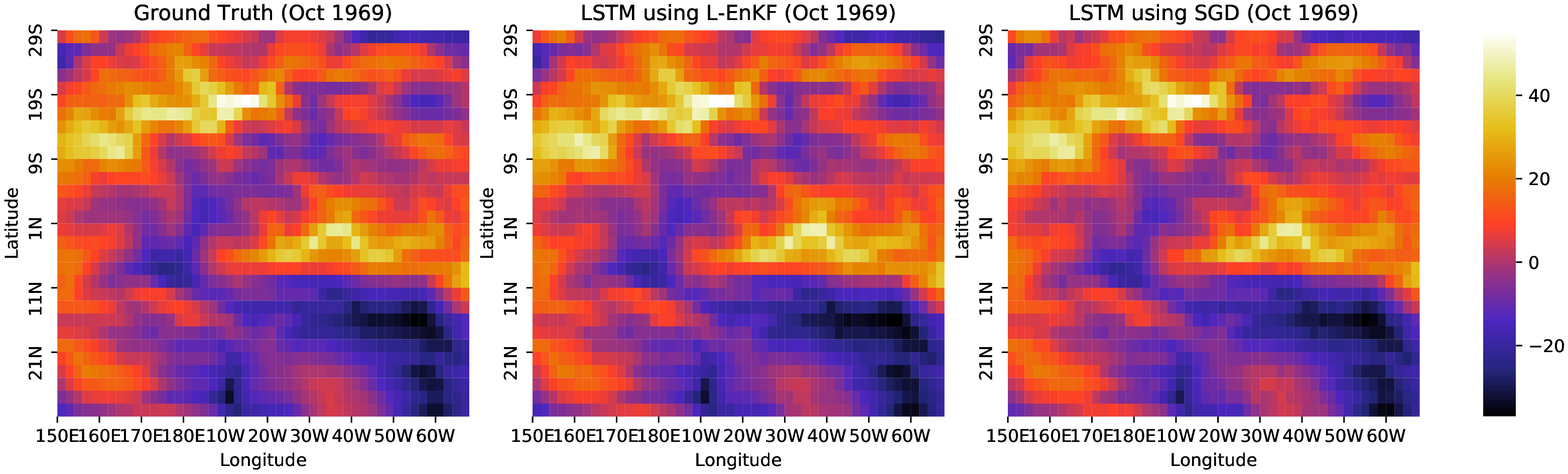}
        \end{subfigure}%
        
        \begin{subfigure}[b]{0.5\textwidth}
        \includegraphics[width=\textwidth]{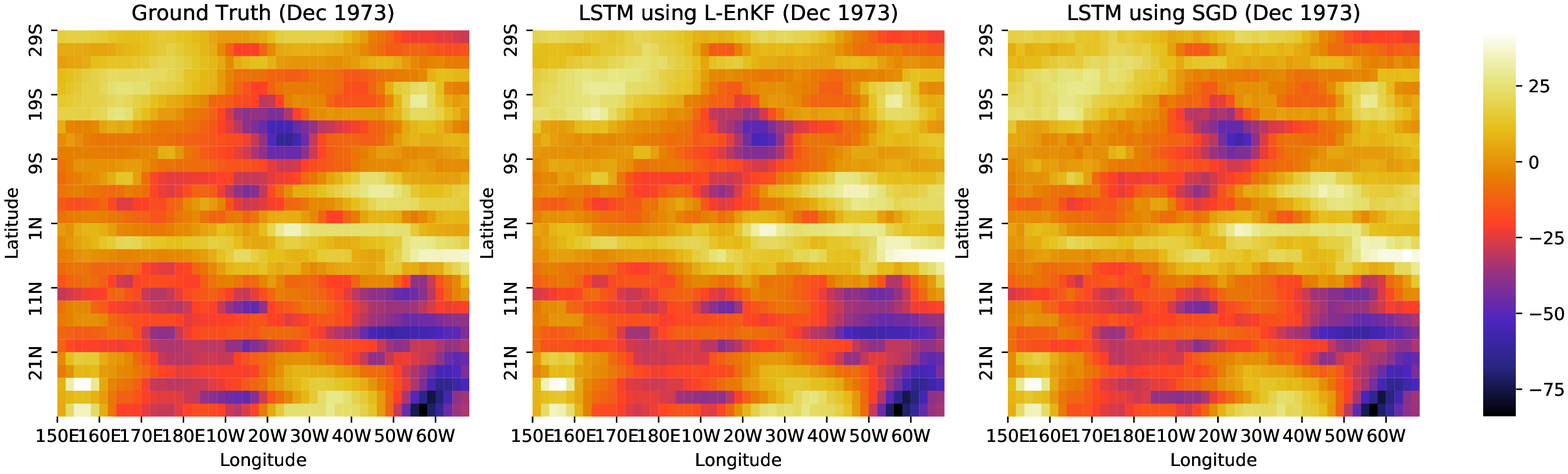}
        \end{subfigure}%
        ~
        \begin{subfigure}[b]{0.5\textwidth}
        \includegraphics[width=\textwidth]{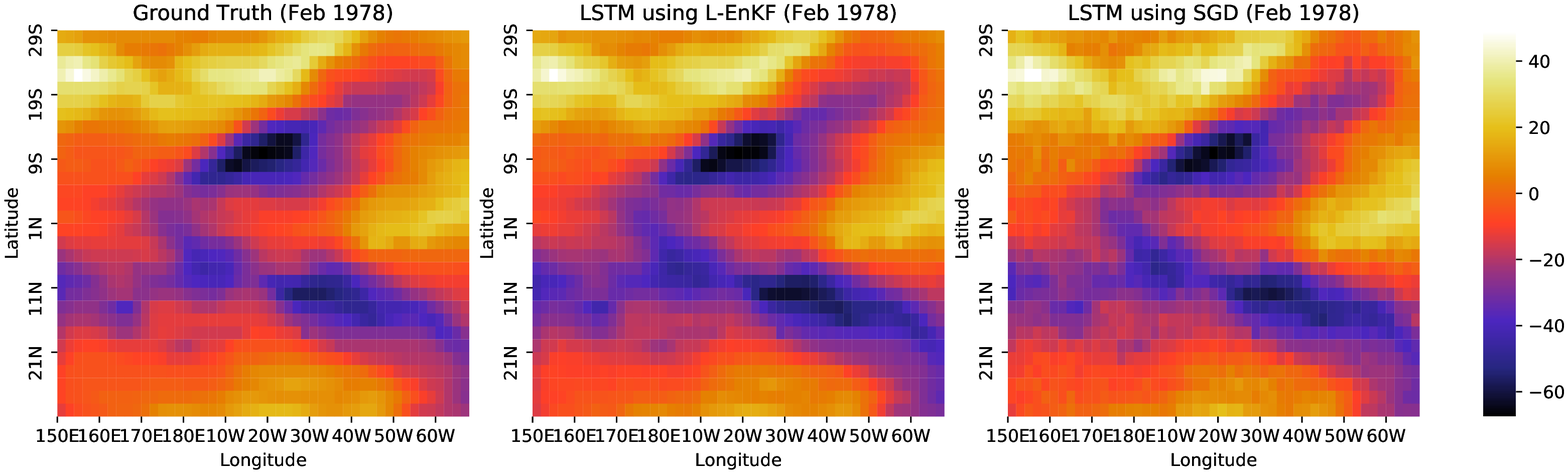}
        \end{subfigure}%
        
        \begin{subfigure}[b]{0.5\textwidth}
        \includegraphics[width=\textwidth]{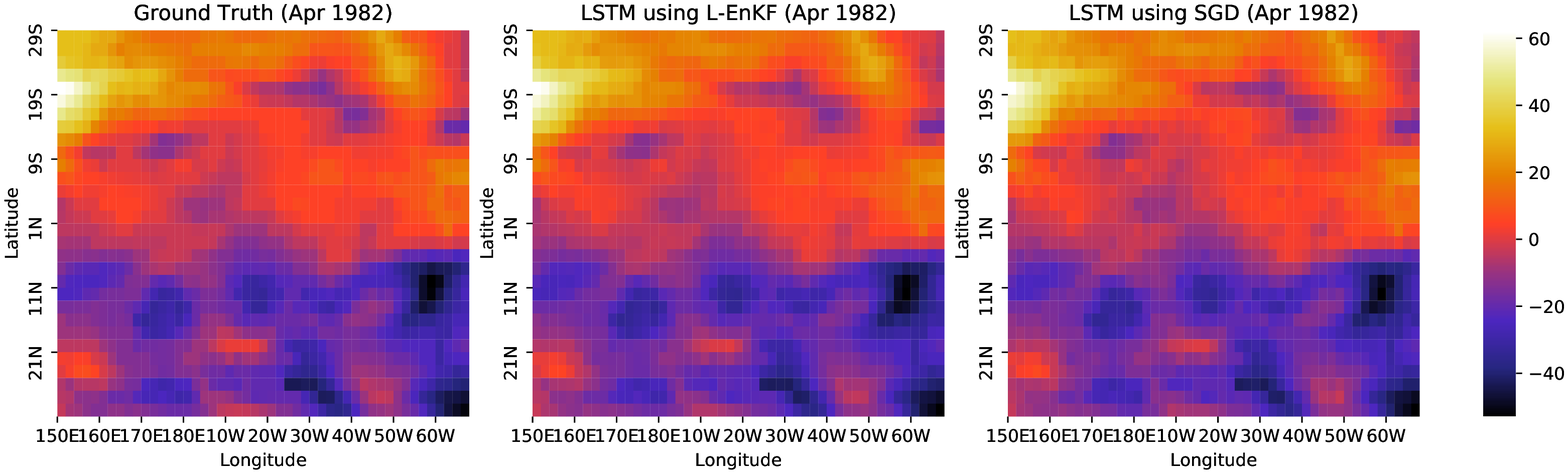}
        \end{subfigure}%
        ~
        \begin{subfigure}[b]{0.5\textwidth}
        \includegraphics[width=\textwidth]{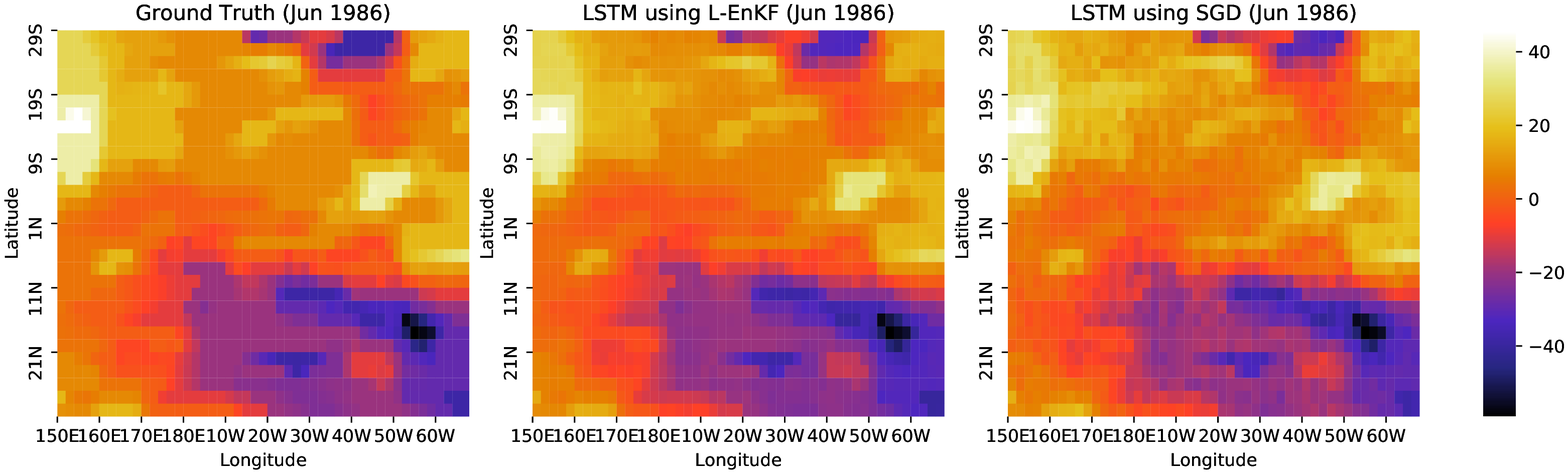}
        \end{subfigure}%
		\caption{Heat maps of the wind stress fitted by the Langevinized EnKF and SGD 
for six different months, August 1965, October 1969, December 1973, February 1978, April 1982, 
and June 1986: For both left and right panels, the left, middle and right columns show the true heat map, the heat map fitted by Langevinized EnKF, and the heat map fitted by SGD, respectively. }
		\label{fig:heatfit}
\end{figure}

Further, we calculated the mean squared fitting error $||\hat{y}_t -y_t||_2^2$
for stages $t=1,2,\ldots, T$ and for both methods. The results are summarized in Figure \ref{fig:windloss}, which indicates that
the Langevinized EnKF produced slightly smaller fitting errors than SGD.
Figure \ref{fig:heatfit} shows the heat maps of the wind stress fitted by the Langevinized EnKF and SGD
for six different months, August 1965, October 1969, December 1973, February 1978, April 1982,
and June 1986. The comparison with the true heat maps indicates that both SGD and Langevinized EnKF
can train the LSTM very well for this example.

In summary, this example shows that the Langevinized EnKF is not only able to train LSTM networks as does SGD, but also able to quantify uncertainty of the resulting estimates.

\section{Discussion}

This paper proposes a new particle filtering algorithm, Langevinized EnKF, by reformulating the EnKF under the framework of Langevin dynamics. The Langevinized EnKF converges to the right filtering distribution in Wasserstein distance under the big data scenario that the dynamic systems consists of a large number of stages and has a large number of observations at each stage.
The Langevinized EnKF  inherits the forecast-analysis procedure from the EnKF and the use of mini-batch data from the stochastic gradient 
MCMC algorithms, enabling it scalable with respect to both the dimension and the sample size.  
The Langevinized EnKF can be applied to both the inverse 
and data assimilation problems. 


Like the EnKF, the Langevinized EnKF is developed under 
the assumption that the state space model 
contains no unknown parameters. Extension of the algorithm to the case with unknown parameters can be done in different ways. One way is with the state augmentation approach. Taking 
the model (\ref{inveq1}) as an example: If the function $\pi(\cdot)$ or more general, the state propagator, 
 contains some unknown parameters, we can augment the state $x_t$ to include these parameters. Another way is with the EM algorithm \citep{DempsterEM1977}. Since the Langevinized EnKF is able to sample from the filtering distribution for given parameters, the EM algorithm can be conveniently used for estimating parameters. A related work is \cite{Aicher2019}, however, where the filtering distribution is simulated using traditional sequential Monte Carlo. Third,
 the parameters can also be estimated using an adaptive SGMCMC algorithm. Taking the model (\ref{inveq1}) as the example again: If the propagator $H_t$ or the observation noise covariance matrix contain some unknown parameters, then the parameters can be estimated in a recursive way. In this case, if the state-augmentation approach is used, the covariance 
 matrix $\Sigma_t= \frac{n}{N} (I-K_tH_t)$ (defined in equation (S3) of the Supplementary Material) will no longer be a constant matrix of the state variable 
 and thus the convergence to the correct filtering distribution will not be guaranteed any more. To tackle this issue,  
 a parameter updating step can be added to the Langevinized EnKF Algorithm \ref{EnKFnew} as follows:

 \begin{itemize}
 \item[3.] {\it (Parameter updating) Update the parameters according to the recursion}
 \[
  \bvartheta_{t}=(1-a_t) \bvartheta_{t-1}+a_t \phi(\bvartheta_{t-1},\bx_t^{a}),
 \]
 {\it  where $\bvartheta$ denotes the vector of unknown parameters, $\{a_t\}$ is a pre-specified, positive, decreasing sequence
 satisfying the conditions $\sum_t a_t =\infty$ and $\sum_t a_t^2 <\infty$,
 $\bx_t^a=(x_t^{a_1}, \ldots, x_t^{a_m})$ denotes the ensemble of samples at stage $t$, 
 and $\phi(\bvartheta_{t-1},\bx_t^{a})$ is
 a mapping driving to an estimate of $\bvartheta$ based on the ensemble $\bx_t^a$.}
 \end{itemize}

 Based on the theory of stochastic approximation \citep{RobbinsM1951},
 the mapping $\phi(\bvartheta_{t-1},\bx_t^{a})$ can be easily designed. 
 With the parameter updating step, the Langevinized EnKF will become an adaptive
 SGMCMC algorithm, where the target distribution varies from iteration to iteration. 
 The convergence of the adaptive SGMCMC algorithm has been studied in \cite{DengZLL2019}. 
 The theory of \cite{DengZLL2019} implies that 
 as $t\to \infty$, $\bvartheta_t$ converges to the true parameter vector
 in probability and the Langevinized EnKF converges
 to the correct posterior distribution.
 The above three parameter estimation methods also apply to  Algorithm \ref{EnKFnonlinear} and Algorithm \ref{EnKFasimb}.

 In practice, we often encounter the problems where the response variable follows a non-Gaussian distribution, e.g., multinomial or Poisson. The Langevinized EnKF 
 can be extended to these problems by
 introducing a latent variable. For example, consider an inverse problem, for which the latent variable model can be formulated as
 \begin{equation} \label{NasimExt}
 z  | y \sim \rho(z|y), \quad   y  =h(x)+\eta, \quad \eta \sim N(0, \Gamma),  
 \end{equation}
 where  $z$ is observed data following a non-Gaussian distribution $\rho(\cdot)$, 
 $y$ is the latent Gaussian variable, and $x$ is the parameter.
 To adapt the Langevinized EnKF to simulating from the posterior of $x$, we only need to add an
 imputation step in Algorithm \ref{EnKFnew}, between the forecast step and the analysis step.
 The imputation step is to simulate a latent vector $y$ from the distribution
 $\pi(y|x,z) \propto \rho(z|y) f(y|x)$. Since the imputation will lead to an
 unbiased estimate for the gradient of the involved log-density function,
 the proposed extension is valid, with which samples from the target posterior $\pi(x|z)$ can be generated.   
 A further extension of this algorithm to data assimilation problems is straightforward.



\appendix

\section*{Appendix}




\setcounter{table}{0}
\renewcommand{\thetable}{S\arabic{table}}
\setcounter{figure}{0}
\renewcommand{\thefigure}{S\arabic{figure}}
\setcounter{equation}{0}
\renewcommand{\theequation}{S\arabic{equation}}
\setcounter{algorithm}{0}
\renewcommand{\thealgorithm}{S\arabic{algorithm}}
\setcounter{lemma}{0}
\renewcommand{\thelemma}{S\arabic{lemma}}
\setcounter{corollary}{0}
\renewcommand{\thecorollary}{S\arabic{corollary}}
\setcounter{theorem}{0}
\renewcommand{\thetheorem}{S\arabic{theorem}}
\setcounter{remark}{0}
\renewcommand{\theremark}{S\arabic{remark}}


This supplementary material is organized as follows. Section A compares the Langevinized EnKF with parallel SGLD, pSGLD and SGNHT for a linear regression example. Section B proves the convergence of Algorithm 2 and Algorithm 4.

\section{Comparison of Langevinized EnKF with parallel SGLD, pSGLD and SGNHT}

For a thorough comparison, we also ran SGLD \citep{Welling2011BayesianLV}, pSGLD \citep{Li2016PSGLD} and SGNHT \citep{Ding2014BayesianSU} in parallel for the linear regression example considered in Section 4.1.
Each of the three algorithms was run for 1,000 iterations with 100 chains and exactly the same parameter setting as used in Section 4.1. 
The CPU times costed by SGLD, pSGLD, and SGNHT were 
38.18, 43.19, 42.39 CPU seconds, respectively. Recall that
the Langevinized EnKF with an ensemble size of $m=100$
cost  35.14 CPU seconds for 1,000 iterations on the same
computer. Figure \ref{para_linear} shows the trajectories of $(\beta_1,\beta_2,\ldots,\beta_9)$ produced by SGLD, pSGLD, SGNHT and Langevinized EnKF in the runs, where each trajectory was obtained by averaging over 100 chains at each iteration. 
For this example, Langevinized EnKF took less than 100 iterations to converge to the true values, SGNHT took about 1000 iterations, while SGLD and pSGLD failed to 
converge with 1000 iterations. 
 
\begin{figure}[htbp]
\begin{center}
\begin{tabular}{c}
\epsfig{figure=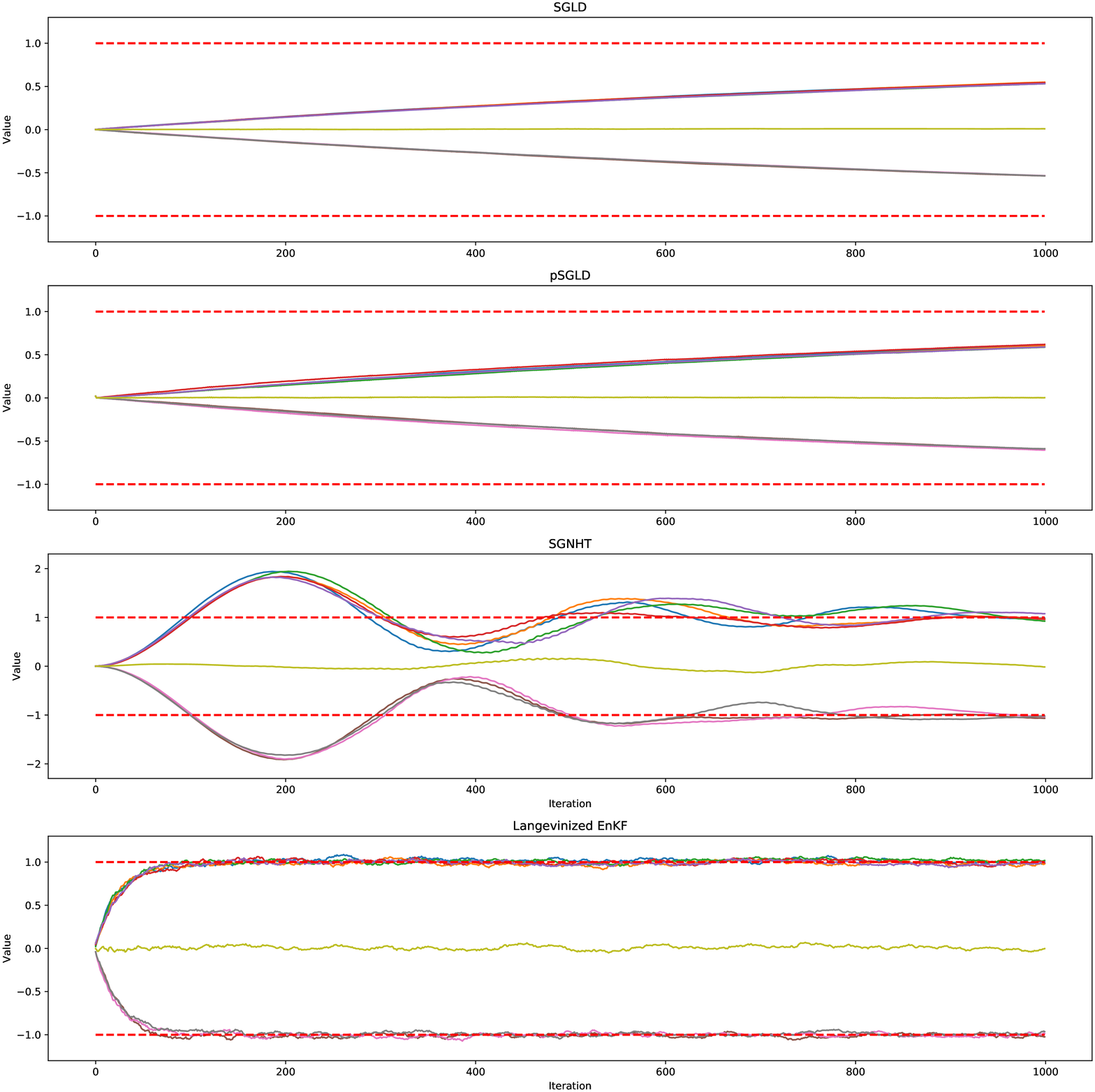,height=4.0in,width=6.0in,angle=0} 
\end{tabular}
\end{center}
 \caption{Convergence trajectories of SGLD, pSGLD, SGNHT and Langevinized EnKF for a large-scale linear regression example.
} 
\label{para_linear}
\end{figure}

\section{Proofs for the Convergence of Algorithm 2 and Algorithm 4}

We are interested in studying the convergence of the Langevinized EnKF in Wasserstein distance.
The $r$-th order Wasserstein distance between two probability measures $\mu$ and $\nu$ is defined by 
\[
W_r(\mu,\nu)=\left( \inf_{\pi\in \Pi(\mu,\nu)} \int_{\BX \times \BX} d(x,y)^r d\pi(x,y) \right)^{1/r}= 
 \inf_{\pi\in \Pi(\mu,\nu)} \left\{E_{\pi} d(X,Y)^r \right\}^{1/r}, 
\]
where $\Pi(\mu,\nu)$ denotes the collection of all probability measures on $\BX \times \BX$ with 
marginals $\mu$ and $\nu$ respectively.

\subsection{Convergence of Algorithm 2} 
 
 \begin{theorem} \label{them1}  Let $\Sigma_t=\frac{n}{N}(I-K_t H_t)$. 
   If $\lambda_l \leq \inf_t \lambda_{\min}(\Sigma_t) \leq \sup_t \lambda_{\max}(\Sigma_t) \leq \lambda_u$ holds for some positive constants $\lambda_l$ and $\lambda_u$, the log-prior density function $\log\pi(x)$ is differentiable with respect to $x$,
    and the learning rate $\epsilon_t=O(t^{-\varpi})$ for some $0<\varpi<1$, 
 then $\lim_{t\to \infty} W_2(\tilde{\pi}_t,\pi_*)=0$, where $\tilde{\pi}_t$ denotes the empirical distribution of $x_t^a$, $\pi_*=\pi(x|y)$ denotes the target posterior distribution, and $W_2(\cdot,\cdot)$ denotes the second-order Wasserstein distance.
 \end{theorem}
\begin{proof} 
  First, we consider the Kalman gain matrix $K_t=Q_t H_t^T(R_t+H_tQH_t^T)^{-1}$, which, with some algebra, can be shown
 \begin{equation} \label{Keq}
  K_t=(I-K_t H_t) Q_t H_t^T R_t^{-1}=(H_t^T R_t^{-1} H_t+Q_t^{-1})^{-1} H_t^T R_t^{-1}. 
 \end{equation}
 Let $\mu_t=E(x_t^f|x_{t-1}^a)=x_{t-1}^a+\delta_t$, where
 $\delta_t= \epsilon_t \frac{n}{2N} \nabla \log\pi(x_{t-1}^a)$.
 Therefore, $x_t^f=\mu_t+w_t$.

 Taking conditional expectation on both sides of  equation (7) of the main text, 
 we have
 \begin{equation} \label{newtoneq2}
  E(x_t^a|x_{t-1}^a) =\mu_t+ K_t (y_t-H_t \mu_t) 
         = x_t^f+ K_t(y_t-H_t x_t^f)-(I-K_tH_t)w_t. 
 \end{equation}
 With the identity (\ref{Keq}), (\ref{newtoneq2}) can be further written as
 \begin{equation} \label{newtoneq3}
 \begin{split}
  E(x_t^a|x_{t-1}^a) &={x}_{t-1}^a+\delta_t+K_t(y_t-H_t {x}_{t-1}^a-H_t\delta_t) 
   = {x}_{t-1}^a + K_t(y_t- H_t {x}_{t-1}^a)+(I-K_tH_t) \delta_t \\
        &=  {x}_{t-1}^a + (I-K_t H_t)Q_t H_t^T R_t^{-1}(y_t-H_t {x}_{t-1}^a)+(I-K_tH_t) \delta_t \\
        &=  {x}_{t-1}^a + (I-K_t H_t)Q_t \left[ H_t^T R_t^{-1}(y_t- H_t {x}_{t-1}^a) + Q_t^{-1} \delta_t \right] \\
        &=   {x}_{t-1}^a + \frac{n}{2N} (I-K_tH_t)Q_t \left[ \frac{N}{n}H_t^T V_t^{-1}(y_t-H_t {x}_{t-1}^a) + 
               \nabla \log \pi(x_{t-1}^a) \right], \\
        &= {x}_{t-1}^a+ \frac{\epsilon_t}{2} \Sigma_t \left[ \frac{N}{n}H_t^T V_t^{-1}(y_t-H_t {x}_{t-1}^a) + 
               \nabla \log \pi(x_{t-1}^a) \right], \\
  \end{split}
 \end{equation}
 by defining  $\Sigma_t= \frac{n}{N} (I-K_tH_t)$ and by noting $Q_t=\epsilon_t I_p$
 and $R_t=2 V_t$.

 For the Langevinized EnKF, the difference between the equation (7) of the main text and equation 
 (\ref{newtoneq2}) is
 \[
  e_t=(I-K_t H_t) w_t -K_t v_t=w_t-K_t(H_tw_t+v_t),
 \]
 for which the mean $E(e_t)=0$ and the covariance
 \[
 \begin{split} 
\var(e_t) & = \frac{n}{N} Q_t+K_t(\frac{n}{N} H_t Q_t H_t^T+\frac{n}{N} R_t) K_t^T-2 \frac{n}{N} K_t H_t Q_t 
  =\frac{n}{N}\left[ Q_t+ K_t H_t Q_t-2 K_t H_t Q_t\right] \\
 & = \frac{n}{N} (I-K_t H_t) Q_t= \epsilon_t \Sigma_t, \\
\end{split}
 \]
 where the second equality holds due to the symmetry of $Q_t$ and $R_t$ and the
 identity $K_t( H_t Q_t H_t^T+R_t) K_t^T=K_t( H_t Q_t H_t^T+R_t) (H_t Q_t^T H_t^T+R_t^T)^{-1} H_t Q_t^T=K_t H_t Q_t$.
 Then, by (\ref{newtoneq3}), we have 
 \begin{equation} \label{June9eq}
 \begin{split}
  x_t^a & =x_t^f +K_t\left[y_n- H_t x_t^f-v_t\right] \\
   & =x_{t-1}^a+ \frac{\epsilon_t}{2} \Sigma_t \left[ \frac{N}{n}H_t^T V_t^{-1}(y_t-H_t {x}_{t-1}^a) + 
      \nabla \log \pi(x_{t-1}^a) \right] +e_t, \\
\end{split}
 \end{equation}
 where $\frac{N}{n} H_t^T V^{-1} (y_t- H_t x_{t-1}^a)$ represents an unbiased
  estimator for the gradient of
  the log-likelihood function, and $\nabla \log \pi(x_{t-1}^a)$ represents the gradient
  of the log-prior density function. 
  By Corollary \ref{decay} (with $\eta = 0$), we have
  $\lim_{t\to \infty} W_2(\tilde{\pi}_t, \pi_*)=0$.
 For this algorithm, it is easy to see that (\ref{biaseq0}) is satisfied, for which the bias factor $\eta=0$ as $\frac{N}{n}H_t^T V_t^{-1}(y_t-H_t {x}_{t-1}^a) + 
      \nabla \log \pi(x_{t-1}^a)$ forms an unbiased estimator of $\nabla \log \pi(x|y)$ at $x_{t-1}^a$,
      and the variance of the estimation error is upper bounded by a quadratic function of $\|x_{t-1}^a\|$.
\end{proof}

\subsection{Convergence of Algorithm 4} 

Let $\pi_t = \pi(x_{t}|y_{1:t})$ denote the filtering distribution at stage $t$, and let $\tilde \pi_t$ denote 
the marginal distribution of $x_{t,\mK}^{a,i}$ generated by Algorithm 4 at iteration $\mK$ of stage $t$. 
 The following conditions are assumed for the dynamic system (1) of the main text at 
 each stage $t$:
  
\begin{itemize}
\item[(A.1)] $\pi_t$ is $s_t$-strongly log-concave:
\begin{equation} \label{Aeq3}
f(x_t)-f(x_t^*)-\nabla f(x_t^*)^T (x_t-x_t^*) \geq \frac{s_t}{2} \|x_t-x_t^*\|_2^2,  \quad \forall x_t,x_t^* \in \mR_t^p,
\end{equation}
where $f(x_t)= -\log \pi(x_t|y_{1:t})=-\log \pi_t$, and $s_t$ is a positive number satisfying 
$s_t \geq c N_t$ for some constant $c>0$. 

\item[(A.2)] $\log(\pi_t)$ is $S_t$-gradient Lipschitz continuous:
 \begin{equation} \label{Aeq4}
 \| \nabla f(x_t)-\nabla f(x_t^*) \|_2 \leq S_t \|x_t-x_t\|_2, \quad \forall x_t,x_t^* \in \mR_t^p.  
\end{equation}
where $S_t$ is a positive number satisfying $S_t\leq CN_t$ for some constant $C>0$. Note that we must have $s_t\leq S_t$.

\item[(A.3)] Let $\Sigma_{t,k}=\frac{n}{N}(I-K_{t,k}H_{t,k})$, and assume that
$\lambda_{t,l}\leq\inf_k\lambda_{\min}(\Sigma_{t,k})\leq \sup_k\lambda_{\max}(\Sigma_{t,k})\leq \lambda_{t,u}$
for some $\lambda_{t,l}$ and $\lambda_{t,u}$, where
$\lambda_{\max}(\cdot)$ and $\lambda_{\min}(\cdot)$ denote the largest and smallest eigenvalues, respectively.

\item[(A.4)] 
The stochastic error induced by the sub-sampling procedure has a bounded variance, i.e.,  $\forall x\in \mR_t^p$,
\[
E[\|(N/n)H_{t,k}^TV_t^{-1}(y_{t,k}-H_{t,k}x) - H_t^T\Gamma_t^{-1}(y_t-H_tx) \|^2]\leq \sigma_{t,s}^2(p+\|x\|^2),
\]
for some constant $\sigma_{t,s}^2>0$, where the expectation is with respect to random sub-sampling.


\item[(A.5)] The state propagator $g(x_t)$ is $l$-Lipschitz and bounded by $M_g$ (i.e., $\sup_x\|g(x)\|\leq M_g$), and $\lambda_{\min}(U_t) \geq \lambda_{t,s}>0$ for some positive constant $\lambda_{t,s}$. 

\item[(A.6)] 
There exist some constant $M$ such that $W_2(\nu_{t+1}, \pi_{t+1})\leq M$ for all $t\geq 0$,
where 
$\nu_{t+1}(x_{t+1})=\int \pi(x_{t+1}|x_{t})\pi_t(x_t)dx_t$ is the ideal stage initial distribution of $x_{t+1,0}^{a,i}$ for $t\geq1$, and $\nu_1$ is the initial distribution used at stage 1. 
Similarly, we  define $\tilde\nu_{t+1}(x_{t+1})=\int \pi(x_{t+1}|x_{t})\tilde\pi_t(x_t)dx_t$ to be the   practical stage initial distribution of $x_{t+1,0}^{a,i}$ for $t\geq 1$.



\end{itemize}



\begin{remark} \label{remarkApp1}
Log-concavity and strong log-concavity are preserved by products and marginalization \citep{SaumardWellner2014}.
If the prior density $\pi(x_1)$ is log-concave, the state transition density  $\pi(x_{t}|x_{t-1})$ is log-concave with respect to both $x_{t}$ and $x_{t-1}$ for each
stage $t$, and the emission density $\pi(y_t|x_t)$ is $\lambda_t$-strongly-log-concave 
for each stage $t$ where 
$\lambda_t$ is the smallest eigenvalue of $H_t^T\Gamma_t^{-1}H_t$, then condition (A.1) holds with $s_t=\lambda_t$.
Furthermore, by Brascamp-Lieb inequality \citep{brascamp2002extensions}, we must have that $\pi_t$ has finite variance, that is
\begin{equation}\label{postvar}
p\sigma_{t,v}^2:=E_{\pi_t} \|X-E(X)\|^2 \leq p/s_t.
\end{equation}
Strongly log-concave conditions are commonly used in the theoretical study of Langevin Monte Carlo, see e.g. \cite{DalalyanK2017} and \cite{cheng2018convergence}. These conditions potentially can be relaxed following the work of \cite{durmus2017nonasymptotic}.
\end{remark}

\begin{remark}
Assumption (A.6) says the ideal initialization distribution in each stage is not too bad, and  essentially, it actually requires that the data are coherent to the state space model (1) given in the main text, in the sense that the predictive distribution based on $y_{1:t}$ ($\pi(x_{t+1}|y_{1:t})$) and estimated parameter based on $y_{t+1}$ ($\hat x_{t+1}=(H_{t+1}^TH_{t+1})^{-1}H_{t+1}^T y_{t+1}$) are close.
\end{remark}



\begin{lemma}\label{lem0}
Let $\mu$ and $\nu$ be two distribution laws on $\mathbb{R}^p$, and let $f$ be an $L$-Lipschitz continuous function,
then
\[
\bigg\|\int f(x)d\mu(x)-\int f(x)d\nu(x)\bigg\|\leq LW_2(\mu,\nu).
\]
\end{lemma}
\begin{proof}
By definition of Wasserstein distance, there exist random variables 
$X_1$ and $X_2$, whose marginal distributions follow $\mu$ and $\nu$ respectively, such that $\|X_1-X_2\|_{L_2}=(E\|X_1-X_2\|_2^2)^{1/2}=W_2(\mu,\nu)$.
\[
\begin{split}
    &\bigg\|\int f(x)d\mu(x)-\int f(x)d\nu(x)\bigg\| = \|Ef(X_1)-Ef(X_2)\| \leq E\|f(X_1)-f(X_2)\|\\
    \leq& EL\|X_1-X_2\|_2 = LE\sqrt{\|X_1-X_2\|_2^2}\leq L\sqrt{E\|X_1-X_2\|_2^2}=LW_2(\mu,\nu).
\end{split}
\]
\end{proof}

\begin{lemma}\label{lem6}
Under the assumption (A.5), we have
\[
W_2(\tilde \nu_{t+1}, \nu_{t+1})\leq lW_2(\pi_t,\tilde\pi_t)
\]
\end{lemma}
\begin{proof}
By definition of Wasserstein distance, there exist random variables 
$X_1$ and $X_2$, whose marginal distributions are $\pi_t$ and $\tilde \pi_t$ respectively, and $E(\|X_1-X_2\|_2^2)=W_2^2(\tilde\pi_t,\pi_t)$. Define $Y_1=g(X_1)+u$, and $Y_2=g(x_2)+u$, where $u \sim N(0,U_{t+1})$ such that the marginal distributions of $Y_1$ and $Y_2$ are $\nu_{t+1}$ and $\tilde\nu_{t+1}$ respectively. Then,
\[
\begin{split}
    &W_2^2(\tilde \nu_{t+1}, \nu_{t+1}) \leq
    E\|Y_1-Y_2\|^2_2 = E\|g(X_1)-g(X_2)\|_2^2
    \leq El^2\|X_1-X_2\|_2^2 = l^2W_2^2(\pi_t,\tilde\pi_t).
\end{split}
\]
\end{proof}

\begin{lemma}\label{lem2}
If $f$ is an $L$-Lipschitz continuous function, then $E\|f(X)-E(f(X))\|^2_2\leq L^2E\|X-E(X)\|^2_2$.
\end{lemma}
\begin{proof}
Let $X_1$ and $X_2$ be two independent copies of $X$. Then 
\[
\begin{split}
    &E\|f(X)-E(f(X))\|^2_2=(1/2)E\|f(X_1)-f(X_2)\|^2_2\leq (1/2)E(L\|X_1-X_2\|_2)^2\\
    \leq& L^2(1/2)E(\|X_1-X_2\|_2)^2=L^2E\|X-E(X)\|_2^2.
\end{split}
\]
\end{proof}

\begin{lemma}\label{lem3}
Let $X\sim\mu$ and $Y\sim\nu$, then \[E\|Y-E(Y)\|_2^2\leq E\|X-E(X)\|_2^2+W_2^2(\mu,\nu)+2W_2(\mu,\nu)\sqrt{E\|X-E(X)\|_2^2}.\]
\end{lemma}
\begin{proof}
By definition of Wasserstein metric, we can assume that
$X$ and $Y$ satisfy $\|X-Y\|_{L_2}=(E\|X-Y\|_2^2)^{1/2}=W_2(\mu,\nu)$.
Without loss of generality, we also assume that $EX$, the mean of measure $\mu$, is 0. Then
\[
\begin{split}
    &[E\|Y-E(Y)\|_2^2-E\|X\|_2^2]-W_2^2(\mu,\nu)\\
    =&EY^TY-(EY)^T(EY)-EX^TX-EX^TX-EY^TY+2EX^TY\\
    =&2EX^TY-2EX^TX-(EY)^T(EY)\leq 2EX^T(Y-X)\leq 2E\|X\|_2\|Y-X\|_2\\
    \leq&2\sqrt{E\|X\|_2^2E\|Y-X\|_2^2}=2W_2(\mu,\nu)\sqrt{E\|X\|_2^2}.
\end{split}
\]
\end{proof}

 
\begin{lemma}\label{w20}
Let $X\sim\mu$ and $Y\sim\nu$, then 
\[
 E\|Y\|^2\leq E\|X\|^2+W_2^2(\mu,\nu)+2W_2(\mu,\nu)\sqrt{E\|X\|^2},
 \]
 where $W_2(\cdot,\cdot)$ denotes the second order Wasserstein distance. 
\end{lemma}
\begin{proof}
By definition of Wasserstein metric, W.O.L.G, we can assume that
$X$ and $Y$ satisfy $E\|X-Y\|^2=W_2^2(\mu,\nu)$. Then
\[
\begin{split}
    &[E\|Y\|^2-E\|X\|^2]-W_2^2(\mu,\nu)\\
    =&EY^TY-EX^TX-EX^TX-EY^TY+2EX^TY\\
    =&2EX^TY-2EX^TX = 2EX^T(Y-X)\leq 2E\|X\| \|Y-X\|\\
    \leq&2\sqrt{E\|X\|^2E\|Y-X\|^2}=2W_2(\mu,\nu)\sqrt{E\|X\|^2}.
\end{split}
\]
\end{proof}

 Lemma \ref{lemS2} is a generalization of Theorem 4 of \cite{DalalyanK2017}, as well as a generalization of 
 Lemma S2 of \cite{SongLiang2020}. 
 
\begin{lemma} \label{lemS2} 
 Let $x_{k}$ and $x_{k+1}$ be two random vectors in $\mR^p$ satisfying
\[
 x_{k+1} = x_{k}-\epsilon \Sigma [\nabla f(x_{k})+\zeta_{k}]+\sqrt{2\epsilon}e_{k+1},
\]
where $e_{k+1} \sim N(0,\Sigma)$, and $\zeta_{k}$ denotes the random error of the gradient estimate 
 which can depend on $x_{k}$.  
Let $\pi_k$ be the distribution of $x_{k}$, and let $\pi_*\propto \exp\{-f\}$ be the target distribution.
Suppose that $\zeta_{k}$ satisfies
 \begin{equation} \label{biaseq0}
 \| E(\zeta_{k}|x_{k}) \|^2  \leq \eta^2p, 
 \quad E[ \| \zeta_{k}- E(\zeta_{k}| x_{k}) \|^2 ] \leq \sigma^2_1 p+ \sigma^2_2\|x_{k}\|^2,
 \end{equation}
 for some constants $\eta$ and $\sigma$, and $\zeta_{k}$'s 
 are independent of $e_{k+1}$'s.
 If the function $f$ is $s$-strongly convex and $S$-gradient-Lipschitz, 
 $\lambda_{\min}(\Sigma)=\lambda_l$, $\lambda_{\max}(\Sigma)=\lambda_u$, 
 and the learning rate $\epsilon \leq 2/(s \lambda_l+S \lambda_u)$, then 
 \begin{equation} \label{coneq}
 \begin{split}
 W_2^2(\pi_{k+1},\pi_*)    \leq&\left[(1-\lambda_l s\epsilon+\sqrt 2\sigma_2 \lambda_u \epsilon)W_2(\pi_{k},\pi_*)+1.65S(\lambda_u^3\epsilon^3p)^{1/2}+\epsilon\eta\lambda_u\sqrt p\right]^2\\
   &+\epsilon^2\sigma^2_1\lambda_u^2p+2\epsilon^2\sigma^2_2\lambda_u^2pV.
   \end{split}
 \end{equation}
 where $V=\int \|x\|^2 \pi_*(x)dx$. 
\end{lemma}

\begin{proof}
First of all, the updating iteration can be rewritten as:
\[
\tilde x_{k+1} =\tilde   x_{k}-\epsilon [ \nabla\tilde f(\tilde x_{k})+\tilde\zeta_{k}]+\sqrt{2\epsilon}\tilde e_{k+1},
\]
where $\tilde f(x) = f(\Sigma^{1/2}x)$, $\tilde x_k =\Sigma^{-1/2}  x_{k} $, $\tilde\zeta_{k}=\Sigma^{1/2}\zeta_{k}$ and $\tilde e_{k+1} \sim N(0,I)$.

Let $\tilde \pi_*$ denote the distribution $\tilde\pi_*\propto \exp\{-\tilde f\}$.
It is easy to see that the distribution $\tilde\pi_*$ is $s\lambda_l$-strongly log-concave and $S\lambda_u$-gradient-Lipschitz. In addition, $\tilde\zeta_k$ satisfies
\begin{equation}\label{eql0}
\begin{split}
\| E(\tilde\zeta_{k}|\tilde x_{k}) \|^2 &= \|\Sigma^{1/2} E(\zeta_{k}|x_{k}) \|^2  \leq \lambda_u \eta^2 p\\
E[ \| \tilde\zeta_{k}- E(\tilde\zeta_{k}|\tilde x_{k}) \|^2 ]
&=E[ \| \Sigma^{1/2}\zeta_{k}- E(\Sigma^{1/2}\zeta_{k}|x_{k}) \|^2 ]
\leq \lambda_u\sigma^2_1 p+ \lambda_u\sigma^2_2\|\Sigma^{1/2}\tilde x_{k}\|^2,
\end{split}
\end{equation}
 
Let $L_t$ be the stochastic process defined by $dL_t=-(\Sigma^{1/2} \nabla f(\Sigma^{1/2}L_t))dt+\sqrt 2dW_t$ with initialization $L_0\sim\tilde\pi_*$ (hence $L_t\sim\tilde\pi_*$). Define $\Delta_2=L_\epsilon-\tilde x_{t+1}$ and 
$\Delta_1=L_0-\tilde x_{t}$. Then, by the same arguments used in the proof of Proposition 2 in \cite{DalalyanK2017}.
we have 
\begin{equation}\label{eql1}
\begin{split}
&\|\Sigma^{1/2}\Delta_2\|_{L_2}^2\\
\leq& \{\|\Sigma^{1/2}\Delta_1-\epsilon \Sigma^{1/2}U\|_{L_2} +\|\Sigma^{1/2}W\|_{L_2}+\epsilon\|\Sigma^{1/2}E(\tilde\zeta_{k}|\tilde x_{k})\|_{L_2}\}^2+
\epsilon^2\|\Sigma^{1/2}(\tilde\zeta_{k}- E(\tilde\zeta_{k}|\tilde x_{k}) ) \|_{L_2}^2,
\end{split}
\end{equation}
where $W=\int_0^{\epsilon}(\nabla\tilde f(L_t)-\nabla\tilde f(L_0))dt$ and 
$U=\nabla\tilde f(\tilde x_k+\Delta_1) - \nabla\tilde f(\tilde x_k)$.

By Lemma 4 of \cite{DalalyanK2017},
$\|W\|_{L_2}\leq 0.5\sqrt{\epsilon^4S^3\lambda_u^3p}+(2/3)\sqrt{2\epsilon^3p}S\lambda_u\leq 1.65S\lambda_u(\epsilon^3p)^{1/2}$.
By similar argument of Lemma 2 of \cite{DalalyanK2017}, we can show that
$\|\Sigma^{1/2}\Delta_1-\epsilon \Sigma^{1/2}U\|_{2} \leq \rho \|\Sigma^{1/2}\Delta_1\|_2$, where $\rho = \max(1-s\lambda_l\epsilon,S\lambda_u\epsilon-1)=1-s\lambda_l\epsilon$. Combining with (\ref{eql0}) and (\ref{eql1}), we have that
\begin{equation*}
\begin{split}
\|\Sigma^{1/2}\Delta_2\|_{L_2}^2
\leq \{\rho\|\Sigma^{1/2}\Delta_1\|_{L_2} +1.65S(\lambda_u^3\epsilon^3p)^{1/2}+\epsilon\lambda_u\eta\sqrt p\}^2+\epsilon^2\lambda_u^2\sigma^2_1 p+\epsilon^2\lambda_u^2\sigma^2_2E\| x_{k}\|^2,
\end{split}
\end{equation*}
which further implies 
\begin{equation*}
\begin{split}
W_2^2(\pi_{k+1},\pi_*)
\leq \{(1-s\lambda_l\epsilon)W_2^2(\pi_{k+1},\pi_*) +1.65S(\lambda_u^3\epsilon^3p)^{1/2}+\epsilon\lambda_u\eta\sqrt p\}^2+\epsilon^2\lambda_u^2\sigma^2_1 p+\epsilon^2\lambda_u^2\sigma^2_2E\| x_{k}\|^2.
\end{split}
\end{equation*}

By Lemma \ref{w20}$, E\|x_{k}\|^2\leq (W_2(\pi_{k},\pi_*)+\sqrt V)^2$,
we can derive that
\[
\begin{split}
W_2^2(\pi_{k+1},\pi_*)
 \leq&\left[(1-s\lambda_l\epsilon)W_2(\pi_{k},\pi_*)+1.65S(\lambda_u^3\epsilon^3p)^{1/2}+\epsilon\eta\lambda_u\sqrt p\right]^2\\
   &+\epsilon^2\sigma^2_1\lambda_u^2p+\epsilon^2\sigma^2_2\lambda_u^2(W_2(\pi_{k},\pi_*)+\sqrt V)^2\\
   \leq &\left[(1-s\lambda_l\epsilon)W_2(\pi_{k},\pi_*)+1.65S(\lambda_u^3\epsilon^3p)^{1/2}+\epsilon\eta\lambda_u\sqrt p\right]^2\\
   &+\epsilon^2\sigma^2_1\lambda_u^2p+2\epsilon^2\sigma^2_2\lambda_u^2pV+2\epsilon^2\sigma^2_2\lambda_u^2W_2^2(\pi_{k},\pi_*)\\
   \leq &\left[(1-s\lambda_l\epsilon+\sqrt 2\sigma_2\epsilon\lambda_u)W_2(\pi_{k},\pi_*)+1.65S(\lambda_u^3\epsilon^3p)^{1/2}+\epsilon\eta\lambda_u\sqrt p\right]^2\\
  & +\epsilon^2\sigma^2_1\lambda_u^2p
    +2\epsilon^2\sigma^2_2\lambda_u^2pV,
\end{split}
\]
which concludes the proof.
\end{proof}
 

\begin{remark} \label{rem6}
Condition (\ref{biaseq0}) in Lemma \ref{lemS2} can be further relaxed. For example if we assume bias of gradient estimation is not uniformly bound, increasing with $\|x_{k}\|$ in the form $\| E(\zeta_{k}|x_{k}) \|^2  \leq \eta^2 (\sqrt p+\|x_{k}\|)^2$, then by the same technique we can prove that
 \begin{equation*}
 \begin{split}
 W_2^2(\pi_{k+1},\pi_*)    \leq &\left[(1-s\lambda_l\epsilon+\eta\lambda_u\epsilon+\sqrt 2\sigma_2\epsilon\lambda_u)W_2(\pi_{k},\pi_*)+1.65S(\lambda_u^3\epsilon^3p)^{1/2}+\epsilon\eta\lambda_u(\sqrt p+\sqrt V)\right]^2\\
   &+\epsilon^2\sigma^2_1\lambda_u^2p+2\epsilon^2\sigma^2_2\lambda_u^2pV.
\end{split}\end{equation*}
In the proof of Lemma \ref{lemS2}, $p$ is treated as a constant, so it is  trivial to extend the proof to the case that $\| E(\zeta_k|x_{k}) \|^2$ and $E[ \| \zeta_k- E(\zeta_k|x_{k}) \|^2 ]$ increase in a polynomial of $p$.
In that case, the statements in Remark 7, Corollary \ref{decay}, and Theorem \ref{Them2} can be updated accordingly. 
\end{remark}
 
\begin{remark} \label{remarkApp2} 
Consider a Langevin Monte Carlo algorithm with inaccurate gradients, varying conditioning matrices and a constant learning rate $\epsilon$, i.e.,
\[x_{k+1} = x_{k}-\epsilon\Sigma_{k}[\nabla f(x_{k})+\zeta_{k}]+\sqrt{2\epsilon}\xi_{k+1};\quad \xi_{k+1}\sim N(0,\Sigma_{k}). \]
If  $\Sigma_{k}$ is positive definite, $\lambda_l\leq \inf_k\lambda_{\min}(\Sigma_k)\leq \sup_k\lambda_{\max}(\Sigma_k)
\leq \lambda_u$ and $\epsilon\leq 2/(s \lambda_l+S \lambda_u)$, then (\ref{coneq}) holds for all iterations. Combining it with Lemma 1 of \cite{DalalyanK2017}, it is easy to obtain that
   \begin{equation} \label{coneq2}
  \begin{split}
  W_2(\pi_k,\pi_*) 
  \leq& (1-s\lambda_l\epsilon+\sqrt 2\sigma_2\epsilon\lambda_u)^k W_2(\pi_0,\pi_*)\\
  &+\frac{\eta\lambda_u\sqrt p}{s\lambda_l-\sqrt 2\sigma_2\lambda_u}  +\frac{1.65S(\lambda_u^3\epsilon p)^{1/2}}{s\lambda_l-\sqrt 2\sigma_2\lambda_u} + \frac{\sqrt{\epsilon}\lambda_u(\sigma^2_1p+2\sigma^2_2V)}{ 1.65S(\lambda_u p)^{1/2}},
  \end{split}
 \end{equation}
 where the first term in the RHS of (\ref{coneq2}) converges to 0 if $s\lambda_l>\sqrt{2}\sigma_2\lambda_u$.
\end{remark} 

The next corollary provides a decaying-learning-rate version of the convergence result (\ref{coneq2}).
\begin{corollary}\label{decay}
Consider a Langevin Monte Carlo algorithm
\[x_{k+1} = x_{k}-\epsilon_{k+1}\Sigma_{k}[\nabla f(x_{k})+\zeta_{k}]+\sqrt{2\epsilon_{k+1}}\xi_{k};\quad \xi_{k}\sim N(0,\Sigma_k),\]
where $\zeta_k$ satisfies (\ref{biaseq0}), $\Sigma_{k}$ is positive definite, $\lambda_l\leq \inf_k\lambda_{\min}(\Sigma_k)\leq \sup_k\lambda_{\max}(\Sigma_k)
\leq \lambda_u$ and the learning rate  
$\epsilon_k =  2/[(s \lambda_l+S \lambda_u)k^{\varpi}]$ for some $\varpi\in(0,1)$. If $s\lambda_l>\sqrt{2}\sigma_2\lambda_u$, then we have
\begin{equation}
   \begin{split}
 \limsup_{k\rightarrow\infty} W_2(\pi_k,\pi_*)\leq \frac{\varphi}{1-\varphi}\frac{\eta\lambda_u\sqrt p}{s\lambda_l-\sqrt 2\sigma_2\lambda_u},\mbox{ for some constant }\varphi\in(0,1).
  \end{split}
\end{equation}
\end{corollary}
\begin{proof}
The proof of this corollary closely follows the proof of Theorem 2(i) in \cite{SongLiang2020}.
Let $K_0=0$, and $K_i$ ($i>0$) be the smallest integer such that $K_i^{-\varpi}\leq (1+i)^{-\chi}$, where $\chi=\varpi/(1-\varpi)$. Thus,  asymptotically, we have $K_{i+1}-K_{i}\approx (\chi/\varpi)K_{i+1}^\varpi$.

In the spirit of (\ref{coneq2}), we have 
\[
  \begin{split}
W_2(\pi_{K_{i+1}},\pi_*) 
  \leq& (1-(s\lambda_l-\sqrt 2\sigma_2\lambda_u)\epsilon_{K_{i+1}})^{K_{i+1}-K_i} W_2(\pi_{K_i},\pi_*)\\
  &+\frac{\eta\lambda_u\sqrt p}{s\lambda_l-\sqrt 2\sigma_2\lambda_u}  +\left[\frac{1.65S(\lambda_u^3 p)^{1/2}}{s\lambda_l-\sqrt 2\sigma_2\lambda_u} + \frac{\lambda_u(\sigma^2_1p+2\sigma^2_2V)}{ 1.65S(\lambda_u p)^{1/2}}\right]\sqrt{\epsilon_{K_{i}}}.
  \end{split}
  \]
Note that due to the fact that $K_{i+1}-K_{i}\approx  (\chi/\varpi)\epsilon_{K_{i+1}}^{-1}$, we have
\[\lim_{i\rightarrow\infty}[1-(s\lambda_l-\sqrt 2\sigma_2\lambda_u)\epsilon_{K_{i+1}}]^{K_{i+1}-K_i} = \exp\left\{-\frac{2(s\lambda_l-\sqrt 2\sigma_2\lambda_u)\chi}{(s \lambda_l+S \lambda_u)\varpi}\right\}<1.\] 
Therefore, there exists some positive constant $\varphi\in(\exp(-\frac{2(s\lambda_l-\sqrt 2\sigma_2\lambda_u)\chi}{(s \lambda_l+S \lambda_u)\varpi}),1)$, such that for all $i\geq 1$, $(1-(s\lambda_l-\sqrt 2\sigma_2\lambda_u)\epsilon_{K_{i+1}})^{K_{i+1}-K_i}\leq \varphi$, i.e.,
\[
  \begin{split}
W_2(\pi_{K_{i+1}},\pi_*) 
  \leq \varphi W_2(\pi_{K_i},\pi_*)
  +\frac{\eta\lambda_u\sqrt p}{s\lambda_l-\sqrt 2\sigma_2\lambda_u}  +\left[\frac{1.65S(\lambda_u^3 p)^{1/2}}{s\lambda_l-\sqrt 2\sigma_2\lambda_u} + \frac{\lambda_u(\sigma^2_1p+2\sigma^2_2V)}{ 1.65S(\lambda_u p)^{1/2}}\right]\sqrt{\epsilon_{K_{i}}}.
  \end{split}
  \]
The above recursive inequality implies that 
\begin{equation}\label{w2decaynonasy}
  \begin{split}
&W_2(\pi_{K_{T}},\pi_*) 
  \leq \varphi^T W_2(\pi_{K_0}=\pi_0,\pi_*)
  +(\sum_{t=1}^{T}\varphi^{t-1})\frac{\eta\lambda_u\sqrt p}{s\lambda_l-\sqrt 2\sigma_2\lambda_u} \\
  &\quad+(\sum_{t=1}^{T} \varphi^{t-1}K_{T-t}^{-\varpi/2}) \sqrt{\frac{2}{s \lambda_l+S \lambda_u}}\left[\frac{1.65S(\lambda_u^3 p)^{1/2}}{s\lambda_l-\sqrt 2\sigma_2\lambda_u} + \frac{\lambda_u(\sigma^2_1p+2\sigma^2_2V)}{ 1.65S(\lambda_u p)^{1/2}}\right].
  \end{split}
  \end{equation}
As $T\rightarrow\infty$, $\sum_{t=1}^{T} \varphi^{t-1}K_{T-t}^{-\varpi/2}\rightarrow 0$, hence we have that $W_2(\pi_{K_{T+1}},\pi_*) \rightarrow \frac{\varphi}{1-\varphi}\frac{\eta\lambda_u\sqrt p}{s\lambda_l-\sqrt 2\sigma_2\lambda_u}$.
\end{proof}

\begin{remark}
For  technical simplicity, we require $\varpi<1$ for the decay of the learning rate ($\epsilon_t \propto t^{-\varpi}$) in the above corollary. We conjecture that the corollary still holds under the choice $\epsilon_t\propto t^{-1}$, i.e., $\varpi=1$. However, more subtle technical tools are necessary to rigorously characterize the convergence rate under $\epsilon_t\propto t^{-1}$, see \cite{Teh2016SGLD}.
\end{remark}

\begin{theorem} \label{Them2} 
Assume that the conditions (A.1)-(A.6) hold, the ensemble size is sufficiently large, $c_1N_t\leq s_t\leq S_t\leq c_2N_t$, $c_3(n_t/N_t)\leq \lambda_{t,l}\leq \lambda_{t,u}\leq c_4(n_t/N_t)$, and $\sigma^2_{t,s}\leq c_5(N_t/n_t)$.
Let $\sup_{t} 1/(N_tn_t)\leq c_6$ for some sufficiently small constant $c_6$
such that 
\[
\varphi_0=2\frac{c_1c_3-2\sqrt{c_5c_6}c_4}{c_1c_3+c_2c_4}\in(0,1).
\]
Denote $V_t = \int \|x\|^2d\pi_{t}$ and it is assumed to satisfy the constraint $V_t\leq c_7p$ for some constant $c_7>0$.
We choose the learning rate $\epsilon_{t,k} =2/[(s_{t}\lambda_{t,l}+S_t\lambda_{t,u})k^{\varpi}]$ ($k=1,\dots,\mK$, $t=1,\dots,\infty$) for some $\varpi\in(0,1)$. If all $N_t$'s are bounded away from 0, then  $\limsup_{t\rightarrow\infty} W_2(\tilde\pi_{t+1},\pi_{t+1})= O(1/\lim\inf_t N_t)$, as long as $\mK$ is sufficiently large (refer to (\ref{formalres}) for formal conditions).


\end{theorem}
\begin{proof}
Define $K_i$ as in the proof of Corollary \ref{decay}, and we let $\mathcal{K}=K_{\varkappa}(\asymp \varkappa^{\chi/\varpi})$ for some $\varkappa$ to be specified later, where $\chi=\varpi/(1-\varpi)$.

At stage $t=1$, Algorithm 4 performs exactly as Algorithm 2; that is, it is a Langevin Monte Carlo algorithm with a varying conditioning matrix as discussed in Section 3. By Corollary \ref{decay} with no bias $\eta=0$, we obtain that there exists some $\varphi\in(\exp\{-\varphi_0\chi/\varpi\},1)$, such that
\begin{equation}\label{w2diff1}
\begin{split}
 W_2(\tilde\pi_1,\pi_1)\leq&\varphi^\varkappa W_2(\nu_{1}, \pi_1)+(\sum_{j=1}^{\varkappa} \varphi^{j-1}K_{\varkappa-j}^{-\frac{\varpi}{2}}) \sqrt{\frac{2}{s_{1} \lambda_{1,l}+S_{1} \lambda_{1,u}}}\\
 &\times   \left[ \frac{1.65S_1\lambda_{1,u}\sqrt{p\lambda_{1,u}}}{s_1\lambda_{1,l}-\sqrt 2\sigma_{1,s}\lambda_{1,u}}+
 \frac{\sigma^2_{1,s}\lambda_{1,u}(p+2V_1)}{1.65 S_1\sqrt{\lambda_{1,u}p}
 }\right]\\
 \leq &\varphi^\varkappa M+C_0(\sum_{j=1}^{\varkappa} \varphi^{j-1}K_{\varkappa-j}^{-\frac{\varpi}{2}})\sqrt{\frac{p}{N_1}},\\
\end{split}
\end{equation}
for some constant $C_0$ (which only depends on the $c_i$'s values in the statement of this theorem), where $\nu_{1}$ denotes the prior distribution of $x_1$. 

Now, we study the relationship between $W_2(\tilde\pi_{t+1},\pi_{t+1})$ and $W_2(\tilde\pi_t,\pi_t)$ for $t\geq2$. As discussed in Section 3.1, at stage $t+1$, the algorithm can be rewritten as 
\begin{equation}\label{sg}
\begin{split}
x_{t+1,k+1}^{a,i}& =x_{t+1,k}^{a,i}
+\epsilon_{t+1,k+1}\Sigma_{t+1,k+1}\left[
\frac{N}{n}H_{t+1,k}^TV_{t+1}^{-1}(y_{t+1,k}-H_{t+1,k}x_{t+1,k}^{a,i})
+\nabla\log \pi(x_{t+1,k}^{a,i}|\tilde{x}_{t,k}^i)
\right]+e_{t+1}\\
& \stackrel{\Delta}{=} x_{t+1,k}^{a,i}+ \epsilon_{t+1,k+1}\Sigma_{t+1,k+1} \left[(I)+(II)\right] + e_{t+1}, \\ 
\end{split}
\end{equation}
where $e_{t+1}\sim N(0,2\epsilon_{t+1,k+1}\Sigma_{t+1,k+1})$, and $\tilde{x}_{t,k}^i$ denotes a sample drawn from 
 the set $\mX_t$ according to an importance weight proportional 
 to $\pi(x_{t+1,k}^{a,i}|\tilde{x}_{t,k}^i)$.

To apply (\ref{w2decaynonasy}), it is necessary to study the bias and variance of the gradient estimate used in (\ref{sg}). Note that the term (I) is unbiased due to the property of simple random sampling. 
To study the bias of term (II), we define 
$
\pi(z|x_{t+1,k}^{a,i}, y_{1:t}) \propto \pi(x_{t+1,k}^{a,i}|z) \pi_t(z|y_{1:t})$;
that is, $\pi(z|x_{t+1}^{a,i},y_{1:t})$ can be viewed as a posterior density obtained with the prior density $\pi_t(z|y_{1:t})$ and the likelihood $\pi(x_{t+1}^{a,i}|z)$. Similarly, we define 
$
\tilde{\pi}(z|x_{t+1,k}^{a,i}, y_{1:t}) \propto \pi(x_{t+1,k}^{a,i}|z) \tilde{\pi}_t(z|y_{1:t})$.
Then, by equation (18) of the main text,
the bias of term (II) can be bounded by 
\begin{equation}\label{bias1}
\begin{split}
   &\left \|\int \nabla\log \pi(x_{t+1,k}^{a,i}|z) 
  [d\tilde\pi(z|x_{t+1,k}^{a,i},y_{1:t})-d\pi(z|x_{t+1,k}^{a,i}, y_{1:t})]\right \| \\
  =&\left \| \int U_t^{-1}[x_{t+1,k}^{a,i}-g(z)]
  [d\tilde\pi(z|x_{t+1,k}^{a,i},y_{1:t})-d\pi(z|x_{t+1,k}^{a,i}, y_{1:t})]\right \| \\
  =& \left\| -\int U_t^{-1}g(z)
  [d\tilde\pi(z|x_{t+1,k}^{a,i},y_{1:t})-d\pi(z|x_{t+1,k}^{a,i}, y_{1:t})] \right\|
  \leq  2M_g/\lambda_{t,s},
\end{split}
\end{equation}
which holds for any $\tilde{\pi}(\cdot|\cdot)$. 

By Condition (A.4), the variance of term (I) is bounded by 
$\sigma_{t+1,s}^2(p+\|x\|^2)$.
The variance of term (II) is upper bounded by
\begin{equation} \label{bias2}
\begin{split}
&E\bigg\|\nabla\log \pi(x_{t+1,k}^{a,i}|\tilde{x}_{t,k}^i)
 -E\left(\nabla\log \pi(x_{t+1,k}^{a,i}|\tilde{x}_{t,k}^i) \right)\bigg\|^2
\leq (l/\lambda_{t,s})^2E\|\tilde{x}_{t,k}^i-E (\tilde{x}_{t,k}^i)\|^2 \quad\mbox{(by Lemma \ref{lem2})}\\
\leq& (l/\lambda_{t,s})^2\left[W_2(\pi_t,\tilde\pi_t)^2+p\sigma_{t,v}^2+2 W_2(\pi_t,\tilde\pi_t)\sqrt{p\sigma_{t,v}^2} \right]
\quad\mbox{ by Lemma \ref{lem3} and (\ref{postvar})}.
\end{split}
\end{equation}
Combining the above results together, the variance of the estimated gradient is upper bounded by
\begin{equation} \label{newaddeq}
2\sigma^2_{t+1,s}p+2(l/\lambda_{t,s})^2(W_2(\pi_t,\tilde\pi_t)+\sqrt{p\sigma_{t,v}^2})^2  +2\sigma^2_{t+1,s}\|x_{t+1,k}^{a,i}\|^2
:=\sigma_p^2+2\sigma^2_{t+1,s}\|x_{t+1,k}^{a,i}\|^2,
\end{equation}
which implies that a smaller value of  $W_2(\pi_t,\tilde{\pi}_t)$ will help to reduce the variance of the stochastic gradient at iteration $t+1$ 
as well as  $W_2(\tilde{\pi}_{t+1},\pi_{t+1})$ as implied by (\ref{w2diff2}).

Recall that $\tilde{\nu}_{t+1}$ denotes the practical state initial distribution of $x_{t+1,0}^{a,i}$.  
Applying equation (\ref{w2decaynonasy}) with $\sigma_1^2 = \sigma_p^2/p$, $\sigma_2^2=2\sigma^2_{t+1,s}$ and $\eta = 2M_g/(\sqrt{p}\lambda_{t,s})$, we obtain that there exists some $\varphi\in(\exp\{-\varphi_0\chi/\varpi\},1)$ such that
\begin{equation}\label{w2diff2}
\begin{split}
W_2(\tilde\pi_{t+1},\pi_{t+1}) \leq &\varphi^{\varkappa} W_2(\tilde\nu_{t+1},\pi_{t+1})+\frac{\varphi}{1-\varphi} \frac{2M_g\lambda_{t+1,u}}{\lambda_{t,s}(s_{t+1}\lambda_{t+1,l}-2\sigma_{t+1,s}\lambda_{t+1,u})}\\
+&(\sum_{j=1}^{\varkappa} \varphi^{j-1}K_{\varkappa-j}^{-\frac{\varpi}{2}}) \sqrt{\frac{2}{s_{t+1} \lambda_{t+1,l}+S_{t+1} \lambda_{t+1,u}}}\\
&\times\left[\frac{1.65S_{t+1}\sqrt{ p\lambda_{t+1,u}^3}}{s_{t+1}\lambda_{t+1,l}-2\sigma_{t+1,s}\lambda_{t+1,u}}+ \frac{\sqrt{\lambda_{t+1,u}} (\sigma_p^2+4\sigma_{t+1,s}^2V_t)}{1.65 S_{t+1}\sqrt p}\right].\\
\end{split}
\end{equation}
Note that $\varphi$ exists because $\varphi_0\leq \min_t 2\frac{s_{t+1}\lambda_{t+1,l}-2\sigma_{t+1,s}\lambda_{t+1,u}}{s_{t+1}\lambda_{t+1,l}+S_{t+1}\lambda_{t+1,u}}$, and w.o.l.g., we let the $\varphi$'s in (\ref{w2diff2}) and (\ref{w2diff1}) match.
 By Assumption (A.6) and Lemma \ref{lem6}, $W_2(\tilde\nu_{t+1},\pi_{t+1})$ $\leq M+lW_2(\tilde\pi_{t},\pi_{t})$.
 Further, combining with other conditions
 stated in the theorem, we have
\begin{equation}\label{w2diff3}
\begin{split}
W_2(\tilde\pi_{t+1},\pi_{t+1})
\leq&\varphi^{\varkappa}(M+lW_2(\tilde\pi_{t},\pi_{t}))+C_1\frac{\varphi}{1-\varphi}\frac{2M_g}{N_{t+1}}\\
&+C_2(\sum_{j=1}^{\varkappa} \varphi^{j-1}K_{\varkappa-j}^{-\frac{\varpi}{2}})\sqrt{\frac{p}{N_{t+1}}}+C_3(\sum_{j=1}^{\varkappa} \varphi^{j-1}K_{\varkappa-j}^{-\frac{\varpi}{2}})\sqrt{\frac{1}{N_{t+1}^3p}}W_2^2(\pi_t,\tilde\pi_t),
\end{split}
\end{equation}
for some positive constants $C_1$, $C_2$ and $C_3$ (which only depend on the $c_i$'s values in the statement of this theorem), where the last term follows from the 
definition of $\sigma_p^2$ given in (\ref{newaddeq}). 


Based on the recursive inequalities (\ref{w2diff1}) and (\ref{w2diff3}), one can conclude that $W_2(\tilde\pi_{t+1},\pi_{t+1})$ can converge to any arbitrarily small quantity, as long as that $\varkappa$ and $N_t$ are sufficiently large. Let $\tilde\varepsilon$ be some small positive quantity $\tilde\varepsilon\in(0,1)$, and we assume that sample size $N_t$ is non-decreasing with respect to $t$, $N_t$ and the iteration number $\mathcal{K}=K_{\varkappa}(\asymp \varkappa^{\chi/\varpi})$ satisfy
\begin{equation}\label{formalres}
\begin{split}
    &\lim N_t=N_\infty;\\
    &\varphi^\varkappa<\min\{\tilde\varepsilon/(3M),1/(3l)\} \mbox{ and } \sum_{j=1}^{\varkappa} \varphi^{j-1}K_{\varkappa-j}^{-\frac{\varpi}{2}}\leq \min\left\{\frac{\sqrt{N_1^3p}}{3C_3\overline W},1, \frac{\tilde\varepsilon}{C_2\sqrt{p/N_1}} \right\},
\end{split}
\end{equation}

Let's define $\overline W=\max\{1+2C_1\frac{\varphi}{1-\varphi}\frac{2M_g}{N_{1}}+2C_2\sqrt{\frac{p}{N_{1}}},1+C_0\sqrt p\}$, then it is not difficult to verify that 
(i) $W_2(\tilde \pi_t,\pi_t)\leq \overline W$ for all $t\geq 1$ and (ii) $W_2(\tilde \pi_{t+1},\pi_{t+1})\leq \tilde\varepsilon/3+W_2(\tilde \pi_{t},\pi_{t})/3+W_2^2(\tilde \pi_{t},\pi_{t})/(3\overline W)+C_1\frac{\varphi}{1-\varphi}\frac{2M_g}{N_{t+1}}+\tilde\varepsilon$. 
Combine them with the fact that $\lim_{T}\sum_{t=1}^T(2/3)^{T-t}N_t^{-1}=2\lim_t (1/{N_t})$, this further implies
\[
\limsup_{t\rightarrow\infty} W_2(\tilde\pi_{t+1},\pi_{t+1})\leq 4\tilde\varepsilon+ 2C_1\frac{\varphi}{1-\varphi}M_g{N_\infty^{-1}}.
\]

Equivalently, we claim that given a sufficiently large number of iterations in each stage, then
\[
\limsup_{t\rightarrow\infty} W_2(\tilde\pi_{t+1},\pi_{t+1}) = O\left( \frac{1}{{\liminf_t N_t}}\right).
\]
 \end{proof}

\bibliographystyle{asa}
\bibliography{ref}

\end{document}